\documentclass{amsart}
\usepackage{amssymb}
\usepackage{amsfonts}
\usepackage{amsthm}
\usepackage{mathtools}
\usepackage[authoryear]{natbib}
\usepackage{bbold} % for indicator functions
\usepackage{url}
\usepackage[colorlinks=true, allcolors=blue]{hyperref}
\usepackage{bbm}
\usepackage{parskip}
\usepackage[a4paper,top=3cm,bottom=3cm,left=2.5cm,right=2.5cm,marginparwidth=1.75cm]{geometry}
\usepackage{tikz}
\RequirePackage{graphicx}

\theoremstyle{plain}
\newtheorem{thm}{Theorem}[section]
\newtheorem{lem}[thm]{Lemma}
\newtheorem{prop}{Proposition}
\newtheorem{corollary}{Corollary}

\theoremstyle{remark}
\newtheorem{assumption}{Assumption}
\newtheorem{defi}{Definition}
\newtheorem{example}{Example}

\newcommand{\dd}{\mathrm{d}}
\newcommand{\dx}{\dd x}

\newcommand{\eqdef}{:=}
\newcommand{\ind}{\mathbb{1}}  % indicator function

\newcommand{\E}{\mathbb{E}}
\newcommand{\Var}{\mathrm{Var}}
\newcommand{\Cov}{\mathrm{Cov}}
\renewcommand{\P}{\mathbb{P}}
\newcommand{\Q}{\mathbb{Q}}
\newcommand{\iid}{\overset{i.i.d.}{\sim}}

\newcommand{\bigO}{\mathcal{O}}
\newcommand{\bigOproba}{\bigO_\P}

\newcommand{\cvprob}{\stackrel{\mathbb P}{\rightarrow}}

\usepackage[ruled, vlined]{algorithm2e}

\newenvironment{algo}[1]{
  \begin{center}
      \begin{algorithm}
        \caption{#1}
        \DontPrintSemicolon
      }
      {
      \end{algorithm}
  \end{center}
}

\newcommand{\tnn}{\min(\tau_t^{1, \mrp}, N)}

\newcommand{\norm}[1]{\left\lVert#1\right\rVert}
\newcommand{\pr}[1]{\left( #1 \right)} %% parenthesis ()
\newcommand{\ps}[1]{\left[ #1 \right]} %% parenthesis []
\newcommand{\px}[1]{\left\{ #1 \right\}} %% parenthesis {}
\newcommand{\abs}[1]{\left| #1 \right|} %% absolute val ||
\newcommand{\CE}[2]{\E\ps{\left. {#1} \right \vert {#2}}} %% conditional expectation
\newcommand{\CProb}[2]{\P\pr{\left. {#1} \right \vert {#2}}} %% conditional probability
\newcommand{\CVar}[2]{\Var\pr{\left. {#1} \right \vert {#2}}} %% conditional variance
\newcommand{\inv}[1]{{#1}^{-1}}

\newcommand{\infnorm}[1]{\norm{#1}_\infty}
\newcommand{\given}{\textrm{ } | \textrm{ }}
\newcommand{\fastmatrix}[1]{\begin{bmatrix} #1 \end{bmatrix}}
\newcommand{\CCov}[3]{\operatorname{Cov} \pr{\left. #1 , #2 \right \vert #3}}  % conditional covariance
\newcommand{\floor}[1]{\lfloor #1 \rfloor}
\newcommand{\invp}[1]{\inv{(#1)}}
\newcommand{\notsep}{\ \textbullet\ } % seperator for defining notations
\newcommand{\nfrac}[2]{{#1}/{#2}} % use when one wants to quickly replace \frac{a}{b} by a/b -> only needs to change \frac to \nfrac
\newcommand{\splitcell}[2]{\begin{tabular}{l}
		#1\\#2 
\end{tabular}}
\newcommand{\splitcellbis}[3]{
\begin{tabular}{l}
	#1\\#2\\#3 
\end{tabular}}

\newcommand{\bw}[1]{B_t^{N, \mathrm{#1}}}
\newcommand{\bwm}[1]{\hat B_t^{N, \mathrm{#1}}}

\newcommand{\FTP}[1]{\mathcal{F}_{#1}^{+}}
\newcommand{\ceftmone}[1]{\CE{#1}{\FTP{t-1}}}
\newcommand{\norminf}[1]{\norm{#1}_\infty}
\newcommand{\btn}[1]{B_t^{N, \mathrm{#1}}} % duplicate for historical reasons...
\newcommand{\btnhat}[1]{\hat B_t^{N, \mathrm{#1}}}
\newcommand{\epss}{\varepsilon_{\mathrm{S}}}

\newcommand{\FTM}{\mathcal F_T^-}
\newcommand{\ccovftm}[2]{\CCov{#1}{#2}{\FTM}}

\newcommand{\qtnffbs}{\Q_T^{N, \mathrm{FFBS}}}
\newcommand{\mrp}{\mathrm{PaRIS}}
\newcommand{\mrffbs}{\mathrm{FFBS}}
\newcommand{\Proj}{\Pi}
\newcommand{\Projeg}{\Proj^{(\mathcal X, \mathcal Y)}_{\mathcal X}}
\newcommand{\nmt}{(\ell_t^N)^{-1} N^{-1}} % normalising term
\newcommand{\qtnb}[1]{\bar \Q_{#1}^N}
\newcommand{\ddxz}[1]{\dd x_{0:#1}}

\newcommand{\mci}{\mathcal I}
\newcommand{\mcio}{\mci^1}
\newcommand{\mcit}{\mci^2}
\newcommand{\mciso}{\mci^{*1}}
\newcommand{\mcist}{\mci^{*2}}
\newcommand{\trajtwo}{\mcit}

\newcommand{\trajtwos}{\mcist}
\newcommand{\trajthree}{\mci^3}
\newcommand{\trajthrees}{\mci^{*3}}
\newcommand{\trajfours}{\mci^{*4}}

\newcommand{\iit}[1]{\mathcal I_t^{#1}}
\newcommand{\iitmo}[1]{\mathcal I_{t-1}^{#1}}
\newcommand{\iits}[1]{\mathcal I_t^{*#1}}
\newcommand{\iitmos}[1]{\mathcal I_{t-1}^{*#1}}

\newcommand{\atbt}{A_t B_t}
\newcommand{\mblow}{\bar M_\ell}
\newcommand{\mbhigh}{\bar M_h}
\newcommand{\gblow}{\bar G_\ell}
\newcommand{\gbhigh}{\bar G_h}
\newcommand{\rtftm}{R_t, \FTM}
\newcommand{\bwdist}[1]{\frac{G_{t-1}(X_{t-1}^i) m_t(X_{t-1}^i , X_t^{#1})}{\sum_{j=1}^N G_{t-1}(X_{t-1}^j) m_t(X_{t-1}^j, X_t^{#1})}}
\newcommand{\hlbr}{\pr{\frac{\gbhigh \mbhigh}{\gblow \mblow}}^2} % high-low bound ratio
\newcommand{\prtf}[1]{\CProb{#1}{\rtftm}}
\newcommand{\gammastar}{\Gamma^*}
\newcommand{\gdbackward}{\overleftarrow {\mathfrak{C}^{2}}}
\newcommand{\pftm}[1]{\CProb{#1}{\FTM}}
\newcommand{\oneminusmlmh}{\pr{1 - \frac{\mblow}{\mbhigh}}}
\newcommand{\addpsi}[2]{\psi_{#1}(X_{#1-1}^{ \mci_{#1-1}^{#2} } , X_{#1}^{ \mci_{#1}^{#2} })}

\newcommand{\Tei}[2]{\frac 1 N \sum e^{i #1 #2}} %% Tei = Temporary command for Exponential of I
\newcommand{\Teiu}{\Tei{u}{X_{t-1}^n}}
\newcommand{\Teiv}{\Tei{v}{X_{t-1}^n}}
\newcommand{\Teiw}{\Tei{w}{X_t^n}}
\newcommand{\pinfty}{{\infty, \mrp}}

\newcommand{\dimx}{\operatorname{dim}_X}
\newcommand{\dimy}{\operatorname{dim}_Y}
\newcommand{\covx}{C_X}
\newcommand{\covy}{C_Y}
\newcommand{\sigmasmth}{\Sigma^{\mathrm{smth}}}
\newcommand{\sigmapred}{\Sigma^{\mathrm{pred}}}

\newcommand{\expfc}[1]{\exp\pr{-\abs{\log(1-L)} #1{x^{1/k}}}}

\newcommand{\znp}{ z^N \pr{ \frac{r_t^N(x_t)}{\mbhigh} } }
\newcommand{\intxt}[1]{
	\int_{\mathcal X_t} #1
}
\newcommand{\znpt}{ z^N \pr{ \frac{r_t(x_t)}{\mbhigh} } }

\newcommand{\tauinfffbs}{\tau_t^{\infty, \mrffbs}}
\newcommand{\xinfffbs}{X_t^{\infty, \mrffbs}}

\newcommand{\reweight}{
	\underline{Reweight}. Set $\omega_t^n \gets G_t(X_t^n)$ for $n=1,2,\ldots, N$\;
	Set $\ell_t^N \gets \sum_{n=1}^N \omega_t^n/N$\;
	Set $W_t^n \gets \omega_t^n/N\ell_t^N$ for $n=1,2,\ldots, N$\;
}

\newcommand{\prot}{Proof of Theorem\ }
\newcommand{\proe}{Proof of Equation\ }
\newcommand{\propp}{Proof of Proposition\ }
\newcommand{\prots}{Proof of Theorems\ }

\newcommand{\epsa}{\varepsilon_A}
\newcommand{\epsd}{\varepsilon_D}
\newcommand{\xpo}{X_{t-1}^{A_t^{n,1}}}
\newcommand{\xpt}{X_{t-1}^{A_t^{n,2}}}

\newcommand{\defcond}[1][]{X_t^{n#1}=x_t, X_{t-1}^{1:N}=x_{t-1}^{1:N}}
\newcommand{\eno}{^{n,1}}
\newcommand{\ent}{^{n,2}}

\newcommand{\expa}{^\mathrm A}
\newcommand{\expb}{^\mathrm B}
\newcommand{\ndist}{N_{\mathrm{dist}}}

\newcommand{\covmata}{\sigma\expa (\sigma\expa)^\top}
\newcommand{\covmatb}{\sigma\expb (\sigma\expb)^\top}

\begin{document}
%\begin{frontmatter}
\title{On backward smoothing algorithms}
\author{Hai-Dang Dau \& Nicolas Chopin}
\address{CREST-ENSAE, Institut Polytechnique de Paris}
\email{nicolas.chopin@ensae.fr}

%\begin{aug}
%	%%%%%%%%%%%%%%%%%%%%%%%%%%%%%%%%%%%%%%%%%%%%%%%
%	%% Only one address is permitted per author. %%
%	%% Only division, organization and e-mail is %%
%	%% included in the address.                  %%
%	%% Additional information can be included in %%
%	%% the Acknowledgments section if necessary. %%
%	%% ORCID can be inserted by command:         %%
%	%% \orcid{0000-0000-0000-0000}               %%
%	%%%%%%%%%%%%%%%%%%%%%%%%%%%%%%%%%%%%%%%%%%%%%%%
%	\author[A]{\fnms{Hai-Dang}~\snm{Dau}\orcid{0000-0002-0617-7566}}\and
%	\author[A]{\fnms{Nicolas}~\snm{Chopin}\ead[label=e2]{nicolas.chopin@ensae.fr}\orcid{0000-0002-0628-5815}}
%%	\author[B]{\fnms{Third}~\snm{Author}\ead[label=e3]{third@somewhere.com}}
%	%%%%%%%%%%%%%%%%%%%%%%%%%%%%%%%%%%%%%%%%%%%%%%
%	%% Addresses                                %%
%	%%%%%%%%%%%%%%%%%%%%%%%%%%%%%%%%%%%%%%%%%%%%%%
%	\address[A]{CREST-ENSAE, Institut Polytechnique de Paris, \printead{e2}}
%	
%	%\address[B]{Department,
%%		University or Company Name\printead[presep={,\ }]{e2}}
%\end{aug}

\begin{abstract}
    In the context of state-space models, skeleton-based smoothing algorithms
    rely on a backward sampling step which by default has a $\bigO(N^2)$
    complexity (where $N$ is the number of particles). Existing improvements in
    the literature are unsatisfactory: a popular rejection sampling-- based
    approach, as we shall show, might lead to badly behaved execution time;
    another rejection sampler with stopping lacks complexity analysis; yet
    another MCMC-inspired algorithm comes with no stability guarantee. We
    provide several results that close these gaps. In particular, we prove a
    novel non-asymptotic stability theorem, thus enabling smoothing with truly
    linear complexity and adequate theoretical justification. We propose a
    general framework which unites most skeleton-based smoothing algorithms in
    the literature and allows to simultaneously prove their convergence and
    stability, both in online and offline contexts. Furthermore, we derive, as
    a special case of that framework, a new coupling-based smoothing algorithm
    applicable to models with intractable transition densities. We elaborate
    practical recommendations and confirm those with numerical experiments.

    % of which a straightforward implementation gives an
    % $\mathcal O(N^2)$ cost, where $N$ is the number of particles. Using
    % rejection sampling has been proposed as a solution with a cost of $\mathcal
    % O(N)$ under strong assumptions. Unfortunately, we show that for many
    % practical models, the expected cost can be infinite even for finite $N$. We
    % fix this issue via a hybrid procedure which combines rejection sampling and
    % standard multinomial sampling, and leads to near-linear expected cost for
    % many models. We then use Markov Chain Monte Carlo (MCMC) steps rather than
    % performing exact sampling to make the execution time fully linear and
    % deterministic. Furthermore, we propose a natural extension of our framework
    % in which forward couplings are used instead of backward MCMC steps. This
    % permits online smoothing of models with intractable transition densities.
    % Non-degeneracy of our smoothing estimates as $T \to \infty$ is verified
    % both theoretically and numerically.
\end{abstract}
\maketitle

%\begin{keyword}[class=MSC]
%	\kwd[Primary ]{62M05}
%	\kwd{65C05}
%	\kwd[; secondary ]{65Y20}
%\end{keyword}
%
%\begin{keyword}
%	\kwd{state-space model}
%	\kwd{smoothing}
%	\kwd{sequential Monte Carlo}
%\end{keyword}
%\end{frontmatter}

\section{Introduction} 
  
\subsection{Background} 

A state-space model is composed of an unobserved Markov process $X_0, \ldots,
X_T$ and observed data $Y_0, \ldots, Y_T$. Given $X_0, \ldots, X_T$, the data
$Y_0, \ldots, Y_T$ are independent and generated through some specified
emission distribution $Y_t | X_t \sim \boldsymbol f_t(\dd y_t | x_t)$. These
models have wide-ranging applications (e.g.\ in biology, economics and
engineering). Two important inference tasks related to state-space models are
filtering (computing the distribution of $X_t$ given $Y_0, \ldots, Y_t$) and
smoothing (computing the distribution of the whole trajectory $(X_0, \ldots,
X_t)$, again given all data until time $t$). Filtering is usually carried out
through a particle filter, that is, a sequential Monte Carlo algorithm that
propagates $N$ weighted particles (realisations) through Markov and importance
sampling steps; see \cite{SMCbook} for a general introduction to state-space
models (Chapter 2) and particle filters (Chapter 10). 

This paper is concerned with skeleton-based smoothing algorithms, i.e.
algorithms that approximate the
smoothing distributions with empirical distributions based on the  
output of a particle filter (i.e.\ the locations and weights of the $N$
particles at each time step).  A simple example is genealogy tracking
(empirically proposed in \citealp{Kitagawa1996} and theoretically analysed in \citealp{delmoral_genealogy}) which keeps track of
the ancestry (past states) of each particles. This smoother suffers from 
degeneracy: for $t$ large enough, all the particles have the same ancestor at
time $0$. 

The forward filtering backward smoothing (FFBS) algorithm \citep{GodsDoucWest}
has been proposed as a solution to this problem. Starting from the filtering
approximation at time $t$, the algorithm samples successively particles at
times $t-1$, $t-2$, etc.\ using backward kernels. Its theoretical properties,
in particular the stability as $t \to \infty$, have been studied by
\citet{DelMoral2010, Douc2011}. 

%The naive implementation has
%an $\bigO(N^2)$ cost. However, if the Markov transition density is bounded, a
%rejection sampling--based scheme can be used \citep{Douc2011}. Its complexity
%is shown to be $\bigO(N)$ under restrictive assumptions on the model.

In many applications, one is mainly interested in approximating smoothing
expectations of additive functions of the form 
\[\CE{\psi_0(X_0) + \psi_1(X_0, X_1) + \cdots + \psi_t(X_{t-1}, X_t)}{Y_0, \ldots, Y_t}
\]
for some functions $\psi_0, \ldots, \psi_t$.  Such expectations can be
approximated in an online fashion by a procedure described in
\citet{DelMoral2010}. Inspired by this, the particle-based, rapid incremental
smoother (PaRIS) algorithm of \cite{Olsson2017} replaces some of the
calculations with an additional layer of Monte Carlo approximation. 

The backward sampling operation is central to both the FFBS and the PaRIS
algorithms. The naive implementation has an $\mathcal O(N^2)$ cost. There have
been numerous attempts at alleviating this problem in the literature, but, to our
knowledge, they all  lack formal support in terms of either computation
complexity or stability as $T \to \infty$. 

In the following five paragraphs, we elaborate on this limitation for each of
the three major contenders, and we point out two related challenges with current
backward sampling algorithms that we try to resolve in this article.

\subsection{State of the art}
\label{subsect:sota}

First, \citet{Douc2011} proposed to use rejection sampling for the generation
of backward indices in FFBS, and \citet{Olsson2017} extended this technique to
PaRIS. If the model has upper-and-lower bounded transition densities, this
sampler has an $\mathcal O(N)$ expected execution time (\citealt[Proposition
2]{Douc2011} and \citealt[Theorem 10]{Olsson2017}). Unfortunately, most
practical state space models (including linear Gaussian ones) violate this
assumption, and the behaviour of the algorithm in this case is unclear.
Empirically, it has been observed (\citealt{taghavi_adaptive,godsillmcmc};
\citealt[Section 4.3]{Olsson2017}) that in real examples, FFBS-reject and
PaRIS-reject frequently suffer from low acceptance rates, in contrary to what users would expect from an algorithm with linear complexity. To cite \citet{godsillmcmc}, ``[a]lthough theoretically elegant, the [...] algorithm has been found to suffer from such high rejection rates as to render it consistently slower than direct sampling implementation on problems with more than one state dimension''. To the best of our knowledge, no theoretical result has been put forward to formalise or quantify this bad behaviour.

Second, \citet{taghavi_adaptive} and \citet[Section 4.3]{Olsson2017} suggest putting a threshold on the number of rejection sampling trials to get more stable execution times. The thresholds are either chosen adaptively using a Kalman filter in \citet{taghavi_adaptive} or fixed at $\sqrt N$ in \citet[Section 4.3]{Olsson2017}. Although improvements are empirically observed, to the best of our knowledge, no theoretical analysis of the complexity of the resulting algorithm or formal justification of the proposed threshold is available.

Third, \citet{godsillmcmc} use MCMC moves starting from the filtering ancestor instead of a full backward sampling step. Empirically, this algorithm seems to prevent the degeneracy associated with the genealogy tracking smoother using a very small number of MCMC steps (e.g. less than five). Unfortunately, as far as we know, this stability property is not proved anywhere in the literature, which deters the adoption of the algorithm. Using MCMC moves provides a procedure with truly linear and deterministic run time, and a stability result is the only missing piece of the puzzle to completely resolve the $\mathcal O(N^2)$ problem. We believe one reason for the current state of affair is that the stability proof techniques employed by \citet{Douc2011} and \citet{Olsson2017} are difficult to extend to the MCMC case.

Fourth, and this is related to the third point, the stability of the PaRIS algorithm has only been proved in the asymptotic regime. More specifically, \citet{Olsson2017} established a central limit theorem as $N \to \infty$ in Corollary 5, then showed that the corresponding asymptotic error remains controlled as $T \to \infty$ in Theorem 8 and Theorem 9. While non-asymptotic stability bounds for the FFBS algorithm are already available in \citet{DelMoral2010,Douc2011,DubarryLeCorff2011}, we do not think that they can be extended straightforwardly to PaRIS and we are not aware of any such attempt in the literature.

Fifth, all backward samplers mentioned thus far require access to the transition density. Many models have dynamics that can be simulated from but transition densities that are not explicitly calculable. Enabling backward sampling in this scenario is challenging and will certainly require some kind of problem-specific knowledge to extract information from the transition densities, despite not being able to evaluate them exactly.

\subsection{Structure and contribution}\label{sec:motivation} 
Section~\ref{sec:smoothing_generic}
presents a general framework which admits as particular cases
a wide variety of existing algorithms (e.g.\ FFBS, forward-additive smoothing, PaRIS) as well as the novel ones considered later in the paper.
It allows to simultaneously prove the consistency as $N \to \infty$ and the stability as $T \to \infty$ for all of them.
The main ingredient is the discrete backward kernels, which are essentially random $N \times N$
matrices employed differently in the offline and the online settings.
On the technical side, the stability result is proved using a new technique, yielding a non-asymptotic bound that addresses the fourth point in subsection~\ref{subsect:sota}.

Next, we closely look at the use of rejection sampling and realise that in many
models, the resulting execution time may be significantly heavy-tailed; see
Section~\ref{sec:ffbs_kernel}. For instance, the run time of PaRIS may have
infinite  expectation, whereas the run time of FFBS may have infinite variance.
(Since it is technically more involved, the material for the FFBS algorithm is
delegated to Supplement~\ref{apx:ffbs_exec_result}.) These results address the
first point in subsection~\ref{subsect:sota} and we discuss their severe
practical implications.

We then derive and analyse hybrid rejection sampling schemes (i.e. schemes that
use rejection sampling only up to a certain number of attempts, and then switch
to the standard method). 
We show that they lead to a nearly $\mathcal{O}(N)$ algorithm (up to some
$\log$ factor) in Gaussian models; again see Section~\ref{sec:ffbs_kernel}.
This stems from the subtle interaction between the tail of Gaussian densities
and the finite Feynman-Kac approximation. Outside this class of model, the
hybrid algorithm can still escape the $\mathcal O(N^2)$ complexity, although it
might not reach the ideal linear run time target. These results shed some light
on the second issue mentioned in subsection~\ref{subsect:sota}.

In Section~\ref{sec:new_kernels}, we look at backward kernels that are more
efficient to simulate than the FFBS and the PaRIS ones. 
Section~\ref{sec:mcmc_kernel} describes backward kernels based on MCMC (Markov
chain Monte Carlo) following \citet{godsillmcmc} and extends them to the online
scenario. We cast this family of algorithms as a particular case of the general
framework developed in Section~\ref{sec:smoothing_generic}, which allows
convergence and stability to follow immediately. This solves the long-standing
problem described in the third point of subsection~\ref{subsect:sota}.

MCMC methods require evaluation of the
likelihood and thus cannot be applied to models with intractable transition
densities. In Section~\ref{sec:intractable}, we show how the use of forward
coupling can replace the role of backward MCMC steps in these scenarios. This
makes it possible to obtain stable performance in both on-line and off-line
scenarios (with intractable transition densities) and provides a possible solution to the fifth challenge describe in subsection~\ref{subsect:sota}. 
% off-line \textit{and} on-line smoothing for these models.

Section~\ref{sec:numexp} illustrates the aforementioned algorithms in both
on-line and off-line uses. We highlight how hybrid and MCMC samplers lead to a
more user-friendly (i.e.\ smaller, less random and less model-dependent)
execution time than the pure rejection sampler. We also apply our smoother for
intractable densities to a continuous-time diffusion process with
discretization. We observe that our procedure can indeed prevent degeneracy as $T \to
\infty$, provided that some care is taken to build couplings with good
performance. Section~\ref{sec:conclusion} concludes the paper with final practical
recommendations and further research directions. Table~\ref{table} gives an overview of existing and novel algorithms as well as our contributions for each.

\subsection{Related work}

Proposition 1 of \citet{Douc2011} states that under certain assumptions, the FFBS-reject algorithm has an \textit{asymptotic} $\mathcal O_{\P}(N)$ complexity. This does not contradict our results, which point out the undesirable properties of the \textit{non-asymptotic} execution time. Clearly, non-asymptotic behaviours are what users really observe in practice. From a technical point of view, the proof of \citet[Prop. 1]{Douc2011} is a simple application of Theorem 5 of the same paper. In contrast, non-asymptotic results such as Theorem~\ref{thm:ffbs_exec_infinite_general} and Theorem~\ref{thm:ffbs_exec_infinite_gaussian} require more delicate finite sample analyses.

Figure 1 of \citet{Olsson2017} and the accompanying discussion provide an
excellent intuition on the stability of smoothing algorithms based on the
support size of the backward kernels. We formalise this support size condition for the first time by the inequality ~\eqref{eq:support_cond} and construct a novel non-asymptotic stability result based on it. In contrast, \citet{Olsson2017} depart from their initial intuition and use an entirely different technique to establish stability. Their result is asymptotic in nature.

\citet{gloaguen2021pseudomarginal} briefly mention the use of MCMC in PaRIS algorithm, but their algorithm is fundamentally different to and less efficient than \citet{godsillmcmc}. Indeed, they do not start the MCMC chains at the ancestors previously obtained during the filtering step. They are thus obliged to perform a large number of MCMC iterations for decorrelation, whereas the algorithms described in our Proposition~\ref{prop:mcmc_validity}, built upon the ideas of \citet{godsillmcmc}, only require a single MCMC step to guarantee stability. However, we would like to stress again that \citet{godsillmcmc} did not prove this important fact.

Another way to reduce the computation time is to perform the expensive backward sampling steps at certain times $t$ only. For other values of $t$, the naive genealogy tracking smoother is used instead. This idea has been recently proposed by \citet{mastrototaro2021fast}, who also provided a practical recipe for deciding at which values of  $t$ the backward sampling should take place and derived corresponding theoretical results.

Smoothing in models with intractable transition densities is very challenging.
If these densities can be estimated accurately, the algorithms
proposed by \citet{gloaguen2021pseudomarginal} permit to attack this problem. A
case of particular interest is diffusion models, where unbiased transition
density estimators are provided in \citet{MR2278331, Fearnhead2008}. More
recently, \citet{yonekura2022online_smoothing} use a special bridge path-space
construction to overcome the unavailability of transition densities when the
diffusion (possibly with jumps) must be discretised.

Our smoother for intractable models are based on a general coupling principle
that is not specific to diffusions. We only require users to be able to
simulate their dynamics (e.g.\ using discretisation in the case of diffusions)
and to manipulate random numbers in their simulations so that dynamics starting
from two different points can meet with some probability. Our method does not
directly provide an estimator for the gradient of the transition density with
respect to model parameters and thus cannot be used in its current form to
perform maximum likelihood estimation (MLE) in intractable models; whereas the
aforementioned work have been able to do so in the case of diffusions. However,
the main advantage of our approach lies in its generality beyond the diffusion
case. Furthermore, modifications allowing to perform MLE are possible and might
be explored in further work specifically dedicated to the parameter estimation
problem.

The idea of coupling has been incorporated in the smoothing problem in a
different manner by \citet{jacob2019smoothing}. There, the goal is to provide
offline unbiased estimates of the expectation under the smoothing distribution.
Coupling and more generally ideas based on correlated random numbers are also
useful in the context of partially observed diffusions via the multilevel
approach \citep{jasra2017multilevel}.

In this work, we consider smoothing algorithms that are based on a unique pass
of the particle filter. Offline smoothing can be done using repeated iterations
of the conditional particle filter \citep{PMCMC}. Full trajectories can also be constructed in an online manner if one is willing to accept some lag approximations \citep{duffield2022_online_smoothing}. Another approach to smoothing
consists of using an additional information filter \citep{Fearnhead2010a}, but
it is limited to functions depending on one state only. Each of these
algorithmic families has their own advantages and disadvantages, of which a
detailed discussion is out of the scope of this article \citep[see
however][]{nordhquantitative}.

\section{General structure of smoothing algorithms}\label{sec:smoothing_generic}

In this section, we decompose each smoothing algorithm into two separate parts: the backward kernel (which determines its theoretical properties such as the convergence and the stability) and the execution mode (which is either online or offline and determines its implementation). This has two advantages: first, it induces an easy-to-navigate categorization of algorithms (see Table~\ref{table}); and second, it allows to prove the convergence and the stability for each of them by verifying sufficient conditions on the backward kernel component only.

\begin{table}
	\centering
	\begin{tabular}{|l|l|l|l|l|}
		\hline 
		Mode \textbackslash\  Kernel& FFBS kernel  & PaRIS kernel  & MCMC kernels  & Intract.  \\ 
		\hline 
		Offline& \splitcellbis{(*) FFBS}{(+) Thm. \ref{thm:ffbs_exec_infinite_general}, Thm. \ref{thm:ffbs_exec_infinite_gaussian}}{(+) Thm. \ref{thm:ffbs_hybrid_exec}, Cor. \ref{cor:gaussian_hybrid_ffbs}} &&
		\splitcell{(*) FFBS-MCMC}{(+) Prop. \ref{prop:mcmc_validity}}
		& (**)  \\ 
		\hline 
		Online& (*) Forward-additive & 
		\splitcellbis{(*) PaRIS}{(+) Thm. \ref{thm:stability}, Prop. \ref{prop:inf_expectation}}{(+) Thm. \ref{thm:intermediate_perf}, Thm. \ref{thm:near_linear}}
		& (**)  & (**)  \\ 
		\hline 
	\end{tabular}
	\caption{Summary of smoothing algorithms considered in this paper (classified by the backward kernel and the execution mode) and our contributions. (*) means an existing algorithm, (+) means a novel theoretical result and (**) means a novel algorithm}
	\label{table}
\end{table}

\subsection{Notations}

\textit{Measure-kernel-function notations.} Let $\mathcal X$ and $\mathcal Y$
be two measurable spaces with respective $\sigma$-algebras $\mathcal B(\mathcal
X)$ and $\mathcal B(\mathcal Y)$. The following definitions involve integrals
and only make sense when they are well-defined. For a measure $\mu$ on
$\mathcal X$ and a function $f: \mathcal X \to \mathbb R$, the notations $\mu
f$ and $\mu(f)$ refer to $\int f(x) \mu(\dd x)$.  A kernel (resp.\ Markov
kernel) $K$ is a mapping from $\mathcal X \times \mathcal B(\mathcal Y)$ to
$\mathbb R$ (resp.\ $[0,1]$) such that, for $B\in \mathcal B(\mathcal Y)$
fixed, $x \mapsto K(x, B)$ is a measurable function on $\mathcal X$; and for
$x$ fixed, $B \mapsto K(x, B)$ is a measure (resp.\ probability measure) on $\mathcal Y$. For a real-valued function $g$ defined on $\mathcal Y$, let $Kg: \mathcal X \to \mathbb R$ be the function $Kg(x) \eqdef \int g(y) K(x, \dd y)$. We sometimes write $K(x, g)$ for the same expression. The product of the measure $\mu$ on $\mathcal X$ and the kernel $K$ is a measure on $\mathcal Y$, defined by $\mu K(B)\eqdef \int K(x, B) \mu(\dd x)$. \textit{Other notations.}\notsep The notation $X_{0:t}$ is a shorthand for $(X_0, \ldots, X_t)$\notsep We denote by $\mathcal M(W^{1:N})$ the multinomial distribution supported on $\px{1, 2, \ldots, N}$. The respective probabilities are $W_1, \ldots, W_N$. If they do not sum to $1$, we implicitly refer to the normalised version obtained by multiplication of the weights with the appropriate constant\notsep The symbol $\cvprob$ means convergence in probability and $\Rightarrow$ means convergence in distribution\notsep The geometric distribution with parameter $\lambda$ is supported on $\mathbb Z_{\geq 1}$, has probability mass function $f(n) = \lambda (1-\lambda)^{n-1}$ and is noted by $\operatorname{Geo}(\lambda)$\notsep Let $\mathcal X$ and $\mathcal Y$ be two measurable spaces. Let $\mu$ and $\nu$ be two probability measures on $\mathcal X$ and $\mathcal Y$ respectively. The o-times product measure $\mu \otimes \nu$ is defined via $(\mu \otimes \nu)(h) := \iint h(x,y)\mu(\dd x)\nu(\dd y)$ for bounded functions $h: \mathcal X \times \mathcal Y \to \mathbb R$. If $X \sim \mu$ and $Y \sim \nu$, we sometimes note $\mu \otimes \nu$ by $X \otimes Y$.

\subsection{Feynman-Kac formalism and the bootstrap particle filter}\label{sec:fk_bootstrap} 

Let $\mathcal X_{0:T}$ be a sequence of measurable spaces and $M_{1:T}$ be a
sequence of Markov kernels such that $M_t$ is a kernel from $\mathcal X_{t-1}$
to $\mathcal X_t$. Let $X_{0:T}$ be an unobserved inhomogeneous Markov chain
with starting distribution $X_0 \sim
\mathbb M_0(\dd x_0)$ and Markov kernels $M_{1:T}$; i.e.\ $X_t|X_{t-1}\sim
M_t(X_{t-1}, \dx_t)$ for $t\geq 1$. We aim to study the distribution of
$X_{0:T}$ given observed data $Y_{0:T}$. Conditioned on $X_{0:T}$, the data
$Y_0, \ldots, Y_T$ are independent and 
$ Y_t | X_{0:T} \equiv Y_t | X_t \sim \boldsymbol f_t(\cdot | X_t) $
for a certain emission distribution
$\boldsymbol f_t(\dd y_t | x_t)$. Assume that there exists dominating measures 
$\tilde \lambda_t$ not depending on $x_t$ such that
\[\boldsymbol{f}_t(\dd y_t| x_t) = f_t(y_t|x_t) \tilde \lambda_t(\dd y_t).\]
The distribution of $X_{0:t}|Y_{0:t}$ is then given by
\begin{equation}
    \label{eq:fkmodel}
    \Q_t(\dd x_{0:t}) = \frac{1}{L_t} \mathbb{M}_0 (\dd x_0) \prod_{s=1}^t
    M_s(x_{s-1}, \dd x_s) G_s(x_s)
\end{equation}
where $G_s(x_s) := f(y_s | x_s)$ and $L_t > 0$ is the normalising constant.
Moreover, $\Q_{-1} \eqdef \mathbb M_0$ and $L_{-1}\eqdef 1$ by convention.
Equation~\eqref{eq:fkmodel} defines a Feynman-Kac model \citep{DelMoral:book}.
It does not require $M_t$ to admit a transition density, although herein we
only consider models where this assumption holds. Let $\lambda_t$ be a
dominating measure on $\mathcal X_t$ in the sense that there exists a function
$m_t$ (not necessarily tractable) such that
\begin{equation}
    \label{eq:density_mt}
    M_t(x_{t-1}, \dd x_t) = m_t(x_{t-1}, x_t) \lambda_t(\dd x_t).
\end{equation}

A special case of the current framework are linear Gaussian state space models.
They will serve as a running example for the article, and some of the results
will be specifically demonstrated for models of this class. The rationale is
that many real-world dynamics are partly, or close to, Gaussian. The notations
for linear Gaussian models are given in
Supplement~\ref{apx:linear_gaussian_models} and we will refer to them whenever
this model class is discussed.

Particle filters are algorithms that sample from $\Q_t(\dd x_t)$ in an on-line
manner.  In this article, we only consider the bootstrap particle filter
\citep{Gordon} and we detail its notations in Algorithm~\ref{algo:bootstrap}.
Many results in the following do apply to the auxiliary filter \citep{PittShep}
as well, and we shall as a rule indicate explicitly when it is \textit{not} the
case.

\begin{algo}{Bootstrap particle filter}\label{algo:bootstrap}
\KwIn{Feynman-Kac model~\eqref{eq:fkmodel}}
Simulate $X_0^{1:N} \overset{\text{i.i.d.}}{\sim} \mathbb M_0$\;
Set $\omega_0^n \gets G_0(X_0^n)$ for $n=1,\ldots, N$\;
Set $\ell_0^N \gets \sum_{n=1}^N \omega_0^n/N$\;
Set $W_0^n \gets \omega_0^n / N\ell_0^N$ for $n=1,\ldots, N$\;
\For{$t \gets 1$ \KwTo T}{
	\underline{Resample}. Simulate $A_t^{1:N} \overset{\text{i.i.d.}}{\sim} \mathcal{M}(W_{t-1}^{1:N})$ \;
	\underline{Move}. Simulate $X_t^n \sim M_t(X_{t-1}^{A_t^n}, \dd x_t)$ for $n=1,\ldots, N$\;
	\reweight
}
\KwOut{For all $t\geq 0$ and function $\varphi: \mathcal X_t \to \mathbb R$, the quantity $\sum_{n=1}^N W_t^n \varphi(X_t^n)$ approximates $\int \Q_t(\dd x_{0:t}) \varphi(x_t)$ and the quantity $\ell_t^N$ approximates $L_t/L_{t-1}$}
\end{algo}

We end this subsection with the definition of two sigma-algebras that will be referred to throughout the paper. Using the notations of Algorithm~\ref{algo:bootstrap}, let
\begin{equation}
\begin{split}
\label{eq:def:sigma_algebra_ft}
	\mathcal F_t &\eqdef \sigma(X_{0:t}^{1:N}, A_{1:t}^{1:N}),\\
	\mathcal F_t^- &\eqdef \sigma(X_{0:t}^{1:N}).
\end{split}
\end{equation}

\subsection{Backward kernels and off-line smoothing}\label{sec:examples_bw_ker}

In this subsection, we first describe three examples of backward kernels, in
which we emphasise both the random measure and the random matrix viewpoints. We
then formalise their use by stating a generic off-line smoothing algorithm.

\begin{example}[FFBS algorithm, \citealp{GodsDoucWest}] \label{exp:FFBS}

Once Algorithm~\ref{algo:bootstrap} has been run, the FFBS procedure generates
a trajectory approximating the smoothing distribution in a backward manner.
More precisely, it starts by simulating index $\mathcal{I}_T \sim
\mathcal{M}(W_T^{1:N})$ at time $T$. Then, recursively for $t= T, \ldots,
1$, given indices $\mathcal{I}_{t:T}$, it generates $\mathcal{I}_{t-1} \in
\px{1, \ldots, N}$ with probability proportional to $W_{t-1}^n m_t(X_{t-1}^n,
X_t^{\mathcal{I}_t})$. The smoothing trajectory is returned as $(X_0^{\mathcal
I_0}, \ldots, X_T^{\mathcal I_T})$. Formally, given
$\mathcal{F}_T$, the indices $\mathcal{I}_{0:T}$ are generated according to the
distribution
\[\mathcal{M}(W_t^{1:N})(\dd i_T) \ps{B_T^{N,\mathrm{FFBS}} (i_T, \dd i_{T-1}) B_{T-1}^{N, \mathrm{FFBS}}(i_{T-1}, \dd i_{T-2}) \ldots B_1^{N, \mathrm{FFBS}}(i_1, \dd i_0)}
\]
where the (random) backward kernels $B_t^{N, \mathrm{FFBS}}$ are defined by
\begin{equation} 
    \label{eq:ffbs_bw}
    B_t^{N, \mathrm{FFBS}} (i_t, \dd i_{t-1}) \eqdef \sum_{n=1}^N 
    \frac{W_{t-1}^n m_t(X_{t-1}^n, X_t^{i_t})}{\sum_{k=1}^N W_{t-1}^k
    m_t(X_{t-1}^k, X_t^{i_t})} \delta_{n}(\dd i_{t-1}).
\end{equation}

More simply, we can also look at these random kernels as random $N \times N$
matrices of which entries are given by
\begin{equation}
    \label{eq:ffbs_mt}
    \hat B_t^{N, \mathrm{FFBS}}[i_t, i_{t-1}] \eqdef \frac{W_{t-1}^{i_{t-1}}
        m_t(X_{t-1}^{i_{t-1}}, X_t^{i_t})}{\sum_{k=1}^N W_{t-1}^{k}
    m_t(X_{t-1}^{k}, X_t^{i_t})}.
\end{equation}

We will need both the kernel viewpoint~\eqref{eq:ffbs_bw} and the matrix
viewpoint~\eqref{eq:ffbs_mt} in this paper as the better choice depends on the
context.
\end{example}

\begin{example}[Genealogy tracking, \citealp{Kitagawa1996,delmoral_genealogy}]\label{exp:GT} 

It is well known that Algorithm~\ref{algo:bootstrap} already gives as a
by-product an approximation of the smoothing distribution. This information
can be extracted from the genealogy, by first simulating index
$\mathcal{I}_T \sim \mathcal{M}(W_T^{1:N})$ at time $T$, then successively
appending ancestors until time $0$ (i.e.\ setting sequentially
$\mathcal{I}_{t-1} \gets A_t^{\mathcal I_t}$). The smoothed trajectory is
returned as $(X_0^{\mathcal I_0}, \ldots, X_T^{\mathcal I_T})$. More
formally, conditioned on $\mathcal{F}_{T}$, we simulate the indices
$\mathcal I_{0:T}$ according to
\[\mathcal{M}(W_t^{1:N})(\dd i_T) \ps{B_T^{N,\mathrm{GT}} (i_T, \dd i_{T-1})
    B_{T-1}^{N, \mathrm{GT}}(i_{T-1}, \dd i_{T-2}) \ldots B_1^{N, \mathrm{GT}}(i_1,
\dd i_0)}
\]
where GT stands for ``genealogy tracking'' and the kernels 
$B_t^{N, \mathrm{GT}}$ are simply
\begin{equation}
    \label{eq:bwk:pm}
    B_t^{N, \mathrm{GT}}(i_t, \dd i_{t-1}) := \delta_{A_t^{i_t}}(\dd i_{t-1}).
\end{equation}

Again, it may be more intuitive to view this random kernel as a random 
$N \times N$ matrix, the elements of which are given by
\[\hat B_t^{N, \mathrm{GT}}[i_t, i_{t-1}] \eqdef \mathbbm{1}\px{i_{t-1} = A_t^{i_t}}.
\]
\end{example}

\begin{example}[MCMC backward samplers, \citealp{godsillmcmc}]\label{eg:mcmc_smoother}

In Example~\ref{exp:GT}, the backward variable $\mci_{t-1}$ is simply set to
$A_t^{\mci_t}$. On the contrary, in Example~\ref{exp:FFBS}, we need to launch a
simulator for the discrete measure $W_{t-1}^n m_t(X_{t-1}^n, X_t^{\mci_t})$.
Interestingly, the current value of $A_t^{\mci_t}$ is not taken into account in
that simulator. Therefore, a natural idea to combine the two previous examples
is to apply one (or several) MCMC steps to $A_t^{\mci_t}$ and assign the result
to $\mci_{t-1}$. The MCMC algorithm operates on the space $\px{1, 2, \ldots,
N}$ and targets the invariant measure $W_{t-1}^n m_t(X_{t-1}^n, X_t^{\mci_t})$. If only one independent Metropolis-Hastings (MH) step is used and the proposal
is $\mathcal M(W_{t-1}^{1:N})$, the corresponding random matrix $\hat B_t^{N,
\mathrm{IMH}}$ has values 
\[ \hat B_t^{N, \mathrm{IMH}}[i_t, i_{t-1}] =
    W_{t-1}^{i_{t-1}} \min\pr{1, {m_t(X_{t-1}^{i_{t-1}},
    X_t^{i_t})}/{m_t(X_{t-1}^{A_t^{i_t}}, X_t^{i_t})}}
\]
if $i_{t-1} \neq A_t^{i_t}$, and
\[
    \hat B_t^{N, \mathrm{IMH}}[i_t, A_t^{i_t}] = 1 - \sum_{n \neq A_t^{i_t}}
\hat B_t^{N, \mathrm{IMH}}[i_t, n]. 
\]

This third example shows that some elements of the matrix $\bwm{IMH}$ might be
expensive to calculate. If several MCMC steps are performed, \textnormal{all}
elements of $\bwm{IMH}$ will have non-trivial expressions. Still, simulating
from $\bw{IMH}(i_t, \dd i_{t-1})$ is easy as it amounts to running a standard
MCMC algorithm. MCMC backward samples are studied in more details in 
Section~\ref{sec:mcmc_kernel}.
\end{example}

We formalise how off-line smoothing can be done given random matrices $\hat
B_{1:T}^N$; see Algorithm~\ref{algo:offline_generic}. Note that in the above
examples, our matrices $\hat B_t^N$ are $\mathcal{F}_t$-measurable (i.e.\ they
depend on particles and indices up to time $t$), but this is not necessarily
the case in general (i.e.\ they may also depend on additional random variables,
see Section~\ref{sec:generic_online}). Furthermore, 
Algorithm~\ref{algo:offline_generic} describes how to perform smoothing using the
matrices $\hat B_{1:T}^N$, but does not say where they come from. At this
point, it is useful to keep in mind the above three examples. In 
Section~\ref{sec:validity}, we will give a general recipe for constructing valid
matrices $\hat B_t^N$ (i.e.\ those that give a consistent algorithm).
\begin{algo}{Generic off-line smoother}\label{algo:offline_generic}
	\KwIn{Filtering results $X_{0:T}^{1:N}$, $W_{0:T}^{1:N}$, and $A_{1:T}^{1:N}$ from Algorithm~\ref{algo:bootstrap}; random matrices $\hat B_{1:T}^N$ (see Section~\ref{sec:examples_bw_ker} for two examples of such matrices and Section~\ref{sec:validity} for a general recipe to construct them)}
	\For{$n \gets 1$ \KwTo $N$}{
		Simulate $\mathcal I_T^n \sim \mathcal{M}(W_T^{1:N})$\;
		\For{$t \gets T$ \KwTo 1}{
			Simulate $\mathcal I_{t-1}^n \sim B_t^N(\mathcal I_t^n, \dd i_{t-1})$ (the kernel $B_t^N(\mathcal I_t, \cdot)$ is defined by the $\mathcal I_t$-th row of the input matrix $\hat B_t^N$) \;
		}
	}
	\KwOut{The $N$ smoothed trajectories $(X_0^{\mathcal I_0^n}, \ldots, X_T^{\mathcal I_T^n})$ for $n = 1, \ldots, N$}
\end{algo}

Algorithm~\ref{algo:offline_generic} simulates, given $\mathcal{F}_T$
\textit{and} $\hat B_{1:T}^N$, $N$ i.i.d.\ index sequences $\mathcal I_{0:T}^n$, each distributed according to
\[\mathcal{M}(W_T^{1:N})(\dd i_T) \prod_{t=T}^1 B_t^N(i_t, \dd i_{t-1}).\]
Once the indices $\mathcal I_{0:T}^{1:N}$ are simulated, the $N$ smoothed
trajectories are returned as $(X_0^{\mathcal I_0^n}, \ldots, X_T^{\mathcal
I_T^n})$. Given $\mathcal F_T$ and $\hat B_{1:T}^N$, they are thus
conditionally i.i.d.\ and their conditional distribution is described by the $x_{0:T}$ component of the joint distribution
\begin{equation}
\label{eq:joint_empirical_smoothing}
\bar \Q_T^N(\dd x_{0:T}, \dd i_{0:T}) \eqdef \mathcal{M}(W_T^{1:N})(\dd i_T) \ps{\prod_{t=T}^1 B_t^N(i_t, \dd i_{t-1})} \ps{\prod_{t=T}^0 \delta_{X_t^{i_t}}(\dd x_t)}.
\end{equation}
Throughout the paper, the symbol $\bar \Q_T^N$ will refer to this joint distribution, while the symbol $\Q_T^N$ will refer to the $x_{0:T}$-marginal of $\bar \Q_T^N$ only. This allows the notation $\Q_T^N \varphi$ to make sense, where $\varphi = \varphi(x_0, \ldots, x_T)$ is a real-valued function defined on the hidden states.

\subsection{Validity and convergence}\label{sec:validity}

The kernels $B_t^{N, \text{FFBS}}$ and $B_t^{N, \text{GT}}$ are both valid
backward kernels to generate convergent approximation of the smoothing
distribution \citep{DelMoral:book, Douc2011}. This subsection shows that they
are not the only ones and gives a sufficient condition for a backward kernel to
be valid. It will prove a necessary tool to build more efficient $B_t^N$ later
in the paper.

Recall that Algorithm~\ref{algo:bootstrap} outputs particles $X_{0:T}^{1:N}$,
weights $W_{0:T}^{1:N}$ and ancestor variables $A_{1:T}^{1:N}$. Imagine that
the $A_{1:T}^{1:N}$ were discarded after filtering has been done and we wish to
simulate them back. We note that, since the $X_{0:T}^{1:N}$ are given, the $T
\times N$ variables $A_{1:T}^{1:N}$ are conditionally i.i.d. We can thus
simulate them back from
\[p(a_t^n | x_{0:T}^{1:N}) 
    = p(a_t^n | x_{t-1}^{1:N}, x_t^n) 
    \propto w_{t-1}^{a_t^n} m_t(x_{t-1}^{a_t^n}, x_t^n). 
\]
This is precisely the distribution of $B_t^{N, \text{FFBS}}(n, \cdot)$. It
turns out that any other invariant kernel that can be used for simulating back
the discarded $A_{1:T}^{1:N}$ will lead to a convergent algorithm as well. For
instance,  $\bw{GT}(n, \cdot)$ (Example~\ref{exp:GT}) simply returns back the
old $A_t^{n}$, unlike $\bw{FFBS}(n, \cdot)$ which creates a new version. The
kernel $\bw{IMH}(n, \cdot)$ (Example~\ref{eg:mcmc_smoother}) is somewhat an
intermediate between the two. We formalise these intuitions in the following
theorem. It is stated for the bootstrap particle filter, but as a matter of
fact, the proof can be extended straightforwardly to auxiliary particle 
filters as well.

\begin{assumption}\label{asp:Gbound}
	For all $0 \leq t \leq T$, $G_t(x_t) > 0$ and $\infnorm{G_t} < \infty$.
\end{assumption}

\begin{thm}\label{thm:convergence_mcmc}

We use the same notations as in Algorithms~\ref{algo:bootstrap} and
\ref{algo:offline_generic} (in particular, $\hat B_t^N$ denotes the transition
matrix that corresponds to the considered kernel $B_t^N$). Assume that for any
$1\leq t\leq T$, the random
matrix $\hat B_t^N$ satisfies the following conditions:
\begin{itemize}
    \item given $\mathcal{F}_{t-1}$ and $\hat B_{1:t-1}^N$, the variables
        $(X_t^n, A_t^n, \hat B_t^N(n, \cdot))$ for $n=1, \ldots, N$ are i.i.d.\ 
        and their distribution only depends on $X_{t-1}^{1:N}$, where $\hat
        B_t^N(n, \cdot)$ is the $n$-th row of matrix $\hat B_t^N$;
    \item if $J_t^n$ is a random variable such that \[J_t^n\  |\ 
        X_{t-1}^{1:N}, X_t^n, \hat B_{t}^N(n, \cdot) \sim B_t^N(n, \cdot) \],
        then $(J_t^n, X_t^n)$ has the same distribution as $(A_t^n, X_t^n)$
        given $X_{t-1}^{1:N}$.
\end{itemize}

Then under Assumption~\ref{asp:Gbound}, there exists constants $C_T>0$ and
$S_T < \infty$ such that, for any $\delta > 0$ and function $\varphi =
\varphi(x_0, \ldots, x_T)$:
\begin{equation} \label{eq:in_thm_convg}
\P\pr{\abs{\Q_T^N \varphi - \Q_T \varphi} \geq \frac{\sqrt{-2\log (\delta/2C_T)} S_T \norminf{\varphi}}{\sqrt N}} \leq \delta
\end{equation}
where $\Q_T^N$ is defined by \eqref{eq:joint_empirical_smoothing}.
\end{thm}

A typical relation between variables defined in the statement of the theorem is
illustrated by a graphical model in Figure~\ref{fig:thm1:variables}. (See
\citealt[Chapter 8]{Bishop:book} for the formal definition of graphical models
and how to use them.) By ``typical'', we mean that 
Theorem~\ref{thm:convergence_mcmc} technically allows for more complicated
relations, but the aforementioned figure captures the most essential cases.

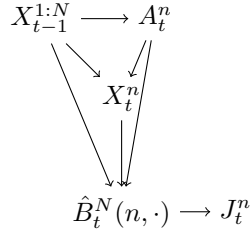
\begin{figure}
\centering
\begin{tikzpicture}[node distance={15mm}]
\node (xtm1) {$X_{t-1}^{1:N}$};
\node (atn) [right of=xtm1] {$A_t^n$};
\node (xtn) [below right of=xtm1] {$X_t^n$};
\node (btn) [below of=xtn] {$\hat B_t^N(n, \cdot)$};
\node (jtn) [right of=btn] {$J_t^n$};
\draw[->] (xtm1) -- (atn);
\draw[->] (xtm1) -- (xtn);
\draw[->] (atn) -- (xtn);
\draw[->] (xtm1) -- (btn);
\draw[->] (atn) -- (btn);
\draw[->] (xtn) -- (btn);
\draw[->] (btn) -- (jtn);
\end{tikzpicture}
\caption{Relation between variables described in
Theorem~\ref{thm:convergence_mcmc}.\label{fig:thm1:variables}}
\end{figure}

Theorem~\ref{thm:convergence_mcmc} is a generalisation of \citet[Theorem
5]{Douc2011}. Its proof thus follows the same lines
(Supplement~\ref{ap:proof:cvg}). However, in our case the measure $\Q_T^N(\dd
x_{0:T})$ is no longer Markovian. This is because the backward kernel
$B_t^N(i_t, \dd i_{t-1})$ does not depend on $X_t^{i_t}$ alone, but also
possibly on its ancestor  and extra random variables. This small difference has a big consequence: compared to \citet[Theorem 5]{Douc2011}, Theorem~\ref{thm:convergence_mcmc} has a much broader applicability and encompasses, for instance, the MCMC-based algorithms presented in Section~\ref{sec:mcmc_kernel} and novel kernels presented in Section~\ref{sec:intractable} for intractable densities.

As we have seen in~\eqref{eq:joint_empirical_smoothing}, $\Q_T^N$ is fundamentally a discrete measure of which the support contains $N^{T+1}$ elements. As such, $\Q_T^N \varphi$ cannot be computed exactly in general and must be approximated using $N$ trajectories $(X_0^{\mathcal I_0^n}, \ldots, X_T^{\mathcal I_T^n})$ simulated via Algorithm~\ref{algo:offline_generic}. Theorem~\ref{thm:convergence_mcmc} is thus completed by the following corollary, which is an immediate consequence of Hoeffding inequality (Supplement~\ref{ap:hoeffding}).

\begin{corollary}
Under the same setting as Theorem~\ref{thm:convergence_mcmc}, we have
\[
\P\pr{\abs{\frac 1N \sum_n \varphi(X_0^{\mathcal I_0^n}, \ldots, X_T^{\mathcal I_T^n})  - \Q_T \varphi} \geq \frac{\sqrt{-2\log \pr{\frac{\delta}{2(C_T + 1)}}} (S_T + 1) \infnorm{\varphi}}{\sqrt N} } \leq \delta. 
\]
\end{corollary}

\subsection{Generic on-line smoother}
\label{sec:generic_online}

As we have seen in Section~\ref{sec:examples_bw_ker} and
Section~\ref{sec:validity}, in general, the expectation $\Q_T^N \varphi$, for a
real-valued function $\varphi=\varphi(x_0, \ldots, x_T)$ of the hidden states,
cannot be computed exactly due to the large support ($N^{T+1}$ elements) of
$\Q_T^N$. Moreover, in certain settings we are interested in the quantities
$\Q_t^N \varphi_t$ for different functions $\varphi_t$. They cannot be
approximated in an on-line manner without more assumptions on the connection
between $\varphi_{t-1}$ and $\varphi_t$. If the family $(\varphi_t)$ is
additive, i.e.\ there exists functions $\psi_t$ such that
\begin{equation}
    \label{eq:add_func}
    \varphi_t(x_{0:t}) \eqdef \psi_0(x_0) + \psi_1(x_0, x_1) + \cdots +
    \psi_t(x_{t-1}, x_t)
\end{equation}
then we can calculate $\Q_t^N \varphi_t$ both exactly \textit{and} on-line. The
procedure was first described in \citet{DelMoral2010} for the kernel  $\Q_t^{N,
\text{FFBS}}$ (i.e.\ the measure defined by 
\eqref{eq:joint_empirical_smoothing} and the random kernels $B_t^{N,
\text{FFBS}}$), but we will use the idea for other kernels as well. In this
subsection, we first explain the principle of the method, then discuss its
computational complexity and the link to the PaRIS algorithm
\citep{Olsson2017}. 

\subsubsection*{Principle} 

For simplicity, we start with the special case $\varphi_t(x_{0:t}) =
\psi_0(x_0)$. Equation~\eqref{eq:joint_empirical_smoothing} and the matrix
viewpoint of Markov kernels then give
\[\Q_t^N \varphi_t =  \fastmatrix{W_t^1 \ldots W_t^N} \hat B_t^N \hat B_{t-1}^N
\ldots \hat B_1^N \fastmatrix{\psi_0(X_0^1) \\ \vdots \\ \psi_0(X_0^N)}.
\]
This naturally suggests the following recursion formula to compute 
$\Q_t^N \varphi_t$: 
\[\Q_t^N \varphi_t =  \fastmatrix{W_t^1 \ldots W_t^N} \hat S_t^N\]
with $\hat S_0^N = [\psi_0(X_0^1) \ldots \psi_0(X_0^N)]^\top$ and
\begin{equation}
    \label{eq:generic_online_recursion_easy}
    \hat S_t^N \eqdef \hat B_t^N \hat S_{t-1}^N.
\end{equation}
In the general case where functions $\varphi_t$ are given by 
\eqref{eq:add_func}, simple calculations
(Supplement~\ref{proof:full_online_recursion}) show that 
\eqref{eq:generic_online_recursion_easy} is replaced by
\begin{equation}
    \label{eq:generic_online_recursion_hard}
    \hat S_{t}^N \eqdef \hat B_t^N \hat S_{t-1}^N + \operatorname{diag}(\hat B_t^N \hat \psi_{t}^N)
\end{equation}
where the $N \times N$ matrix $\hat \psi_t^N$ is defined by  
\[\hat \psi_t^N [i_{t-1}, i_t] \eqdef \psi_t(X_{t-1}^{i_{t-1}}, X_t^{i_t})\] 
and the operator
$\operatorname{diag}: \mathbb R^{N \times N} \rightarrow \mathbb R^N$ extracts
the diagonal of a matrix. This is exactly what is done in 
Algorithm~\ref{algo:online_generic}.
\begin{algo}{Generic on-line smoother for additive functions (one step)}
	\label{algo:online_generic}
	\KwIn{Particles $X_{t-1}^{1:N}$ and weights $W_{t-1}^{1:N}$ at time $t-1$;
    the $N \times 1 $ vector $\hat S_{t-1}^N$ (see text); additive function~\eqref{eq:add_func}}
	Generate $X^{1:N}_t$ and $W_t^{1:N}$ according to the particle filter
    (Algorithm~\ref{algo:bootstrap}) \; Calculate the random matrix $\hat B_t^N$ (see Section~\ref{sec:examples_bw_ker} and Section~\ref{sec:validity})\;
	Create the $N\times 1$ vector $\hat S_t^N$ according to~\eqref{eq:generic_online_recursion_hard}. More precisely:\;
	\For{$i_t \gets 1$ \KwTo $N$}{
		$\hat S_t^N[i_t] \gets \sum_{i_{t-1}} \hat B_t^N[i_t, i_{t-1}] \pr{\hat S_{t-1}^N[i_{t-1}] + \psi_t(X_{t-1}^{i_{t-1}}, X_t^{i_t})}$
	}
	\KwOut{Quantity $\sum_n W_t^n \hat S_t^N[n]$ which is equal to $\Q_t^N \varphi_t$ and is an esimate of $\Q_t(\varphi_t)$; particles $X_t^{1:N}$, weights $W_t^{1:N}$ and vector $S_t^N$ for the next step}
\end{algo}

\subsubsection*{Computational complexity and the PaRIS algorithm}

Equations~\eqref{eq:generic_online_recursion_easy} and
\eqref{eq:generic_online_recursion_hard} involve a matrix-vector multiplication
and thus require, in general, $\mathcal{O}(N^2)$ operations to be evaluated.
When $\hat B_t^N \equiv \hat B_t^{N, \text{FFBS}}$, 
Algorithm~\ref{algo:online_generic} becomes the $\mathcal O(N^2)$ on-line
smoothing
algorithm  of~\cite{DelMoral2010}. %, see also \citet[Algorithm 12.1]{SMCbook}. 
The $\mathcal{O}(N^2)$ complexity can however be lowered to $\mathcal{O}(N)$ if
the matrices $\hat B_t^N$ are \textit{sparse}. This is the idea behind the
PaRIS algorithm \citep{Olsson2017}, where the full matrix $\hat B_t^{N,
\text{FFBS}}$ is unbiasedly estimated by a sparse matrix $\hat B_t^{N,
\text{PaRIS}}$. More specifically, for any integer $\tilde N > 1$, for any $n
\in 1, \ldots, N$, let $J_t^{n, 1}, \ldots, J_t^{n, \tilde N}$ be conditionally
i.i.d.\ random variables simulated from $B_t^{N, \text{FFBS}}(n, \cdot)$. The
random matrix $\hat B_t^{N, \text{PaRIS}}$ is then defined as
\[\hat B_t^{N, \text{PaRIS}}[n,m] \eqdef 
    \frac{1}{\tilde N} \sum_{\tilde n=1}^{\tilde N} \mathbbm{1}\px{J_t^{n,\tilde n} = m}
\]
and the corresponding random kernel is
\begin{equation}
    \label{eq:bwk:paris}
    B_t^{N, \text{PaRIS}}(n, \dd m) = 
    \frac{1}{\tilde N} \sum_{\tilde n=1}^{\tilde N} \delta_{J_t^{n, \tilde n}}(\dd m).
\end{equation}
The following straightforward proposition establishes the validity of the
$B_t^{N, \text{PaRIS}}$ kernel. Together with Theorem~\ref{thm:convergence_mcmc}, it can be thought of as a reformulation of the consistency of the PaRIS algorithm \citep[Corollary 2]{Olsson2017} in the language of our framework.
\begin{prop}
    The matrix $\hat B_t^{N, \mathrm{PaRIS}}$ has only $\mathcal{O}(N \tilde
    N)$ non-zero elements out of $N^2$. It is an unbiased estimate of $\hat
    B_t^{N, \mathrm{FFBS}}$ in the sense that
    \[\CE{\hat B_t^{N, \mathrm{PaRIS}}}{\mathcal F_t} = \hat B_t^{N,
    \mathrm{FFBS}}.  \] Moreover, the sequence of matrices $B_{1:T}^{N,
\mathrm{PaRIS}}$ satisfies the two conditions of
Theorem~\ref{thm:convergence_mcmc}.
\end{prop}
 The proposition also justifies the $\mathcal O(N)$
complexity of~\eqref{eq:generic_online_recursion_easy} and
\eqref{eq:generic_online_recursion_hard}, as long as $\tilde N$ is fixed as
$N \to \infty$.
But it is important to remark that the preceding $\mathcal O(N)$ complexity does not include the cost of
generating the matrices $\hat B_t^{N, \text{PaRIS}}$ themselves, i.e., the
operations required to simulate the indices $J_t^{n, \tilde n}$. In
\citet{Olsson2017} it is argued that such simulations have an $\mathcal{O}(N)$
cost using the rejection sampling method whenever the transition density is
both upper and lower bounded. Section~\ref{sec:ffbs_kernel} investigates the
claim when this hypothesis is violated.

\subsection{Stability}
\label{sec:stability}
When $\hat B_t^N \equiv \hat B_t^{N, \text{GT}}$,
Algorithms~\ref{algo:offline_generic} and \ref{algo:online_generic} reduce to the genealogy tracking smoother \citep{Kitagawa1996}. The matrix 
$\hat B_t^{N, \text{GT}}$ is indeed sparse, leading to the well-known
$\mathcal{O}(N)$ complexity of this on-line procedure. As per
Theorem~\ref{thm:convergence_mcmc}, smoothing via genealogy tracking is
convergent at rate $\mathcal{O}(N^{-1/2})$ if $T$ is \textit{fixed}. When $T
\to \infty$ however, all particles will eventually share the same ancestor at
time $0$ (or any fixed time $t$). Mathematically, this phenomenon is manifested
in two ways: (a) for fixed $t$ and function $\phi_t: \mathcal X_t \to \mathbb
R$, the error of estimating $\E[\phi_t(X_t) | Y_{0:T}]$ grows linearly with
$T$; and (b) the error of estimating $\CE{\sum_{t=0}^T \psi_t(x_{t-1},
x_t)}{Y_{0:T}}$ grows quadratically with $T$. These correspond respectively to
the degeneracy for the fixed marginal smoothing and the additive smoothing
problems; see also the introductory section of \citet{Olsson2017} for a
discussion. The random matrices $\hat B_t^{N, \text{GT}}$ are therefore said to
be \textit{unstable} as $T \to \infty$, which is not the case for $\hat B_t^{N,
\text{FFBS}}$ or $\hat B_t^{N, \text{PaRIS}}$. This subsection gives sufficient
conditions to ensure the stability of a general $\hat B_t^N$.

The essential point behind smoothing stability is simple: the support of
$B_t^{N, \mathrm{FFBS}}(n, \cdot)$ or $B_t^{N, \mathrm{PaRIS}}(n, \cdot)$ for
$\tilde N \geq 2$ contains more than one element, contrary to that of $B_t^{N,
\mathrm{GT}}(n, \cdot)$. This property is formalised by~\eqref{eq:support_cond}. 
To explain the intuitions, we use the notations of
Algorithm~\ref{algo:offline_generic} and consider the estimate 
\begin{equation*}
    N^{-1}
    \pr{\psi_0(X_0^{\mathcal I_0^1}) + \cdots + \psi_0(X_0^{\mathcal I_0^N})} 
\end{equation*}
of $\CE{\psi_0(X_0)}{Y_{0:T}}$ when $T \to \infty$. The variance of the quantity
above is a sum of $\Cov(\psi_0(X_0^{\mathcal I_0^i}),
\psi_0(X_0^{\mathcal I_0^j}))$ terms. It can therefore be understood by looking
at a pair of trajectories simulated using Algorithm~\ref{algo:offline_generic}.

At final time $t=T$, $\mathcal I_T^1$ and $\mathcal I_T^2$ both follow the
$\mathcal{M}(W_T^{1:N})$ distribution. Under regularity conditions (e.g.\ no
extreme weights), they are likely to be different, i.e., $\P(\mathcal I_T^1 =
\mathcal{I}_T^2) = \mathcal O(1/N)$. This property can be propagated backward:
as long as $\mathcal I_t^1 \neq \mathcal I_t^2$, the two variables $\mathcal
I_{t-1}^1$ and $\mathcal{I}_{t-1}^2$ are also likely to be different, with
however a small $\mathcal{O}(1/N)$ chance of being equal. Moreover, as long as
the two trajectories have not met, they can be simulated independently given
$\mathcal F_T^-$ (the sigma algebra defined
in~\eqref{eq:def:sigma_algebra_ft}). In mathematical terms, under the two
hypotheses of Theorem~\ref{thm:convergence_mcmc}, given $\mathcal{F}_T^-$ and
$\mathcal{I}_{t:T}^{1,2}$, it can be proved that the two variables
$\mathcal{I}_{t-1}^1$ and $\mathcal I_{t-1}^2$ are independent if $\mathcal
I_t^1 \neq \mathcal I_t^2$ (Lemma~\ref{lem:two_backward_traj},
Supplement~\ref{proof:thm:stability}).

Since there is an $\mathcal{O}(1/N)$ chance of meeting at each time step, if $T \gg N$, it is likely that the two paths will meet at some point $t \gg 0$. When $\mathcal{I}_t^1 = \mathcal I_t^2$, the two indices $\mathcal I_{t-1}$ and $\mathcal I_{t-2}$ are both simulated according to $B_t^N(\mathcal I_t^1, \cdot)$. In the genealogy tracking algorithm, $B_t^{N, \text{GT}}(i, \cdot)$ is a Dirac measure, leading to $\mathcal I_{t-1}^1 = \mathcal{I}_{t-1}^2$ almost surely. This spreads until time $0$, so $\operatorname{Corr}(\psi_0(X_0^{\mathcal I_0^1}), \psi_0(X_0^{\mathcal I_0^2}))$ is almost $1$ if $T \gg N$.

Other kernels like $B_t^{N, \text{FFBS}}$ or $B_t^{N, \text{PaRIS}}$ do not suffer from the same problem. For these, the  support size of $B_t^{N}(\mathcal I_t^1, \cdot)$ is greater than one and thus there is some real chance that $\mathcal I_{t-1}^1 \neq \mathcal{I}_{t-1}^2$. If that does happen, we are again back to the regime where the next states of the two paths can be simulated independently. Note also that the support of $B_t^{N}(\mathcal I_t^1, \cdot)$ does not need to be large and can contain as few as $2$ elements. Even if $\mathcal I_{t-1}^1$ might still be equal to $\mathcal I_{t-1}^2$ with some probability, the two paths will have new chances to diverge at times $t-2$, $t-3$ and so on. Overall, this makes $\operatorname{Corr}(\psi_0(X_0^{\mathcal I_0^1}), \psi_0(X_0^{\mathcal I_0^2}))$ quite small (Lemma~\ref{lem:cond_covar_two_traj}, Supplement~\ref{proof:thm:stability}).

We formalise these arguments in the following theorem, whose proof
(Supplement~\ref{proof:thm:stability}) follows them very closely. The price for
proof intuitiveness is that the theorem is specific to the bootstrap filter,
although numerical evidence (Section~\ref{sec:numexp}) suggests that other filters
are stable as well.

\begin{assumption}\label{asp:mt_2ways_bound}
The transition densities $m_t$ are upper and lower bounded: 
\[\bar M_\ell \leq m_t(x_{t-1}, x_t) \leq \bar M_h\]
for constants $0 < \bar M_\ell < \bar M_h < \infty$.
\end{assumption}

\begin{assumption}\label{asp:g_2ways_bound}
The potential functions $G_t$ are upper and lower bounded:
\[ \bar G_\ell \leq G_t(x_t) \leq \bar G_h \]
for constants $0 < \bar G_\ell < \bar G_h < \infty$.
\end{assumption}
\textbf{Remark.} Since Assumption~\ref{asp:mt_2ways_bound} implies that the
$\mathcal X_t$'s are compact, Assumption~\ref{asp:Gbound} automatically implies 
Assumption~\ref{asp:g_2ways_bound} as soon as the $G_t$'s' are continuous functions.

\begin{thm}\label{thm:stability} 
    We use the notations of Algorithms~\ref{algo:bootstrap}
    and~\ref{algo:offline_generic}. Suppose that
    Assumptions~\ref{asp:mt_2ways_bound} and~\ref{asp:g_2ways_bound} hold and
    the random kernels $B_{1:T}^{N}$ satisfy the conditions of
    Theorem~\ref{thm:convergence_mcmc}. If, in addition, for the pair of random
    variables $(J_t^{n, 1}, J_t^{n, 2})$ whose distribution given
    $X_{t-1}^{1:N}$, $X_t^n$ and $\hat B_t^N(n, \cdot)$ is defined by $B_t^N(n,
    \cdot) \otimes B_t^N(n, \cdot)$, we have
    \begin{equation}
        \label{eq:support_cond}
        \CProb{J_t^{n,1}\neq J_t^{n,2}}{ X_{t-1}^{1:N}, X_t^n} \geq \epss
    \end{equation}
    for some $\epss > 0$ and all $t$, $n$; then there exists a constant
    $C$ not depending on $T$ such that:
    \begin{itemize}
        \item fixed marginal smoothing is stable, i.e.\ for $s \in \px{0, \ldots, T}$
            and a real-valued function $\phi_s: \mathcal{X}_s \to \mathbb R$ of the
            hidden state $X_s$, we have
            \begin{equation}
                \label{eq:thm_stability_fixed} \E\ps{\pr{\int \Q_T^N(\dd x_s)
                \phi_s(x_s) - \CE{\phi_s(X_s)}{Y_{0:T}}}^2} \leq \frac{C
            \infnorm{\phi_s}^2}{N};
        \end{equation}
    \item additive smoothing is stable, i.e.\ for $T \geq 2$ and the function
        $\varphi_T$ defined in \eqref{eq:add_func}, we have
        \begin{equation}
            \label{eq:thm_stability_additive}
            \E\ps{\pr{\Q_T^N(\varphi_T) - \Q_T(\varphi_T)}^2} \leq
            \frac{C\sum_{t=0}^T \infnorm{\psi_t}^2}{N} \pr{1 + \sqrt{\frac
            TN}}^2.
        \end{equation}
\end{itemize}
\end{thm}

In particular, when $B_t^N$ is the PaRIS kernel with $\tilde N \geq 2$, Theorem~\ref{thm:stability} implies a novel non-asymptotic bound for the PaRIS algorithm. \citet{Olsson2017} first established a central limit theorem as $N\to\infty$ and $T$ fixed, then showed that the asymptotic variance is controlled as $T \to \infty$. In contrast, we follow an original approach (whose intuition is explained at the beginning of this subsection) in order to derive a finite sample size bound.

The main technical difficulty is to prove the fast mixing of the Markov kernel product $B_t^{N} B_{t-1}^{N}\ldots B_{t'}^{N}$ in terms of $t-t'$. For the original FFBS kernel,
the stability proof by \citet{Douc2011} relies on the uniform Doeblin property of each of the term $B_s^{N, \mrffbs}$
(page 2136, towards the end of their proof of Lemma 10) and from there, deduces
the exponentially fast mixing of the product. When $B_s^{N,\mrffbs}$ is
approximated by a sparse matrix $B_s^N$ (which is the case for PaRIS, but also
for certain MCMC-based and coupling-based smoothers that we shall see later),
the aforementioned property no longer holds for each individual term $B_s^N$.
Interestingly however, the good mixing of $B_t^{N,\mrffbs} \ldots B_{t'}^{N,
\mrffbs}$ is still conserved in the product $B_t^N \ldots B_{t'}^N$. In
Lemma~\ref{lem:cond_covar_two_traj}, we show that two trajectories generated
via the latter kernel have such a small correlation that they are virtually indistinguishable from two independent trajectories generated via the former one.

Theorem~\ref{thm:stability} is stated under strong assumptions (similar to
those used in \citealt[Chapter 11.4]{SMCbook}, and slightly stronger than
\citealt[Assumption 4]{Douc2011}). On the other hand, it applies to a large class
of backward kernels (rather than only FFBS), including the new ones introduced
in the forthcoming sections.

In the proof of this theorem, we proceed in two steps: first, we apply existing bounds (\citealp[Theorem 3.1]{DubarryLeCorff2011} and  \citealp[Chapter 17]{delmoral_mean_field}) for the error between the $\btn{FFBS}$-induced distribution and the true target; and second, we use our own techniques to control the error when $\btn{FFBS}$ is replaced by any other kernel $B_t^n$ satisfying \eqref{eq:support_cond}. The $(1+\sqrt{T/N})^2$ term in \eqref{eq:thm_stability_additive} comes from the first part and we
do not know whether it can be dropped. However, it does not affect the scaling
of the algorithm. Indeed, with or without it, the inequality implies that in
order to have a constant error in the additive smoothing problem, one only has
to take $N=\mathcal O(T)$ (instead of $N = \mathcal O(T^2)$ without backward
sampling). Moreover, from an asymptotic point of view, we always have
$\sigma^2(T) = \mathcal O(T)$ regardless of the presence of the $(1 +
\sqrt{T/N})^2$ term, where
$ \sigma^2(T) \eqdef \lim_{N\to\infty} N \E\ps{\pr{\Q_T^N(\varphi_T) -
\Q_T(\varphi_T)}^2}$.
% For many of those, the
% computational cost is not altered if the aforementioned assumptions are
% violated. The rejection-based backward simulation \citep{Douc2011, Olsson2017}
% does not enjoy the same property, as we shall show in
% Section~\ref{sec:ffbs_kernel}.

\section{Sampling from the FFBS Backward Kernels}
\label{sec:ffbs_kernel}
Sampling from the FFBS backward kernel lies at the heart of both the FFBS
algorithm (Example~\ref{exp:FFBS}) and the PaRIS one
(Section~\ref{sec:generic_online}). Indeed, at time $t$, they require
generating random variables distributed according to $\btn{FFBS}(i_t, \dd
i_{t-1})$ for $i_t$ running from $1$ to $N$. Since sampling from a discrete
measure on $N$ elements requires $\mathcal O(N)$ operations (e.g.\ via CDF inversion), the total computational cost becomes $\mathcal{O}(N^2)$. To reduce this, we start by considering the subclass of models satisfying the following assumption, which is much weaker than Assumption~\ref{asp:mt_2ways_bound}.
\begin{assumption}
	\label{asp:Ct}
	The transition density $m_t(x_{t-1}, x_t)$ is strictly positive and upper bounded, i.e.\ there exists $\mbhigh > 0$ such that $0 < m_t(x_{t-1}, x_t) \leq \mbhigh, \forall \ (x_{t-1}, x_t)$.
\end{assumption}
The motivation for the first condition $0 < m_t(x_{t-1}, x_t)$ will be clear
after Assumption~\ref{asp:space} is defined. For now, we see that it is
possible to sample from $\btn{FFBS}(i_t, \dd i_{t-1})$ using rejection sampling
via the proposal distribution $\mathcal{M}(W_{t-1}^{1:N})$. After an $\mathcal
O(N)$-cost initialisation, new draws can be simulated from the proposal in
amortised $\mathcal O(1)$ time; see \citet[Python Corner, Chapter 9]{SMCbook},
see also \citet[Appendix B.1]{Douc2011} for an alternative algorithm with an
$\mathcal O(\log N)$ cost per  draw. The resulting procedure is summarised in
Algorithm~\ref{algo:pure_rejection_sampler}. Compared to traditional FFBS or
PaRIS implementations, these rejection--based variants have a random
execution time that is more difficult to analyse. Under
Assumption~\ref{asp:mt_2ways_bound}, \citet{Douc2011} and \citet{Olsson2017}
derive an $\mathcal O(N \mbhigh/\mblow)$ expected complexity. However, the
general picture, where the state space is not compact and only
Assumption~\ref{asp:Ct} holds, is less clear.
\begin{algo}{Pure rejection sampler for simulating from $\bw{FFBS}(i_t, \dd i_{t-1})$}
	\label{algo:pure_rejection_sampler}
	\KwIn{Particles $X_{t-1}^{1:N}$ and weights $W_{t-1}^{1:N} $ at time $t-1$; particle $X_t^{i_t}$ at time $t$; constant $\mbhigh$; pre-initialised $\mathcal{O}(1)$ sampler for $\mathcal{M}(W_{t-1}^{1:N})$}
	\Repeat{$U\leq m_t(X_{t-1}^{\mathcal I_{t-1}}, X_t^{i_t})/\mbhigh$}{
		$\mathcal I_{t-1} \sim \mathcal{M}(W_{t-1}^{1:N})$ using the pre-initialised $\mathcal{O}(1)$ sampler \;
		$U \sim \operatorname{Unif}[0,1]$\;
	}
	\KwOut{$\mathcal I_{t-1}$, which is distributed according to $\bw{FFBS} (i_t, \dd i_{t-1})$. }
\end{algo}

The present subsection intends to fill this gap. Our main focus is the PaRIS
algorithm of which the presentation is simpler. Results for the FFBS algorithm
can be found in Supplement~\ref{apx:ffbs_exec_result}. We restrict ourselves to
the case where $\mathcal X_t = \mathbb R^{d_t}$, although extensions to other
non compact state spaces are possible. Only the bootstrap particle filter is
considered, and results from this section do \textit{not} extend trivially to
other filtering algorithms. In Section~\ref{sec:numexp}, we shall employ
different types of particle filters and see that the performance could change
from one type to another, which is an additional weak point of rejection-based
algorithms.

\begin{assumption}
	\label{asp:space}
	The hidden state $X_t$ is defined on the space $\mathcal{X}_t = \mathbb{R}^{d_t}$. The measure $\lambda_t(\dd x_t)$ with respect to which the transition density $m_t(x_{t-1}, x_t)$ is defined (cf. \eqref{eq:density_mt}) is the Lebesgue measure on $\mathbb R^{d_t}$.
\end{assumption}

This assumption together with the condition $m_t(x_{t-1}, x_t) > 0$ of
Assumption~\ref{asp:Ct} ensures that the state space model is ``truly
non-compact''. Indeed, if $m_t(x_{t-1}, x_t)$ is zero whenever $x_{t-1} \notin
\mathcal C_{t-1}$ or $x_t \notin \mathcal C_t$, where $\mathcal C_{t-1}$ and
$\mathcal C_t$ are respectively two compact subsets of $\mathbb R^{d_{t-1}}$
and $\mathbb R^{d_t}$, then we are basically reduced to a state space model
where $\mathcal X_{t-1} = \mathcal C_{t-1}$ and $\mathcal X_t = \mathcal C_t$.

\subsection{Complexity of PaRIS algorithm with pure rejection sampling}

We consider the PaRIS algorithm (i.e.\ Algorithm~\ref{algo:online_generic} using the $\bw{PaRIS}$ kernels). Algorithm~\ref{algo:paris_concrete} provides a concrete description of the resulting procedure, using the bootstrap particle filter. At each time $t$, let $\tau_t^{n, \mrp}$ be the number of rejection trials required to sample from $\btn{FFBS}(n, \dd m)$. We then have
\begin{equation}
\label{eq:dist_tau_N}
\tau_t^{n, \mrp} \textrm{ } | \textrm{ } \mathcal{F}_{t-1}, X_t^n \sim \operatorname{Geo}\pr{\frac{\sum_i W_{t-1}^i m_t(X_{t-1}^i, X_t^n)}{\mbhigh}}
\end{equation}
with $\mbhigh$ defined in Assumption~\ref{asp:Ct}.

\begin{algo}{Concrete implementation of PaRIS algorithm (i.e.\ Algorithm~\ref{algo:online_generic} with the $\bw{PaRIS}$ backward kernel) using the bootstrap particle filter}
	\label{algo:paris_concrete}
	\KwIn{Particles $X_{t-1}^{1:N}$; weights $W_{t-1}^{1:N}$; vector $S_{t-1}^N$ in $\mathbb{R}^N$; pre-initialised sampler for $\mathcal{M}(W_{t-1}^{1:N})$; function $\psi_t$ (cf. \eqref{eq:add_func}); user-specified parameter $\tilde N$}
	\For{$n \gets 1$ \KwTo $N$ }{
		$A_t^n \sim \mathcal{M}(W_{t-1}^{1:N})$ $(\star)$\;
		$X_t^n \sim M_t(X_{t-1}^{A_t^n}, \dd x_t)$ \;
		Simulate $J_t^{n, 1:\tilde N} \overset{\textrm{i.i.d.}}{\sim} \bw{FFBS}(n, \dd n') $ using either the pure rejection sampler (Algorithm~\ref{algo:pure_rejection_sampler}) or the hybrid rejection sampler (Algorithm~\ref{algo:hybrid_rejection_sampler})\;
		$S_t^N[n] \gets \tilde N^{-1} \sum_{\tilde n=1}^{\tilde N} \px{ S_{t-1}^N[J_t^{n, \tilde n}] + \psi_t(X_{t-1}^{J_t^{n, \tilde n}}, X_t^n)}$\;
	}
	\For{$n \gets 1$ \KwTo $N$}{
		$W_t^n \gets G_t(X_t^n)/\sum_{i} G_t(X_t^i)$\;
	}
	$\mu_t^N \gets \sum_{n=1}^N W_t^n S_t^N(n)$\;
	Initialise a sampler for $\mathcal{M}(W_t^{1:N})$\;
	\KwOut{Estimate $\mu_t^N$ of $\CE{\varphi(X_{0:t})}{Y_{0:t}}$; particles $X_t^{1:N}$; weights $W_t^{1:N}$; vector $S_t^N$ in $\mathbb{R}^N$ and pre-initialised sampler $\mathcal{M}(W_t^{1:N})$ for the next iteration}
\end{algo}

By exchangeability of particles, the expected cost of the PaRIS algorithm at
step $t$ is proportional to $N\tilde N \E[\tau_t^{1, \mrp}]$, where $\tilde N$
is a fixed user-chosen parameter. Occasionally, $X_t^1$ falls into an unlikely
region of $\mathbb R^d$ and the acceptance rate becomes low. In other words,
$\tau_t^{1, \mrp}$ is a mixture of geometric distribution, some components of
which might have a large expectation. Unfortunately, these inefficiencies add up
and produce an unbounded execution time in expectation, as shown in the
following proposition.

\begin{prop}
	\label{prop:inf_expectation}
	Under Assumptions~\ref{asp:Ct} and~\ref{asp:space}, the version of Algorithm~\ref{algo:paris_concrete} using the pure rejection sampler satisfies $\E[\tau_t^{1, \mrp}] = \infty$, where $\tau_t^{1,\mrp}$ is defined in \eqref{eq:dist_tau_N}.
\end{prop}

\begin{proof}
	We have
	\begin{align*}
	\E[\tau_t^{1,\mrp}] &= \mbhigh \E\ps{\frac{1}{\sum_n m_t(X_{t-1}^n, X_t^1)
    W_{t-1}^n}} \quad \textrm{via \eqref{eq:dist_tau_N}} \\
	&= \mbhigh \E\ps{\CE{\frac{1}{\sum_n m_t(X_{t-1}^n, X_t^1) W_{t-1}^n}}{\mathcal{F}_{t-1}}} \\
	&= \mbhigh \E\ps{\int_{\mathcal{X}_t} \frac{1}{\sum_n m_t(X_{t-1}^n, x) W_{t-1}^n}\pr{\sum m_t(X_{t-1}^n, x) W_{t-1}^n} \lambda_t(\dd x)} \\
	&= \mbhigh \E\ps{\int_{\mathcal{X}_t} 1 \times \lambda_t(\dd x)} = \infty
    \quad\textrm{by Assumption~\ref{asp:space}}.
	\end{align*}
\end{proof}

In
highly parallel computing architectures, each processor only handles one or a
small number of particles. As such, the heavy-tailed nature of the execution
time means that a few machines might prevent the whole system from moving
forward.
In \textit{all} computing architectures, an execution time without expectation is essentially
unpredictable. A common practice to estimate execution time is to run a certain
algorithm with a small number $N$ of particles, then ``extrapolate'' to the
$N_\mathrm{final}$ of the definitive run. However, as $\E[\tau_t^{1, \mrp}]$ is
infinite for any $N$, it is unclear what kind of information we might get from 
preliminary runs. In  Supplement~\ref{apx:ffbs_exec_result}, besides studying the
execution time of rejection-based implementations of the FFBS algorithm, we
will delve deeper into the difference between the non-parallel and parallel
settings.

From the proof of Proposition~\ref{prop:inf_expectation}, it is clear that the
quantity $\sum_n W_{t-1}^n m_t(X_{t-1}^n, x_t)$ will play a key role in the
upcoming developments. We thus define it formally.
\begin{defi}
	\label{def:rt}
	The true predictive density function $r_t$ and its approximation $r_t^N$ are defined as
	\begin{align*}
	r_t(x_t) &\eqdef \frac{(\Q_{t-1}M_t)(\dd x_t)}{\lambda_t(\dd x_t)} \\
	r_t^N(x_t) &\eqdef \sum W_{t-1}^n m_t(X_{t-1}^n, x_t)
	\end{align*}
	where the first equation is understood in the sense of the Radon-Nikodym derivative and the density $m_{t-1}(x_{t-1}, x_t)$ is defined with respect to the dominating measure $\lambda_t(\dd x_t)$ on $\mathcal X_t$ (cf. \eqref{eq:density_mt}).
\end{defi}

\subsection{Hybrid rejection sampling} \label{subsect:hybrid} 

To solve the aforementioned issues of the pure rejection sampling procedure, we
propose a hybrid rejection sampling scheme. The basic observation is that, for
a single
$m$, direct simulation (e.g.\ via CDF inversion) of $\bw{FFBS}(i_t, \dd
i_{t-1})$ costs $\mathcal{O}(N)$. Thus, once $K = \mathcal O(N)$ rejection
sampling trials have been attempted, one should instead switch to a direct
simulation method. In other words, it does not make sense (at least
asymptotically) to switch to direct sampling after $K$ trials if $K \ll
\mathcal{O}(N)$ or $K \gg \mathcal{O}(N)$. The validity of this method is
established in the following proposition, where we actually allow $K$ to depend
on trials drawn so far. The proof, which is \textit{not} an immediate
consequence of the validity of ordinary rejection sampling, is given in
Supplement~\ref{apx:proof:hybrid_validity}.

\begin{prop}
    \label{prop:hybrid_validity}
    Let $\mu_0(x)$ and $\mu_1(x)$ be two probability densities defined on some
    measurable space $\mathcal X$ with respect to a dominating measure
    $\lambda(\dd x)$. Suppose that there exists $C>0$ such that $\mu_1(x) \leq
    C\mu_0(x)$. Let $(X_1, U_1), (X_2, U_2), \ldots$ be a sequence of i.i.d.
    random variables distributed according to $\mu_0 \otimes
    \operatorname{Unif}[0,1]$ and let $X^* \sim \mu_1$ be independent of that
    sequence. Put
    \[K^* \eqdef \inf \px{n \in \mathbb Z_{\geq 1} \textrm{ such that } U_n
    \leq \frac{\mu_1(X_n)}{C\mu_0(X_n)}} \]
    and let $K$ be \textnormal{any} stopping time with respect to the natural
    filtration associated with the sequence $\px{(X_n, U_n)}_{n=1}^\infty$. Let $Z$
    be defined as $X_{K^*}$ if $K^* \leq K$ and $X^*$ otherwise. Then $Z$ is $\mu_1$-distributed.
\end{prop}

Proposition~\ref{prop:hybrid_validity} thus allows users to pick $K = \alpha N$, where $\alpha > 0$ might be chosen somehow adaptively from earlier trials. In the following, we only consider the simple rule $K = N$, which does not induce any loss of generality in terms of the asymptotic behaviour and is easy to implement. The resulting iteration is described in Algorithm~\ref{algo:hybrid_rejection_sampler}.
\begin{algo}{Hybrid rejection sampler for simulating from $\bw{FFBS}(i_t, \dd i_{t-1})$}
	\label{algo:hybrid_rejection_sampler}
	    \SetKw{KwNot}{not}
	\SetKw{KwBreak}{break}
	\KwIn{Particles $X_{t-1}^{1:N}$ and weights $W_{t-1}^{1:N} $ at time $t-1$; particle $X_t^{i_t}$ at time $t$; constant $\mbhigh$; pre-initialised $\mathcal{O}(1)$ sampler for $\mathcal{M}(W_{t-1}^{1:N})$}
	$accepted \gets \operatorname{False}$ \;
	\For{$i\gets 1$ \KwTo $N$}{
		$\mathcal I_{t-1} \sim \mathcal{M}(W_{t-1}^{1:N})$ using the pre-initialised $\mathcal{O}(1)$ sampler \;
		$U \sim \operatorname{Unif}[0,1]$ \;
		\If{$U \leq m_t(X_{t-1}^{\mathcal I_{t-1}}, X_t^{i_t})/ \mbhigh$}{
			$accepted \gets \operatorname{True}$ \;
			\KwBreak \;
		}
	}
	\If{\KwNot $accepted$}{
		$\mathcal I_{t-1} \sim \mathcal{M}(W_{t-1}^n m(X_{t-1}^n, X_t^{i_t}))$\;
	}
	\KwOut{$\mathcal I_{t-1}$, which is distributed according to $\bw{FFBS}(i_t, \dd i_{t-1})$. }
\end{algo}

When applied in the context of Algorithm~\ref{algo:paris_concrete},
Algorithm~\ref{algo:hybrid_rejection_sampler} gives a smoother of expected
complexity proportional to 
\[N\tilde N \E[\min(\tau_t^{1, \mrp}, N)]
\] 
at time $t$, where $\tau_t^{1, \mrp}$ is defined in \eqref{eq:dist_tau_N}). 
This quantity is no longer infinite, but its growth when $N \to \infty$ might
depend on the model. Still, in all cases, it remains strictly larger than
$\mathcal{O}(N)$ and strictly smaller than $\mathcal{O}(N^2)$. Perhaps more
surprisingly, in linear Gaussian models (see
Supplement~\ref{apx:linear_gaussian_models} for detailed notations), the smoother
is of near-linear complexity (up to log factors). The following two theorems
formalise these claims.

\begin{assumption}
    \label{asp:continuous} The predictive density $r_t$ of $X_t$ given
    $Y_{0:t-1}$ and the potential function $G_t$ are continuous functions on
    $\mathbb{R}^{d_t}$. The transition density $m_t(x_{t-1}, x_t)$ is a
    continuous function on $\mathbb{R}^{d_{t-1}} \times \mathbb{R}^{d_t}$.
\end{assumption}

\begin{thm}
	\label{thm:intermediate_perf}
		Under Assumptions~\ref{asp:Gbound}, \ref{asp:Ct}, \ref{asp:space} and \ref{asp:continuous}, the version of Algorithm~\ref{algo:paris_concrete} using the hybrid rejection sampler (Algorithm~\ref{algo:hybrid_rejection_sampler}) satisfies $\lim_{N \to \infty} \E[\min(\tau_t^{1, \mrp}, N)] = \infty$ and $\lim_{N \to \infty} {\E[\min(\tau_t^{1,\mrp}, N)]}/{N} = 0$, where $\tau_t^{1,\mrp}$ is defined in \eqref{eq:dist_tau_N}.
\end{thm}
\begin{thm}
	\label{thm:near_linear}
	We assume the same setting as Theorem~\ref{thm:intermediate_perf}. In linear Gaussian state space models (Supplement~\ref{apx:linear_gaussian_models}), we have $\E[\min(\tau_t^{1,\mrp}, N)] = \mathcal{O}((\log N)^{d_t/2})$.
\end{thm}
\newcommand{\temptaup}{\tau_t^{1, \mrp}}
\newcommand{\tempxsbt}{{x^2/2}}
While Proposition~\ref{prop:inf_expectation} shows that $\temptaup$ has infinite expectation, Theorem~\ref{thm:near_linear} implies that its $N$-thresholded version only displays a slowly increasing mean. To give a very rough intuition on the phenomenon, consider  $X \sim \mathcal N(0,1)$. Then
\[ \E\ps{e^{X^2/2}} = \int_\mathbb R e^{x^2/2} \frac{e^{-x^2/2}}{\sqrt{2\pi}} = +\infty \]
whereas
\begin{equation}
\label{eq:near_linear_simplified}
	\begin{split}
	\E\ps{\min(e^{X^2/2}, N)} &= \int_\mathbb R \min(e^\tempxsbt,N) \frac{1}{\sqrt{2\pi}} e^{-\tempxsbt} \dd x \\
	&= \int_{\abs{x} \leq \sqrt{2 \log N}} \frac{1}{\sqrt{2\pi}} \dd x + N \int_{\abs{x} > \sqrt{2\log N}} \frac{1}{\sqrt{2\pi}} e^{-\tempxsbt}\dd x \\
	&\leq \sqrt{\frac{4\log N}{\pi}} + \frac{1}{\sqrt{\pi \log N}}
	\end{split}
\end{equation}
using the bound $\P(X > x) \leq \frac{e^{-x^2/2}}{x\sqrt{2\pi}}$ for $x > 0$. The main technical difficulty of the proof of Theorem~\ref{thm:near_linear} (see Supplement~\ref{proof:near_linear}) is to perform this kind of argument under the error induced by the finite sample size particle approximation. In the language of this oversimplified example, we want~\eqref{eq:near_linear_simplified} to hold when $X$ does not follow $\mathcal N(0,1)$ any more, but only an $N$-\textit{dependent} approximation of it.

\section{Efficient backward kernels}
\label{sec:new_kernels}

\subsection{MCMC Backward Kernels}
\label{sec:mcmc_kernel}

This subsection analyses and extends the MCMC backward kernel defined in
Example~\ref{eg:mcmc_smoother}. As we remarked there, the matrix $\bwm{IMH}$ is
not sparse and even has some expensive-to-evaluate entries. We thus reserve it
for use in the off-line smoother (Algorithm~\ref{algo:offline_generic}) whereas
in the on-line scenario (Algorithm~\ref{algo:online_generic}), we use its
PaRIS-like counterpart
\begin{equation}
	\label{eq:mcmc_paris_kernel}
	\bwm{IMHP}[i_t, i_{t-1}] \eqdef \frac{1}{\tilde N} \sum_{\tilde n=1}^{\tilde N} \mathbbm 1\px{i_{t-1} = \tilde J_{t}^{i_t, \tilde n}}
\end{equation}
where $\tilde J_t^{i_t, 1:\tilde N}$ is an independent Metropolis-Hastings
chain started at $J_t^{i_t, 1} \eqdef A_t^{i_t}$, targeting the measure
$\btn{FFBS}(i_t, \dd i_{t-1})$ and using the proposal distribution $\mathcal
M(W_{t-1}^{1:N})$. Thus, the parameter $\tilde N$ signifies that $\tilde N - 1$ MCMC steps are applied to $A_t^{i_t}$, and we shall use the same convention for the kernel $ B_t^{N, \operatorname{IMH}}$. In both cases, the complexity of the corresponding algorithms are $\mathcal O((\tilde N - 1)N)$ which is equivalent to $\mathcal O(N)$ as long as $\tilde N$ remains fixed when $N \to \infty$.

The validity and the stability of $\bwm{IMH}$ and
$\bwm{IMHP}$ are established in the following proposition (proved in
Supplement~\ref{apx:proof:mcmc_validity}). For simplicity, only the case $\tilde
N=2$ is examined, but as a matter of fact, the proposition remains true for
$\tilde N \geq 2$. 

\begin{prop}
	\label{prop:mcmc_validity}
	The kernels $\btn{IMH}$ and $\btn{IMHP}$ with $\tilde N=2$ satisfy the hypotheses of Theorem~\ref{thm:convergence_mcmc} and, under Assumptions~\ref{asp:mt_2ways_bound} and~\ref{asp:g_2ways_bound}, those of Theorem~\ref{thm:stability}. Hence, their respective uses in Algorithms~\ref{algo:offline_generic} and~\ref{algo:online_generic} guarantee a convergent and stable smoother.
\end{prop}

From a theoretical viewpoint, Proposition~\ref{prop:mcmc_validity} is the first result
establishing the stability for the use of MCMC moves inside backward sampling. It relies on technical innovations that we have explained in Section~\ref{sec:stability}, in particular after the statement of Theorem~\ref{thm:stability}.

From a practical viewpoint, the advantages of independent Metropolis-Hastings MCMC kernels compared to the
rejection samplers of Section~\ref{sec:ffbs_kernel} are the dispensability of
specifying an explicit $\mbhigh$ and the deterministic $\mathcal O(N)$ nature of the execution
time. In practice, we observe that the MCMC smoothers are usually 10-20 times faster than the rejection sampling--based counterparts (see e.g. Figure~\ref{fig:lg_online_large_exec}) while producing essentially the same sample quality. Finally, it is not hard to imagine situations where some proposal smarter than
$\mathcal M(W_{t-1}^{1:N})$ would be beneficial. However, we only consider that
one here, mainly because it already performs satisfactorily in our numerical
examples.

\subsection{Dealing with intractable transition densities}
\label{sec:intractable}

\subsubsection{Intuition and formulation} The purpose of backward sampling is to re-generate, for each particle, a new ancestor that is different from that of the filtering step. However, backward sampling is infeasible if the transition density $m_t(x_{t-1}, x_t)$ cannot be calculated. To get around this, we modify the particle filter so that each particle might, in some sense, have two ancestors right from the \textit{forward} pass.

Consider the standard PF (Algorithm~\ref{algo:bootstrap}). Among the $N$ resampled particles $X_{t-1}^{A_t^{1:N}}$, let us track two of them, say $x_{t-1}$ and $x_{t-1}'$ for simplicity. The move step of Algorithm~\ref{algo:bootstrap} will push them through $M_t$ using independent noises, resulting in $x_t$ and $x_t'$ (that is, given $x_{t-1}$ and $x_{t-1}'$, we have $x_t \sim M_t(x_{t-1}, \cdot)$ and $x_{t}' \sim M_t(x_{t-1}', \cdot)$ such that $x_t$ and $x_t'$ are independent). Thus, for e.g.\ linear Gaussian models, we have $\P(x_t = x_t') = 0$. However, if the two simulations $x_t \sim M_t(x_{t-1}, \cdot)$ and $x_t' \sim M_t(x_{t-1}', \cdot)$ are done with specifically correlated noises, it can happen that $\P(x_t = x_t') > 0$. The joint distribution $(x_t, x_t')$ given $(x_{t-1}, x_{t-1}')$ is called a \textit{coupling} of $M_t(x_{t-1}, \cdot)$ and $M_t(x_{t-1}', \cdot)$; the event $x_t = x_t'$ is called the \textit{meeting} event and we say that the coupling is \textit{successful} when it occurs. In that case, the particle $x_t$ automatically has two ancestors $x_{t-1}$ and $x_{t-1}'$ at time $t-1$ without needing any backward sampling.

The precise formulation of the modified forward pass is detailed in
Algorithm~\ref{algo:intractable}. It consists of building in an on-line manner
the backward kernels $\btn{ITR}$ (where ITR stands for ``intractable''). The
main interest of this algorithm lies in the fact that while the function $m_t$
may prove impossible to evaluate, it is usually possible to make $x_t$ and
$x_t'$ meet by correlating somehow the random numbers used in their
simulations. One typical example which this article focuses on is the coupling
of continuous-time processes, but it is useful to keep in mind that
Algorithm~\ref{algo:intractable} is conceptually more general than that.

\begin{algo}{Modified forward pass for smoothing of intractable models (one
    time step)}
	\label{algo:intractable}
	\KwIn{Feynman-Kac model \eqref{eq:fkmodel}, particles $X_{t-1}^{1:N}$ and weights $W_{t-1}^{1:N}$} that approximate the filtering distribution at time $t-1$ \;
	\For{$n \gets 1$ \KwTo $N$}{
		\underline{Resample}. Simulate $(A_t^{n,1}, A_t^{n,2})$ such that marginally each component is distributed according to $\mathcal{M}(W_{t-1}^{1:N})$\;
		\underline{Move}. Simulate $(X_t^{n,1}, X_t^{n,2})$ such that marginally the two components are distributed respectively according to $M_t(X_{t-1}^{A_t^{n,1}}, \dd x_t)$ and $M_t(X_{t-1}^{A_t^{n,2}}, \dd x_t)$\;
		Choose $L \sim \operatorname{Uniform}(\{1, 2\})$ \;
		Set $X_t^n \gets X_t^{n, L}$\;
		\underline{Calculate backward kernel}.\;
		\If{$X_t^{n,1} = X_t^{n,2}$}{
			$\btn{ITR}(n, \dd i_{t-1}) \gets \pr{\delta\px{ A_t^{n,1} } + \delta\px{ A_t^{n,2} }}/2$ \;
		}
		\Else
		{
			$\btn{ITR}(n, \dd i_{t-1}) \gets \delta\px{{A_t^{n,L}}}$  \;
	}
	}
	\reweight
	\KwOut{Particles $X_t^{1:N}$ and weights $W_t^{1:N}$ that approximate the filtering distribution at time $t$; backward kernel $\btn{ITR}$ that can be used in either Algorithm~\ref{algo:offline_generic} or \ref{algo:online_generic}}
\end{algo}

\subsubsection{Validity and stability}

The consistency of Algorithm~\ref{algo:intractable} follows straightforwardly from Theorem~\ref{thm:convergence_mcmc}. To produce a stable routine however, some conditions must be imposed on the couplings $(A_t^{n,1}, A_t^{n, 2})$ and $(X_t^{n,1}, X_t^{n,2})$. We want $A_t^{n,1}$ to be different from $A_t^{n,2}$ as frequently as possible. On the contrary, we aim for a coupling of the two distributions $M_t(X_{t-1}^{A_t^{n, 1}}, \cdot)$ and $M_t(X_{t-1}^{A_t^{n,2}}, \cdot)$ with high success rate so as to maximise the probability that $X_t^{n,1} = X_t^{n,2}$.
\begin{assumption}
	\label{asp:coupling:ancestors}
	There exists an $\epsa > 0$ such that
	$ \P(A_t^{n, 1} \neq A_t^{n, 2} | X_{t-1}^{1:N}) \geq \epsa. $
\end{assumption}
\begin{assumption}
	\label{asp:coupling:dynamics}
	There exists an $\epsd > 0$ such that
	\[\P(X_t^{n,2} = X_t^{n,1} | X_{t-1}^{1:N}, A_t^{n,1}, A_t^{n,2}, X_t^{n,1}) \geq \epsd \pr{1 \land \frac{m_t(\xpt, X_t^{n,1})}{m_t(\xpo, X_t^{n,1})}}. \]
\end{assumption}

The letters A and D in $\epsa$ and $\epsd$ stand for ``ancestors'' and
``dynamics''. Assumption~\ref{asp:coupling:dynamics} means that the user-chosen
coupling of $M_t(\xpo, \cdot)$ and $M_t(\xpt, \cdot)$ must be at least as
$\epsd$ times as efficient as their maximal couplings. For details on this
interpretation, see Proposition~\ref{conditional_proba_maximal_coupling_prop}
in the Supplement. In Lemma~\ref{lem:coupling_efficiency_symmetric}, we also show
that in spite of its appearance, Assumption~\ref{asp:coupling:dynamics} is
actually symmetric with regards to $X_t^{n,1}$ and $X_t^{n,2}$.

We are now ready to state the main theorem of this subsection (see
Supplement~\ref{proof:thm:intractable} for a proof).

\begin{thm}
	\label{thm:intractable}
	The kernels $\btn{ITR}$ generated by Algorithm~\ref{algo:intractable}
    satisfy the hypotheses of Theorem~\ref{thm:convergence_mcmc}. Thus, under
    Assumption~\ref{asp:Gbound}, Algorithm~\ref{algo:intractable} provides a
    consistent smoothing estimate. If, in addition, the Feynman-Kac model
    \eqref{eq:fkmodel} satisfies Assumptions~\ref{asp:mt_2ways_bound} and
    \ref{asp:g_2ways_bound} and the user-chosen couplings satisfy
    Assumptions~\ref{asp:coupling:ancestors} and~\ref{asp:coupling:dynamics},
    the kernels $\btn{ITR}$ also fulfil \eqref{eq:support_cond} and the
    smoothing estimates generated by Algorithm~\ref{algo:intractable} are stable.
\end{thm}

\subsubsection{Good ancestor couplings}
\label{subsect:good_ancestor_couplings}
It is notable that Assumption~\ref{asp:coupling:ancestors} only considers the
event $A_t\eno \neq A_t\ent$, which is a pure index condition that does not
take into account the underlying particles $X_{t-1}^{A_t\eno}$ and
$X_{t-1}^{A_t\ent}$. Indeed, if smoothing algorithms prevent degeneracy by
creating multiple ancestors for a particle, we would expect that their
separation (i.e.\ that they are far away in the state space $\mathcal X_{t-1}$,
e.g.\ $\mathbb R^d$) is critical to the performance. Surprisingly, it is
unnecessary: two very close particles (in $\mathbb R^d$) at time $t-1$ may have
ancestors far away at time $t-2$ thanks to the mixing of the model.

We advise choosing an ancestor coupling  $(A_t\eno, A_t\ent)$ such that the distance between $X_{t-1}^{A_t\eno}$ and $X_{t-1}^{A_t\ent}$ is small. It will then be easier to design a dynamic coupling of $M_t(X_{t-1}^{A_t\eno}, \cdot)$ and $M_t(X_{t-1}^{A_t\ent}, \cdot)$ with a high success rate. Furthermore, simulating the dynamic coupling with two close rather than far away starting points can also take less time  when, for instance, the dynamic involves multiple intermediate steps, but the two processes couple early. One way to achieve an ancestor coupling with the aforementioned property is to first simulate $A_t\eno \sim \mathcal M(W_{t-1}^{1:N})$, then move $A_t\eno$ through an MCMC algorithm which keeps invariant $\mathcal M(W_{t-1}^{1:N})$ and set the result to $A_t\ent$. It suffices to use a proposal looking at indices whose underlying particles are close (in $\mathbb R^d$) to $X_{t-1}^{A_t\eno}$. Finding nearby particles are efficient if they are first sorted using the Hilbert curve, hashed using locality-sensitive hashing or put in a KD-tree \citep[see][for a comprehensive review]{samet2006foundations}. In the context of particle filters, such techniques have been studied for different purposes in \citet{MR3351446}, \citet{jacob2019smoothing} and \citet{sen2018coupling}.

\subsubsection{Conditionally-correlated version} In Algorithm~\ref{algo:intractable}, the ancestor pairs $(A_t\eno, A_t\ent)_{n=1}^N$ are conditionally independent given $\mathcal F_t^-$ and the same holds for the particles $(X_t^n)_{n=1}^N$. These conditional independences allow easier theoretical analysis, in particular, the casting of Algorithm~\ref{algo:intractable} in the framework of Theorems~\ref{thm:convergence_mcmc} and~\ref{thm:stability}. However, they are not optimal for performance in two important ways: (a) they do not allow keeping both $X_t\eno$ and $X_t\ent$ when the two are not equal, and (b) the set of ancestor variables $(A_t^{n,1})_{n=1}^N$ is multinomially resampled from $\px{1, 2, \ldots, N}$ with weights $W_{t-1}^{1:N}$. We know that multinomial resampling is not the ideal scheme, see Supplement~\ref{apx:alternative_resampling_schemes} for discussion.

Consequently, in practice, we shall allow ourselves to break free from
conditional independence. The resulting procedure is described in
Algorithm~\ref{algo:intractable_practice} 
(Supplement~\ref{apx:intractable_pratice}). Despite a lack of rigorous
theoretical support, this is the algorithm that we will use in
Section~\ref{sec:numexp} since it enjoys better performance and it constitutes
a fair comparison with practical implementations of the standard particle
filter, which are mostly based on alternative resampling schemes.

\section{Numerical experiments}
\label{sec:numexp}
\subsection{Linear Gaussian state-space models} \label{numexp_linear_gaussian} Linear Gaussian models constitute a particular class of state space models. They are characterised by Markov dynamics that are Gaussian and observations that are projection of hidden states plus some Gaussian noises. Supplement~\ref{apx:linear_gaussian_models} defines, for different components of these models, the notations that we shall use here. In this section, we consider an instance described in \citet{Guarniero2017}, where the matrix $F_X$ satisfies $F_X[i,j] = \alpha^{1 + |i-j|}$ for some $\alpha$. We consider the problem with $\operatorname{dim}_X = \operatorname{dim}_Y = 2$ and the observations are noisy versions of the hidden states with $C_Y$ being $\sigma_Y^2$ times the identity matrix of size $2$. Unless otherwise specified, we take $\alpha=0.4$ and $\sigma_Y^2 = 0.5$. 

In this section, we focus on the performance of different online smoothers
based on either genealogy tracking, pure/hybrid rejection sampling or MCMC. 
Rejection-based online smoothing amounts to the PaRIS algorithm, for which we
use $\tilde N=2$ for the $\btn{PaRIS}$ kernel. We take $T = 3000$ and simulate
the data from the model. The benchmark additive function is simply
    $\varphi_t(x_{0:t}) = \sum_{s=0}^t x_s(0)$
where $x_s(0)$ is the first coordinate of the $\mathbb R^2$ vector $x_s =
[x_s(0), x_s(1)]$. For a study of offline smoothers (including FFBS), see
Supplement~\ref{numexp_linear_gaussian_offline}. In all programs here and there,
we choose $N=1000$ and use systematic resampling for the forward particle
filters (see section~\ref{apx:alternative_resampling_schemes}). Regarding MCMC
smoothers, we employ the kernels $\btn{IMH}$ or $\btn{IMHP}$ consisting of only
one MCMC step. All results are based on $150$ independent runs.

Although our theoretical results are only proved for the bootstrap filter, we
stress throughout that some of them extend to other filters as well. Therefore,
we will also consider guided particle filters in the simulations. An
introduction to this topic can be found in \citet[Chapter 10.3.2]{SMCbook},
where the expression for the optimal proposal is also provided. In linear
Gaussian models, this proposal is fully tractable and is the one we use. 

To present efficiently the combination of two different filters (bootstrap and
guided) and four different algorithms (naive genealogy tracking, pure/hybrid
rejection and MCMC) we use the following abbreviations: ``B''
for bootstrap, ``G'' for guided, ``N'' for naive genealogy tracking, ``P'' for
pure rejection, ``H'' for hybrid rejection and ``M'' for MCMC. For instance, the
algorithm referred to as ``BM'' uses the bootstrap filter for the forward pass
and the MCMC backward kernels to perform smoothing. Furthermore, the letter
``R'' will refer to the rejection kernel whenever the distinction between pure
rejection and hybrid rejection is not necessary. (Recall that the two rejection
methods produce estimators with the same distribution.)

Figure~\ref{fig:lg_online_small_iq} shows the squared interquartile range for the online smoothing estimates $\Q_t(\varphi_t)$ with respect to $t$. It verifies the rates of Theorem~\ref{thm:stability}, although linear Gaussian models are not strongly mixing in the sense of Assumptions~\ref{asp:mt_2ways_bound} and \ref{asp:g_2ways_bound}: the grid lines hint at a variance growth rate of $\mathcal O(T)$ for the MCMC and reject-based smoothers and of $\mathcal O(T^2)$ for the genealogy tracking ones. Unsurprisingly guided filters have better performance than bootstrap.

\begin{figure}
	\centering
	\includegraphics[scale=0.5]{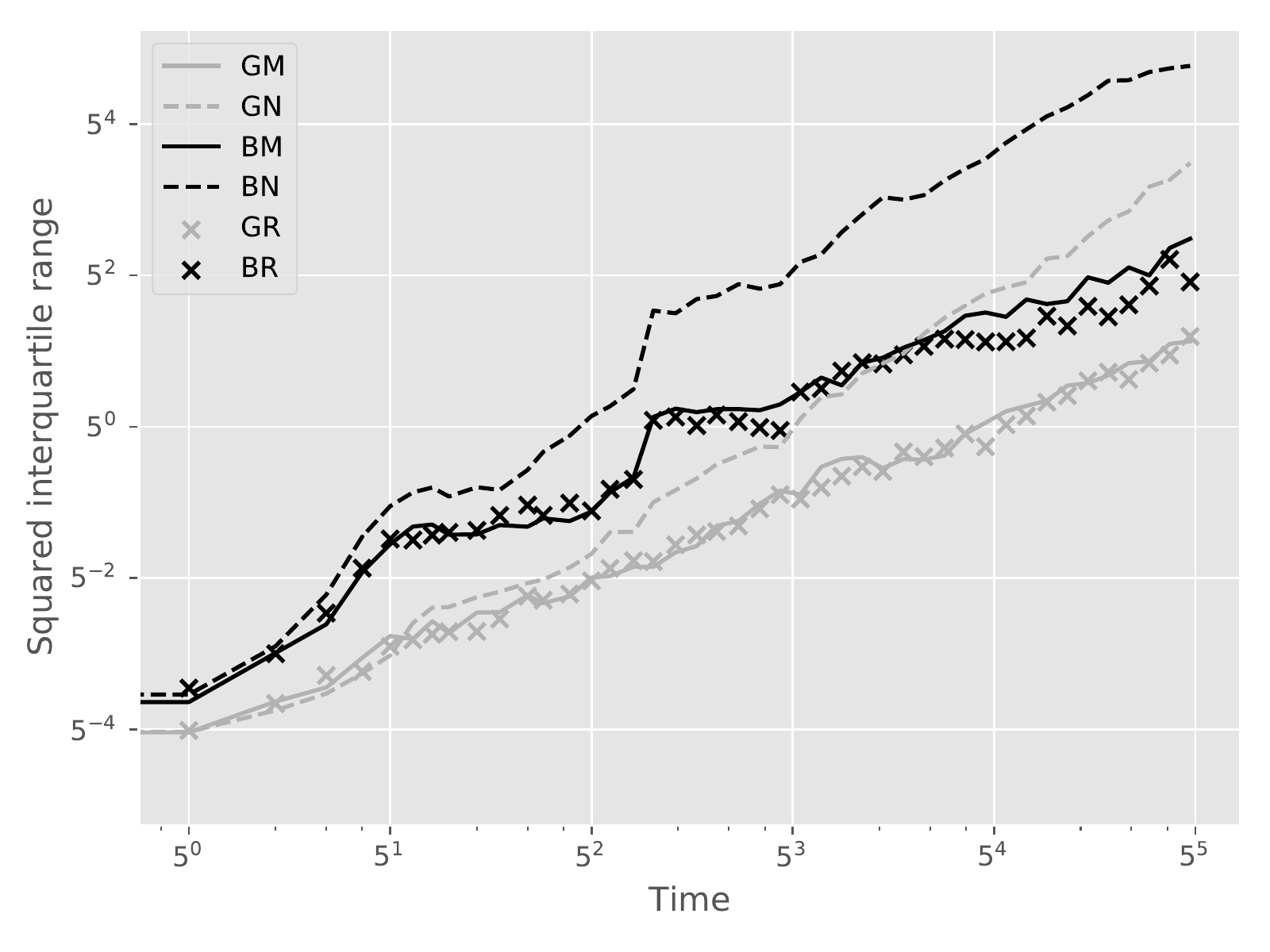}
	\caption{Squared interquartile range of the estimators $\Q_t(\varphi_t)$ with respect to $t$, for different online smoothing algorithms. The model is linear Gaussian with parameters specified in section~\ref{numexp_linear_gaussian}. See text for full explanation of the legend. For readability, the curves are down-sampled to $50$ points before being drawn.}
	\label{fig:lg_online_small_iq}
\end{figure}

Figure~\ref{fig:lg_online_small_exec} show box-plots of the execution
time (divided by $NT$) for different algorithms over $150$ runs. By execution
time, we mean the number of Markov kernel transition density evaluations. We
see that the bootstrap particle filter coupled with pure rejection sampling has
a very heavy-tailed execution time. This behaviour is expected as per
Proposition~\ref{prop:inf_expectation}. Using the guided particle filter seems
to fare better, but Figure~\ref{fig:lg_online_large_exec} (for the same model
but with $\sigma_Y^2 = 2$) makes it clear that this cannot be relied on either.
Overall, these results highlight two fundamental problems with pure rejection
sampling: the computational time has heavy tails and depends on the type of
forward particle filter being used.

\begin{figure}
	\centering
	\includegraphics[scale=0.435]{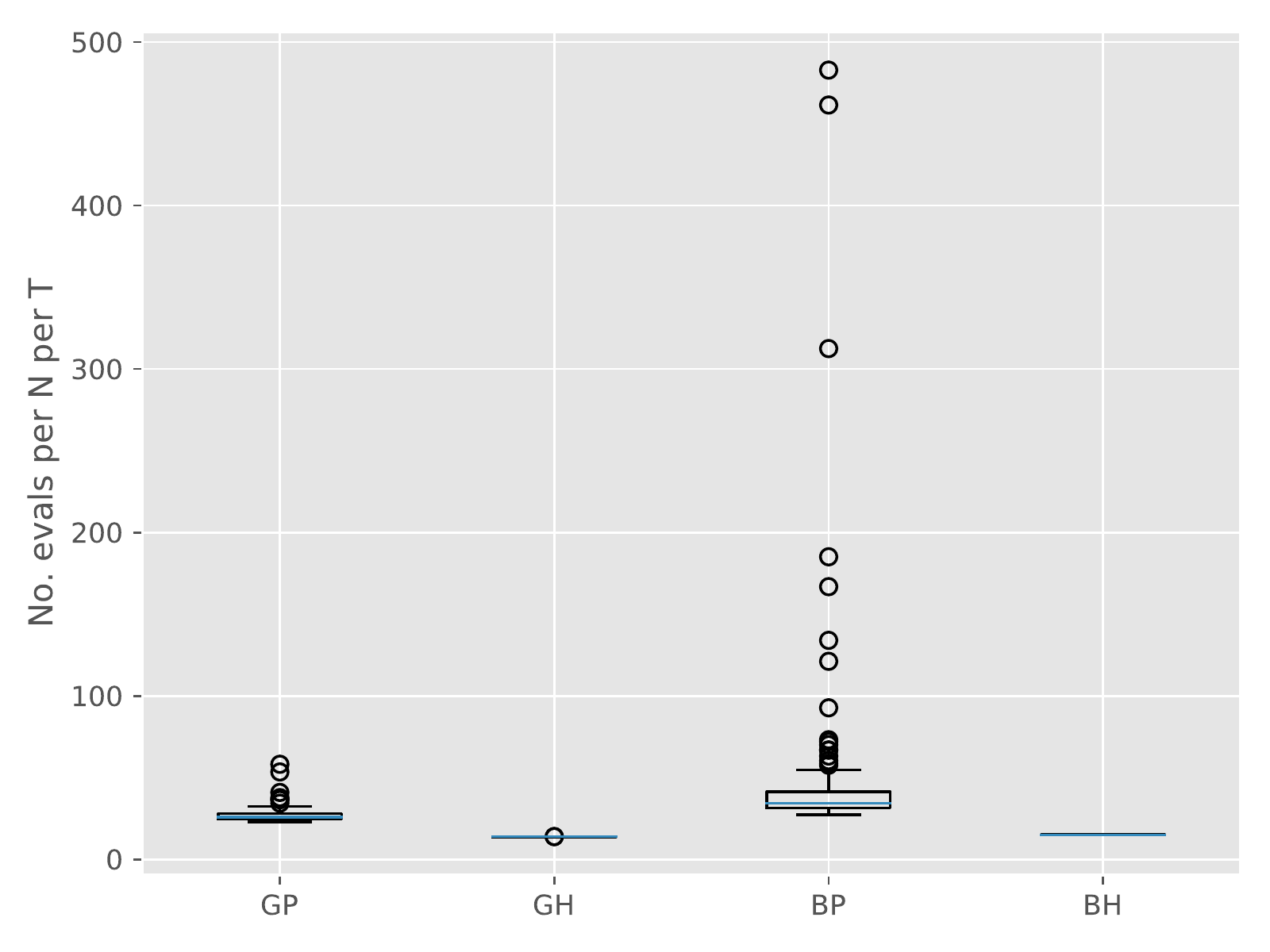}
	\includegraphics[scale=0.435]{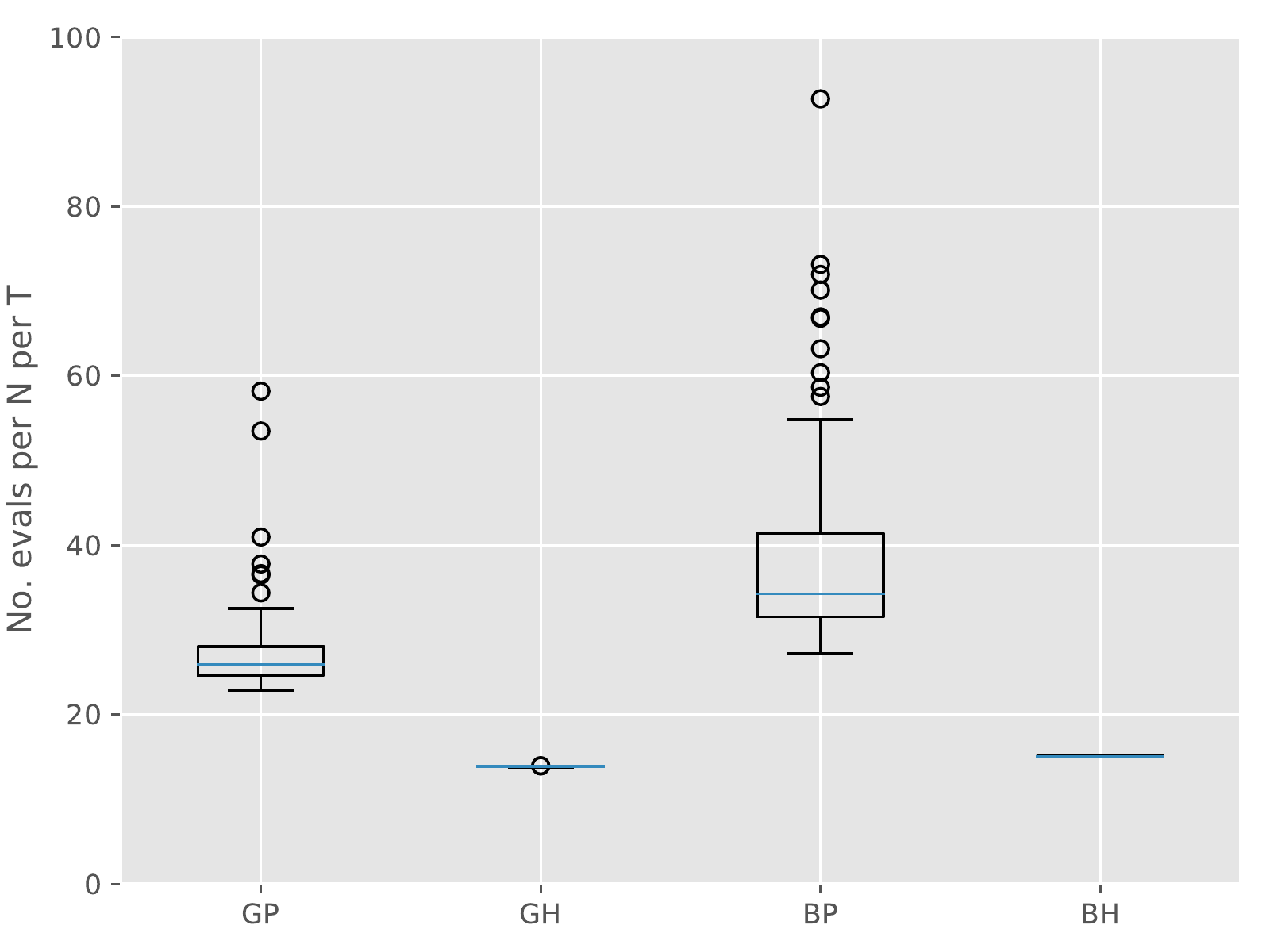}
	\caption{Box plots (based on $150$ runs) of averaged execution times (numbers of transition density evaluations divided by $NT$) for different algorithms on the linear Gaussian model of section~\ref{numexp_linear_gaussian}. Left: original figure, right: zoomed-in version.}
	\label{fig:lg_online_small_exec}
\end{figure}

\begin{figure}
	\centering
	\includegraphics[scale=0.5]{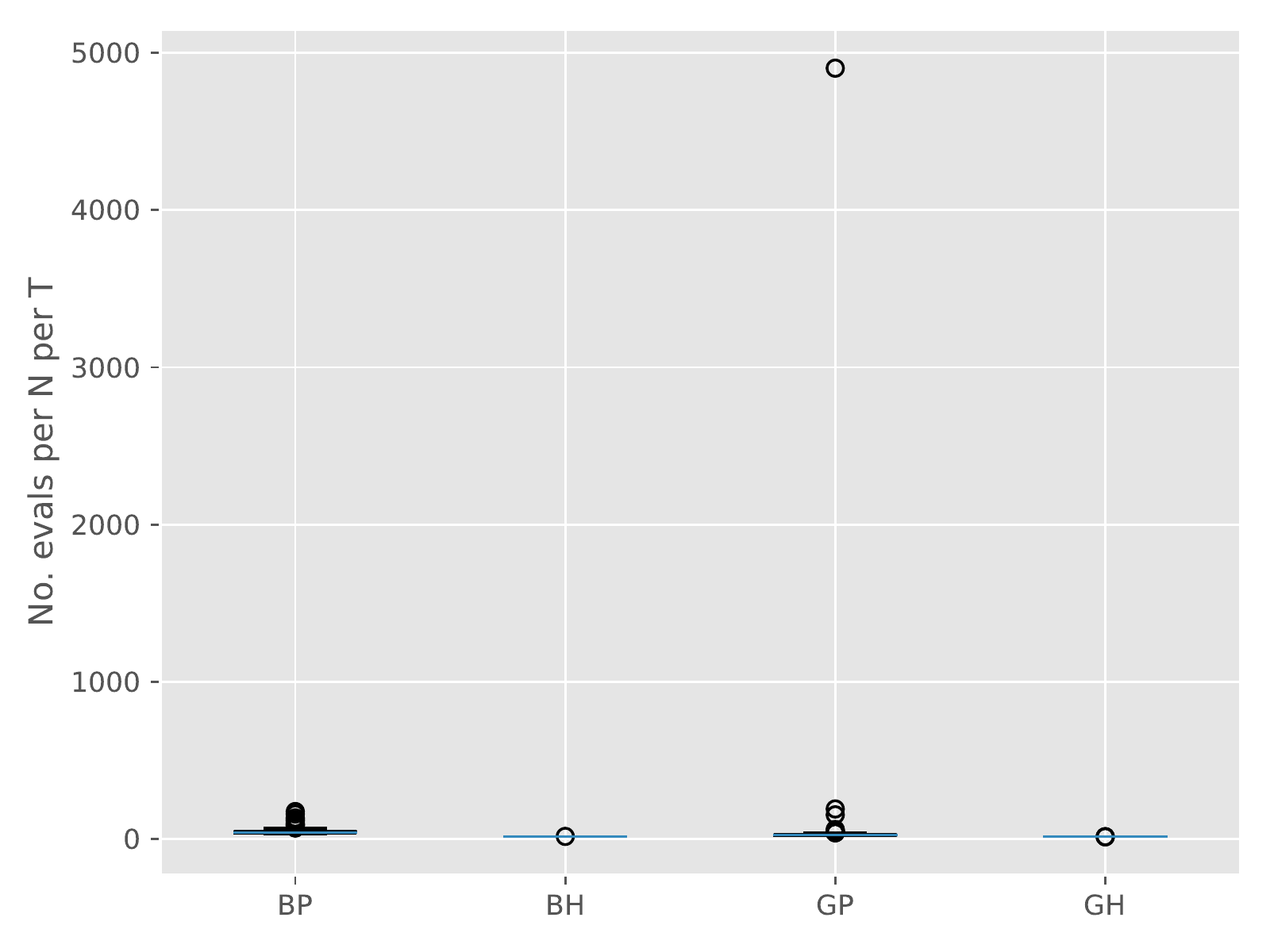}
	\caption{Same as Figure~\ref{fig:lg_online_small_exec}, but for the modified model where $\sigma_Y^2 = 2$.}
	\label{fig:lg_online_large_exec}
\end{figure}

On the other hand, hybrid rejection sampling, despite having random execution time in principle, displays a very consistent number of transition density evaluations over different independent runs. Thus it is safe to say that the algorithm has a virtually deterministic execution time. The catch is that the average computational load (which is around $16$ in Figure~\ref{fig:lg_online_small_exec}) cannot be easily calculated beforehand. In any case, it is much larger than the value $1$ of MCMC smoothers (since only $1$ MCMC step is performed in the kernel $\btn{IMHP}$); whereas the performance (Figure~\ref{fig:lg_online_small_iq}) is comparable.

The bottom line is that MCMC smoothers should be the default option, and one MCMC step seems to be enough. If for some reason one would like to use rejection-based methods, using hybrid rejection is a must.

\subsection{Lotka-Volterra SDE} \label{numexp_lk_sde} Lotka-Voleterra models
(originated in \citealp{lotka1925elements} and
\citealp{volterra1928variations}) describe the population fluctuation of
species due to natural birth and death as well as the consumption of one
species by others. The emblematic case of two species is also known as the
predator-prey model. In this subsection, we study the stochastic differential
equation (SDE) version that appears in \citet{hening2018foodchains}. Let $X_t =
(X_t(0), X_t(1))$ represent respectively the populations of the prey and the
predator at time $t$ and let us consider the dynamics
\begin{equation}
    \label{eq:lotka_sde}
    \left\{
        \begin{aligned}
            \dd X_t(0) &= \bigl[ \beta_0 X_t(0) - \frac 12 \tau_0 [X_t(0)]^2 &&- \tau_1 X_t(0) X_t(1) \bigr] \dd t + X_t(0) \dd E_t(0) \\
            \dd X_t(1) &= \bigl[ \qquad \quad -\beta_1 X_t(1) &&+ \tau_1 X_t(0) X_t(1) \bigr] \dd t + X_t(1) \dd E_t(1)
        \end{aligned}
    \right.
\end{equation}
where $E_t = \Gamma W_t$ with $W_t$ being the standard Brownian motion in
$\mathbb R^2$ and $\Gamma$ being some $2 \times 2$ matrix. The parameters
$\beta_0$ and $\beta_1$ are the natural birth rate of the prey and death rate
of the predator. The predator interacts with (eats) the prey at rate $\tau_1$.
The quantity $\tau_0$ encodes intra-species competition in the prey population.
The $\frac 12$ in its parametrisation is to line up with the Lotka Volterra
jump process in $\mathbb Z^2$ where the population sizes are integers and the interaction term
becomes $\tau_0 X_t(0) [X_t(0) - 1]/2$.

The state space model is comprised of the process $X_t$ and its noisy
observations $Y_t$ recorded at integer times. The Markov dynamics cannot be
simulated exactly, but can be approximated through (Euler) discretisation. Nevertheless, the Euler transition density $m_t^\mathrm E(x_{t-1}, x_t)$ remains intractable (unless the step size is exactly $1$). Thus, the algorithms presented in Subsection~\ref{sec:intractable} are useful. The missing bit is a method to efficiently couple $m_t^\mathrm E(x_{t-1}, \cdot)$ and $m_t^\mathrm E(x_{t-1}', \cdot)$, which we carefully describe in Supplement~\ref{apx:coupling_euler}.

We consider the model with $\tau_0 = 1/800$, $\tau_1 = 1/400$, $\beta_0 = 0.3125$ and $\beta_1 = 0.25$. The matrix $\Gamma$ is such that the covariance matrix of $E_1$ is $\begin{bmatrix}
1/100 & 1/200 \\
1/200 & 1/100
\end{bmatrix}$. The observations are recorded on the log scale with Gaussian error of covariance matrix $\begin{bmatrix}
0.04 & 0.02\\
0.02 & 0.04
\end{bmatrix}$. The distribution of $X_0$ is two-dimensional normal with mean $[100, 100]$ and covariance matrix $\begin{bmatrix}
100 & 50\\
50 & 100
\end{bmatrix}$. This choice is motivated by the fact that the preceding parameters give the stationary population vector $[100, 100]$. According to \citet{hening2018foodchains}, they also guarantee that neither animal goes extinct almost surely as $t \to \infty$.

By discretising \eqref{eq:lotka_sde} with time step $\delta = 1$, one can get some very rough intuition on the dynamics. For instance, per second there are about $31$ preys born. Approximately the same number die (to maintain equilibrium), of which $6$ die due to internal competition and $25$ are eaten by the predator. The duration between two recorded observations corresponds more or less to one-third generation of the prey and one-fourth generation of the predator. The standard deviation of the variation due to environmental noise is about $10$ individuals per observation period, for each animal.

Again, these intuitions are highly approximate. For readers wishing to get more familiar with the model, Supplement~\ref{more_figures_sde} contains real plots of the states and the observations; as well as data on the performance of different smoothing algorithms for moderate values of $T$. We now showcase the results obtained in a large scale problem where $T=3000$ and the data is simulated from the model.

We consider the additive function 
$ \varphi_t(x_{0:t}) \eqdef \sum_{s=0}^t \ps{x_s(0) - 100}$. 
Figure~\ref{fig:sde_big_final_box} represents using box
plots the distributions of the estimators for $\Q_T(\varphi_T)$ using either
the genealogy tracking smoother (with systematic resampling; see
Supplement~\ref{apx:alternative_resampling_schemes}) or
Algorithm~\ref{algo:intractable_practice}. Our proposed smoother greatly
reduces the variance, at a computational cost which is empirically $1.5$ to $2$
times greater than the naive method. Since we used Hilbert curve to design good
ancestor couplings (see Section~\ref{subsect:good_ancestor_couplings}),
coupling of the dynamics succeeds $80 \%$ of the time. As discussed in the
aforementioned section, starting two diffusion dynamics from nearby points make
them couple earlier, which reduces the computational load afterwards.

\begin{figure}
	\centering
	\includegraphics[scale=0.5]{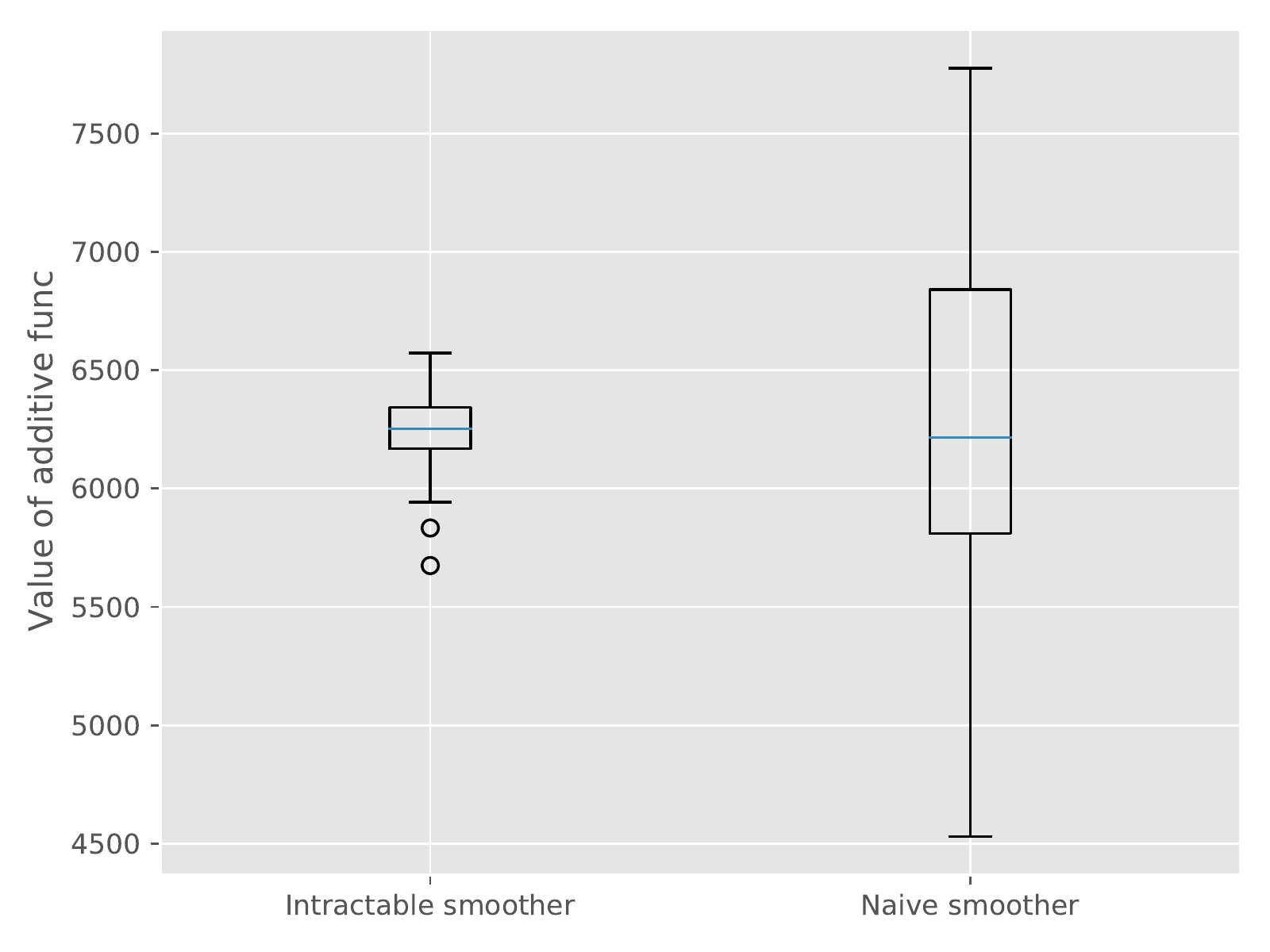}
	\caption{Box plot of estimators (over $50$ independent runs with $N=1000$ particles) for $\Q_T(\varphi_T)$ in the Lotka-Volterra SDE model with $T = 3000$. They are calculated using either the naive genealogy tracking smoother or our smoother developed for intractable models (Algorithm~\ref{algo:intractable_practice}).}
	\label{fig:sde_big_final_box}
\end{figure}

Figure~\ref{fig:sde_big_iq} plots with respect to $t$ the squared interquartile range of the two methods for the estimation of $\Q_t(\varphi_t)$. Grid lines hint at a quadratic growth for the genealogy tracking smoother (as analysed in \citealp[Sect. 1]{Olsson2017}) and a linear growth for the kernel $\btn{ITRC}$ (as described in Theorem~\ref{thm:stability}).

\begin{figure}
	\centering
	\includegraphics[scale=0.5]{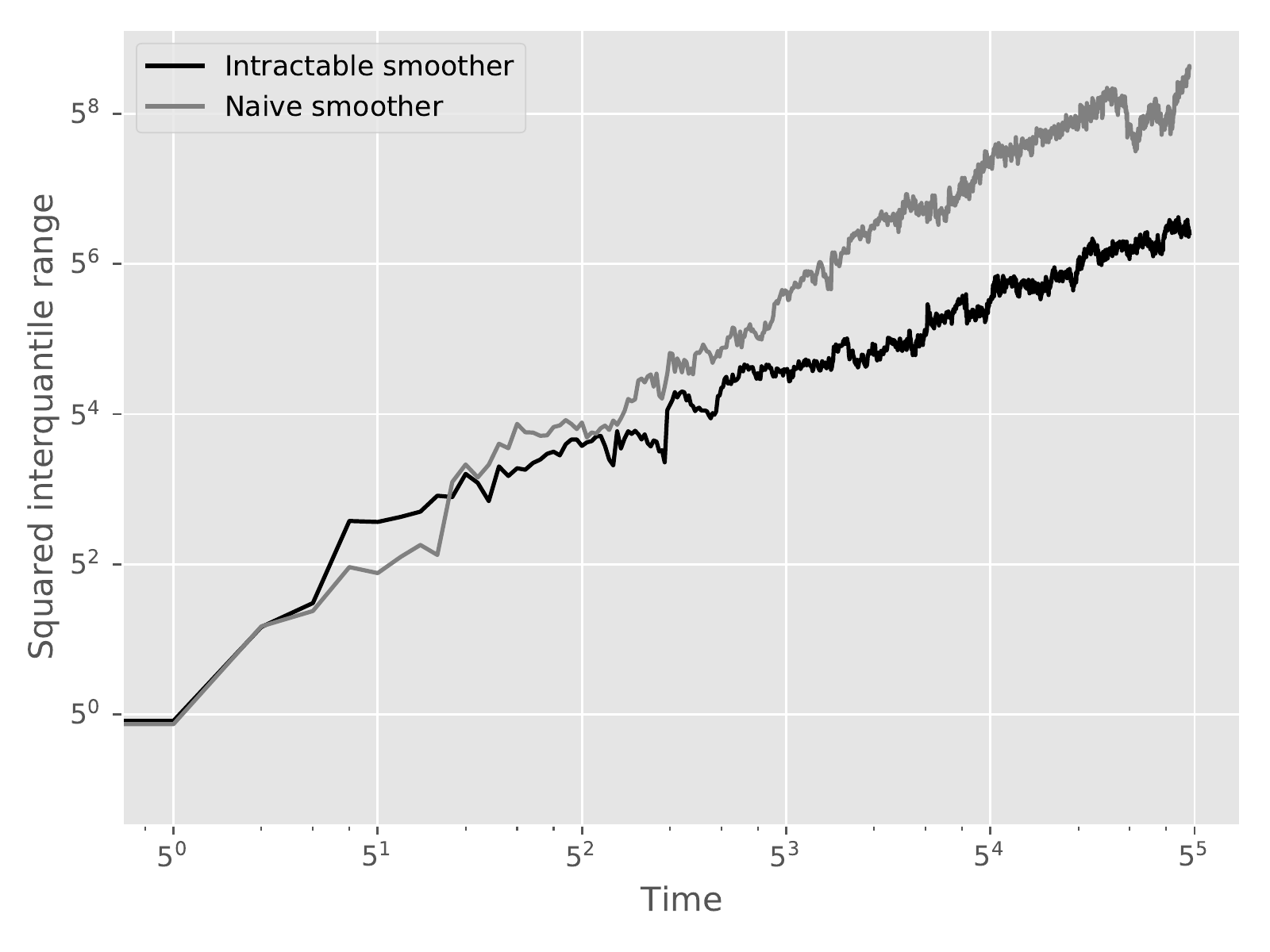}
	\caption{Squared interquartile range for the genealogy tracking smoother and our proposed one. Same context as in Figure~\ref{fig:sde_big_final_box}}
	\label{fig:sde_big_iq}
\end{figure}

Finally, Figure~\ref{fig:sde_big_ess} (Supplement~\ref{more_figures_sde}) shows properties of the effective sample size (ESS) ratio for this model. In a nutshell, while being globally stable (between $40\%$ and $70\%$), it has a tendency to drift towards near $0$ from time to time due to unusual data points. At these moments, resampling kills most of the particles and aggravates the degeneracy problem for the naive smoother. As we have seen in the above figures, systematic resampling is not enough to mitigate this in the long run.

\section{Conclusion}
\label{sec:conclusion}
\subsection{Practical recommendations}

Our first recommendation does not concern the smoothing algorithm per se. It is of
paramount importance that the particle filter used in in the preliminary
filtering step performs reasonably well, since its output defines the support
of the approximations generated by the subsequent smoothing algorithm. 
(Standard recommendations to obtain good performance from a particle filter
are to increase $N$, or to use better proposal distributions, or both.)

When the  transition density is tractable, we recommend the MCMC smoother by
default (rather than even the standard, $\bigO(N^2)$ approach). It has a
deterministic, $\bigO(N)$ complexity, it does not require the transition
density to be bounded, and it seems to perform well even with one or two MCMC
steps.  If one still wants to use the rejection smoother instead,  it is safe
to say that there is no reason not to use the hybrid method. 

Although the assumptions under which we prove the stability of the smoothing
estimates are strong, the general message still holds. The Markov kernel and
the potential functions must make the model forget its past in some ways.
Otherwise, we get an unstable model for which no smoothing methods can work.
The rejection sampling -- based smoothing algorithms can therefore serve as the
ultimate test. Since they simulate exactly independent trajectories given the
skeleton, there is no hope to perform better, unless one switches to another
family of smoothing algorithms.

For intractable models, the key issue is to design couplings with high meeting
probability. Fortunately, the inherent chaos of the model makes it possible to
choose two very close starting points for the dynamics and thus easy to obtain
a reasonable meeting probability. If further difficulties persist, there is a practical (and very heuristic)
recipe to test whether one coupling of $M_t(x, \cdot)$ and $M_t(x', \cdot)$ is
close to optimal. It consists in approximating $M_t(x, \cdot)$ and $M_t(x',
\cdot)$ by Gaussian distributions and deduce the optimal coupling rate from
their total variation distance. There is no closed formula for the total
variation distance between two Gaussian distributions in high dimensions.
However, it can be reliably estimated using the geometric interpretation of the
total variation distance being one minus the area of the intersection created
by the corresponding density graphs. In this way, one can get a very rough idea
of to what extent a certain coupling realises the meeting potential that the
two distributions have. If the coupling seems good and the trajectories still
look degenerate, it can very well be that the model is unstable.

\subsection{Further directions}
The major limitation of our work is the exclusive theoretical analysis under the bootstrap particle filter. Moreover, we require that the $N$ new particles generated at step $t$ are conditionally independent given previous particles at time $t-1$. This excludes practical optimisations like systematic resampling and  Algorithm~\ref{algo:intractable_practice}. Finally, the backward sampling step is also used in other algorithms (in
particular Particle Markov Chain Monte Carlo, see \citealp{PMCMC}) and it would
be interesting to see to what extent our techniques can be applied there.

\subsection{Data and code}
The code used to run numerical experiments is available at
\url{https://github.com/hai-dang-dau/backward-samplers-code}. Some of the algorithms
are already available in an experimental branch of the \verb+particles+ Python
package at \url{https://github.com/nchopin/particles}.

\section*{Acknowledgements}

The first author acknowledges a CREST PhD scholarship via AMX funding, and
wish to thank the members of his PhD defense jury (St\'ephanie Allassoni\`ere,
Randal Douc, Arnaud Doucet, Anthony Lee, Pierre del Moral, Christian Robert)
for helpful comments on the corresponding chapter in his thesis. 

We also thank Adrien Corenflos, Samuel Duffield, the associate  editor, 
and the referees for their comments on a preliminary version of the paper.
% The
% second author acknowledges partial support from Labex Ecodec (Ecodec/ANR-11-
% LABX-0047).

\bibliography{complete.bib,new_entries.bib}
\bibliographystyle{apalike}

\appendix
\newpage
\renewcommand{\appendixname}{Supplement}

\section{Additional notations}
\label{sec:apx_notations}
This section defines new notations that do not appear in the main text (except notations for linear Gaussian models) but are used in the Supplement.

\subsection{Linear Gaussian models} \label{apx:linear_gaussian_models} 

Let $\dimx$ and $\dimy$ be two strictly positive integers and $F_X$ and $F_Y$
be two full-rank matrices of sizes $\dimx \times \dimx$ and $\dimy \times
\dimx$ respectively. Let $C_X$ and $C_Y$ be two symmetric positive definite
matrices of respective sizes $\dimx \times \dimx$ and $\dimy \times \dimy$. A
linear Gaussian state space model has the underlying Markov process defined by
\[ X_t | X_{0:t-1} \sim \mathcal N(F_X X_{t-1}, C_X), \] 
where $X_0$ also
follows a Gaussian distribution; and admits the observation process 
\[Y_t | X_t \sim \mathcal N(F_Y X_t, C_Y). \] 
The predictive ($X_t$ given $Y_{0:t-1}$),
filtering ($X_t$ given $Y_{0:t}$) and smoothing ($X_t$ given $Y_{0:T}$)
distributions are all Gaussian and their parameters can be explicitly
calculated via recurrence formulas \citep{Kalman1960, KalBuc}. We shall denote
their respective mean vectors and covariance matrices by $(\mu_t^\mathrm{pred},
\Sigma_t^\mathrm{pred})$,
$(\mu_t^\mathrm{filt}, \Sigma_t^\mathrm{filt})$ and $(\mu_t^\mathrm{smth},
\Sigma_t^\mathrm{smth})$. In particular, the starting distribution $X_0$ is
$\mathcal N(\mu_0^\mathrm{pred}, \Sigma_0^\mathrm{pred})$.

\subsection{Total variation distance}\label{apx:tv}

Let $\mu$ and $\nu$ be two probability measures on $\mathcal X$. The total variation distance between $\mu$ and $\nu$, sometimes also denoted $\operatorname{TV}(\mu, \nu)$, is defined as $\norm{\mu - \nu}_{\operatorname{TV}}:= \sup_{f: \mathcal X \to [0,1]} \abs{\mu(f) - \nu(f)}$. The definition remains valid if $f$ is restricted to the class of indicator functions on measurable subsets of $\mathcal X$. It implies in particular that $\abs{\mu(f) - \nu(f)} \leq \norm{f}_{\mathrm{osc}}\operatorname{TV}(\mu, \nu)$.

Next, we state a lemma summarising basic properties of the total variation distance and defining coupling-related notions (see, e.g.\ Proposition 3 and formula (13) of \citet{roberts2004general}). While the last property (covariance bound) is not in the aforementioned reference and does not seem popular in the literature, its proof is straightforward and therefore omitted.
\begin{lem}
	\label{lem:properties_TV}
	The total variation distance has the following properties:
	\begin{itemize}
		\item (Alternative expressions.) If $\mu$ and $\nu$ admit densities
            $f(x)$ and $g(x)$ respectively with reference to a dominating measure $\lambda$, we have
		\[\operatorname{TV}(\mu, \nu) = \frac 12 \int \abs{f(x)-g(x)} \lambda (\dd x) = 1 - \int \min(f(x), g(x)) \lambda (\dd x). \]
		\item (Coupling inequality \& maximal coupling.) For any pair of random variables $(M, N)$ such that $M\sim \mu$ and $N \sim \nu$, we have
		\[\P(M \neq N) \geq \operatorname{TV}(\mu, \nu). \]
		There exist pairs $(M^*, N^*)$ for which equality holds. They are called \textnormal{maximal couplings} of $\mu$ and $\nu$.
		\item (Contraction property.) Let $(X_n)$ be a Markov chain with invariant measure $\mu^\star$. Then
		\[\operatorname{TV}(X_n, \mu^*) \geq \operatorname{TV}(X_{n+1}, \mu^*). \]
		\item (Covariance bound.) For any pair of random variables $(M,N)$ such that $M \sim \mu$ and $N \sim \nu$ and real-valued functions $h_1$ and $h_2$, we have
		\[\abs{\Cov(h_1(M), h_2(N))} \leq 2\infnorm{h_1} \infnorm{h_2} \operatorname{TV}\pr{(M,N), \mu \otimes \nu}. \]
	\end{itemize}
\end{lem}

\subsection{Cost-to-go function}
In the context of the Feynman-Kac model \eqref{eq:fkmodel}, define the associated \textit{cost-to-go} function $H_{t:T}$ as (see e.g.\ \citet[Chapter 5]{SMCbook})
\begin{equation}
\label{eq:def-cost-to-go}
H_{t:T}(x_t) \eqdef \prod_{s=t+1}^T M_{s-1}(x_{s-1}, \dd x_s) G_s(x_s).
\end{equation}
This function bridges $\Q_t(\dd x_t)$ and $\Q_T(\dd x_t)$, since $\Q_T(\dd x_t) \propto \Q_t(\dd x_t) H_{t:T}(x_t)$.

\subsection{The projection kernel}
Let $\mathcal X$ and $\mathcal Y$ be two measurable spaces. The projection kernel $\Projeg$ is defined by 
\[\Projeg\pr{(x,y), \dd x^*} \eqdef \delta_x(\dd x^*). \] 
In particular, for any function $g: \mathcal X \to \mathbb R$ and measure $\mu(\dd x, \dd y)$ defined on $\mathcal X \times \mathcal Y$, we have
\begin{align*}
(\Projeg g)(x,y) &= g(x) \\
(\mu \Projeg)(g) &= \iint g(x) \mu(\dd x, \dd y) = \int g(x)\mu(\dd x)
\end{align*}
where the second identity shows the marginalising action of $\Projeg$ on $\mu$. In the context of state space models, we define the shorthand 
\[\Proj^{0:T}_t \eqdef \Proj^{(\mathcal X_0, \ldots, \mathcal X_T)}_{\mathcal X_t}.\]

\subsection{Other notations}\label{sec:apx_other_notations}

For a real number $x$, let $\lfloor x \rfloor$ be the largest integer not
exceeding $x$. The mapping $x \mapsto \lfloor x \rfloor$ is called the floor
function\notsep The Gamma function $\Gamma(a)$ is defined for $a > 0$ and is
given by $\Gamma(a)\eqdef \int_{\mathbb R_+} e^{-x} x^{a-1} \dd x$\notsep Let
$\mathcal X$ and $\mathcal Y$ be two measurable spaces. Let $K(x, \dd y)$ be a
(not necessarily probability) kernel from $\mathcal X$ to $\mathcal Y$. The
norm of $K$ is defined by
$\infnorm{K}:= \sup_{f: \mathcal X \to \mathcal Y, f \neq 0} \infnorm{Kf}/\infnorm{f}$. In particular, for any function $f:\mathcal X \to \mathcal Y$, we have $\infnorm{Kf} \leq \infnorm{K} \infnorm{f}$\notsep Let $X_n$ be a sequence of random variables. We say that $X_n = \bigOproba (1)$ if for any $\varepsilon > 0$, there exists $M > 0$ and $N_0$, both depending on $\varepsilon$, such that $\P(|X_n| \geq M) \leq \varepsilon$ for all $n \geq N_0$. For a strictly positive deterministic sequence $a_n$, we say that $X_n = \bigOproba(a_n)$ if $X_n/a_n = \bigOproba(1)$. See \citet{JansonBigO} for discussions\notsep We use the notation $\mathcal N(x | \mu, \Sigma)$ to refer to the value at $x$ of the density function of the normal distribution $\mathcal N(\mu, \Sigma)$\notsep Let $f:  U \to V$ be a function from some space $U$ to another space $V$. Let $S$ be a subset of $U$. The restriction of $f$ to S, written $f|_S$, is the function from $S$ to $V$ defined by $f|_S(x) = f(x)$, $\forall x \in S$.

\section{FFBS complexity for different rejection schemes}\label{apx:ffbs_exec_result}

\subsection{Framework and notations}\label{apx:ffbs_exec_result_framework}

 The FFBS algorithm is a particular instance of Algorithm~\ref{algo:offline_generic} where $\btn{FFBS}$ kernels are used. If backward simulation is done using pure rejection sampling (Algorithm~\ref{algo:pure_rejection_sampler}), the computational cost to simulate the $t-1$-th index of the $n$-th trajectory has conditional distribution
\begin{equation}
	\label{eq:cond_dist_tau_ffbs}
	\tau_t^{n, \mrffbs} |\ \mathcal F_T, \mathcal I_{t:T}^n \sim \operatorname{Geo}\pr{\frac{\sum_i W_{t-1}^i m_t(X_{t-1}^i, X_t^{\mci_t^n})}{\mbhigh}}.
\end{equation}
At this point, it would be useful to compare this formula with \eqref{eq:dist_tau_N} of the PaRIS algorithm. The difference is subtle but will drive interesting changes to the way rejection-based FFBS behaves.

If hybrid rejection sampling (Algorithm~\ref{algo:hybrid_rejection_sampler}) is
to be used instead, we are interested in the distribution of $\min(\tau_t^{n,
\mrffbs}, N)$, for reasons discussed in Subsection~\ref{subsect:hybrid}. In a
highly parallel setting, it is preferable that the distribution of
\textit{individual} execution times, i.e.\ $\tau_t^{n, \mrffbs}$ or
$\min(\tau_t^{n, \mrffbs},N)$, are not heavy-tailed. In contrast, for
non-parallel hardware, only \textit{cumulative} execution times, i.e.
$\sum_{n=1}^N \tau_{t}^{n, \mrffbs}$ or $\sum_{n=1}^N \min(\tau_t^{n, \mrffbs},
N)$, matter. Even though the individual times might behave badly, the
cumulative times could be much more regular thanks to effect of the central
limit theorem, whenever applicable. Nevertheless, studying the finiteness of
the $k$-th order moment of $\tau_t^{1, \mrffbs}$ is still a good way to get
information about both types of execution times, since it automatically implies
$k$-th order moment (in)finiteness for both of them.

\subsection{Execution time for pure rejection sampling}

We show that under certain circumstances, the execution time of the pure
rejection procedure has infinite expectation. Proposition 1 in \citet{Douc2011}
hints that the cost per trajectory for FFBS-reject might tend to infinity when
$N \to \infty$. In contrast, we show that infinite expectation might very well
happen for \textit{finite} sample sizes. We first give the statement for
general state space models, then focus on their implications for Gaussian ones.
In particular, while infinite expectations occur only under certain
configurations, infinite higher moments happen in \textit{all} linear Gaussian
models with non-degenerate dynamics.
\begin{thm}\label{thm:ffbs_exec_infinite_general} Using the setting and
    notations of Supplement~\ref{apx:ffbs_exec_result_framework}, under
    Assumptions~\ref{asp:Gbound} and~\ref{asp:Ct} , we have $\E[\tau_t^{1,
    \mrffbs}] = \infty$ whenever 
    \[\int_{\mathcal X_t} G_t(x_t) H_{t:T}(x_t) \lambda_t(\dd x_t) = \infty \]
    where the cost-to-go function $H_{t:T}$ is defined in
    \eqref{eq:def-cost-to-go} and the measure $\lambda_t$ is defined in
    \eqref{eq:density_mt}.
\end{thm}
\begin{thm}\label{thm:ffbs_exec_infinite_gaussian} Using the setting and
    notations of Supplement~\ref{apx:ffbs_exec_result_framework}, we consider
    linear Gaussian models and their notations defined in
    Supplement~\ref{apx:linear_gaussian_models}. Then we have $\E[(\tau_t^{1,
    \mrffbs})^k] = \infty$ whenever $k$ is greater than a certain $k_0$ being
    the smallest eigenvalue of the matrix $\operatorname{Id} +
    \covx^{1/2}\pr{\invp{\sigmasmth_t} - \invp{\sigmapred_t}}\covx^{1/2}$.
\end{thm}
The proofs of the two assertions are given in
Supplement~\ref{ap:proof:ffbs_exec_infinite}. We now look at how they are
manifested in concrete examples. The first remark is that for technical
reasons, Theorem~\ref{thm:ffbs_exec_infinite_gaussian} gives no information on
the finiteness of $E[(\tau_t^{1, \mrffbs})^k]$ for $k=1$ (since $k_0$ is
already greater than or equal to $1$ by definition). To study the finiteness of
$\E[\tau_t^{1, \mrffbs}]$, we thus turn to
Theorem~\ref{thm:ffbs_exec_infinite_general}.

\begin{example}
	In linear Gaussian models, the integral of Theorem~\ref{thm:ffbs_exec_infinite_general} is equal to
	\[\int \mathcal N(y_t|F_Y x_t, \covy) \prod_{s=t+1}^T \mathcal N(x_s|F_X x_{s-1}, \covx) \mathcal N(y_s|F_Y x_s, \covy) \dd x_{t:T} \]
	where the notation $\mathcal N(\mu, \Sigma)$ refers to the density of the normal distribution. The integrand is proportional to $\exp[-0.5 (Q(x_{t:T}) - R(x_{t:T}))]$ for some quadratic form $Q(x_{t:T})$ and linear form $R(x_{t:T})$. The integral is finite if and only if $Q$ is positive definite. In our case, this means that there is no non-trivial root for the equation $Q(x_{t:T})=0$, which is equivalent to
	\begin{equation*}
		\begin{cases}
		F_Y x_s &= 0, \forall s = t, \ldots T \\
		F_X x_{s-1} &= x_s, \forall s = t+1, \ldots, T.
		\end{cases}
	\end{equation*}
	Put another way, $\E[\tau_t^{1, \mrffbs}]$ is infinite whenever the intersection
	\[\bigcap_{k=0}^{T-t} \operatorname{Ker}(F_Y F_X^k) = \bigcap_{k=0}^{T-t} F_X^{-k}(\operatorname{Ker}(F_Y)) \]
	contains other things than the zero vector. A common and particularly troublesome situation is when $F_X = c \operatorname{Id}$ for some $c > 0$ (but $\covx$ can be arbitrary) and the dimension of the states ($\dimx$) is greater than that of the observations ($\dimy$). Then the above intersection remains non-trivial no matter how big $T-t$ is. Thus, $\E[\tau_t^{1, \mrffbs}]$ has no expectation for any $t$. In general, the problem is less severe as successive intersections will shrink the space quickly to $\{0\}$. Consequently, Theorem~\ref{thm:ffbs_exec_infinite_general} only points out infiniteness of $\E[\tau_t^{1, \mrffbs}]$ for $t$ close to $T$. The bad news however will come from higher moments, as seen in the below example.
\end{example}
We will now focus on a simple but particularly striking example. Our purpose here is to illustrate the concepts as well as to show that their implications are relevant even in small, familiar settings. More advanced scenarios are presented in Section~\ref{sec:numexp} devoted to numerical experiments.
\begin{example}
	\label{eg:simple_ffbs}
	We consider two one-dimensional Gaussian state-space models: they both have
    $F_X = 0.5$, $\covx = 1$, $X_0 \sim \mathcal N(0, \covx^2/(1-F_X^2))$ and
    $T=3$. The only difference between them is that one has
    $\sigma_y^2\eqdef\covy= 0.5^2$ and another has $ \sigma_y^2 = 3^2$. We are
    interested in the execution times $\tau_1^{n, \mrffbs}$ at time $t=1$ (i.e.
    the rejection-based simulation of indices $\mathcal I_0^n$ at time $t=0$).
    Theorem~\ref{thm:ffbs_exec_infinite_gaussian} then gives $k_0 \approx 1.14$
    for $\sigma_y = 3$ and $k_0 \approx 5$ for $\sigma_y = 0.5$. The first
    implication is that in both cases, $\tau_1^{n, \mrffbs}$ is a heavy-tailed
    random variable and therefore FFBS-reject is not a viable option in a
    highly parallel setting. But an interesting phenomenon happens in the
    sequential hardware scenario where one is rather interested in the
    cumulative execution time, i.e.\ $\sum_{n=1}^N \tau_1^{n, \mrffbs}$, or
    equivalently, the \textnormal{mean} number of trials per
    \textnormal{particle}. In the $\sigma_y = 3$ case, non-existence of
    \text{second} moment prevents the cumulative regularisation effect of the
    central limit theorem. This is not the case for $\sigma_y = 0.5$, in which
    the cumulative execution time actually behaves nicely
    (Figures~\ref{fig:small_eg_ffbs_reject}
    and~\ref{fig:small_eg_ffbs_reject_closeup}). However, the most valuable
    message from this example is perhaps that the performance of FFBS-reject
    depends in a non-trivial (hard to predict) way on the model parameters. 
\end{example}
\begin{figure}
	\centering
	\includegraphics[scale=0.50]{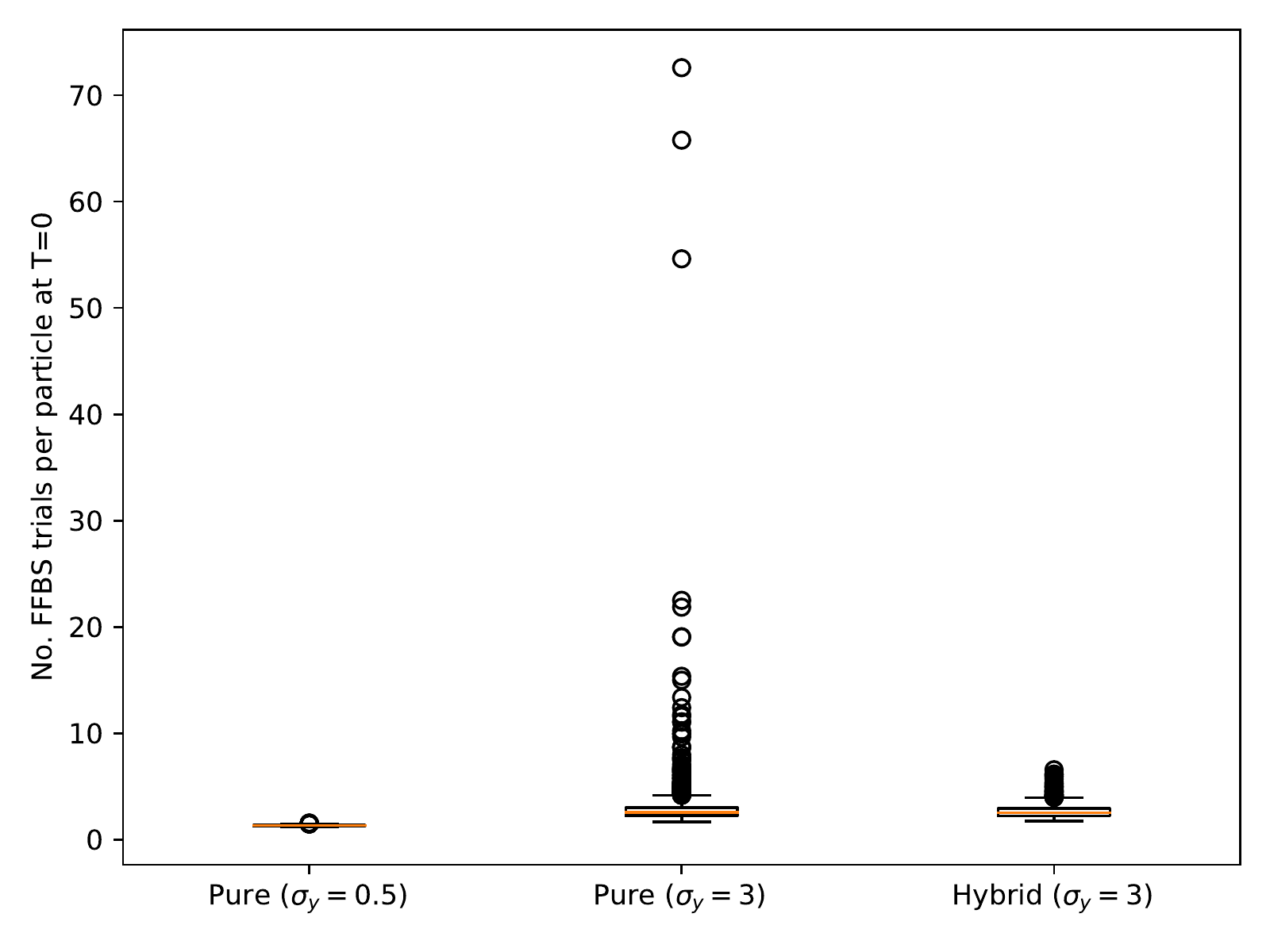}
	\caption{Box plots for the mean number of trials per particles to simulate indices at time $0$, for models described in Example~\ref{eg:simple_ffbs} and for FFBS algorithms based on pure and hybrid rejection sampling. The figure is obtained by running bootstrap particle filters with $N=500$ over $1500$ independent executions.}
	\label{fig:small_eg_ffbs_reject}
\end{figure}
\begin{figure}
	\centering
	\includegraphics[scale=0.50]{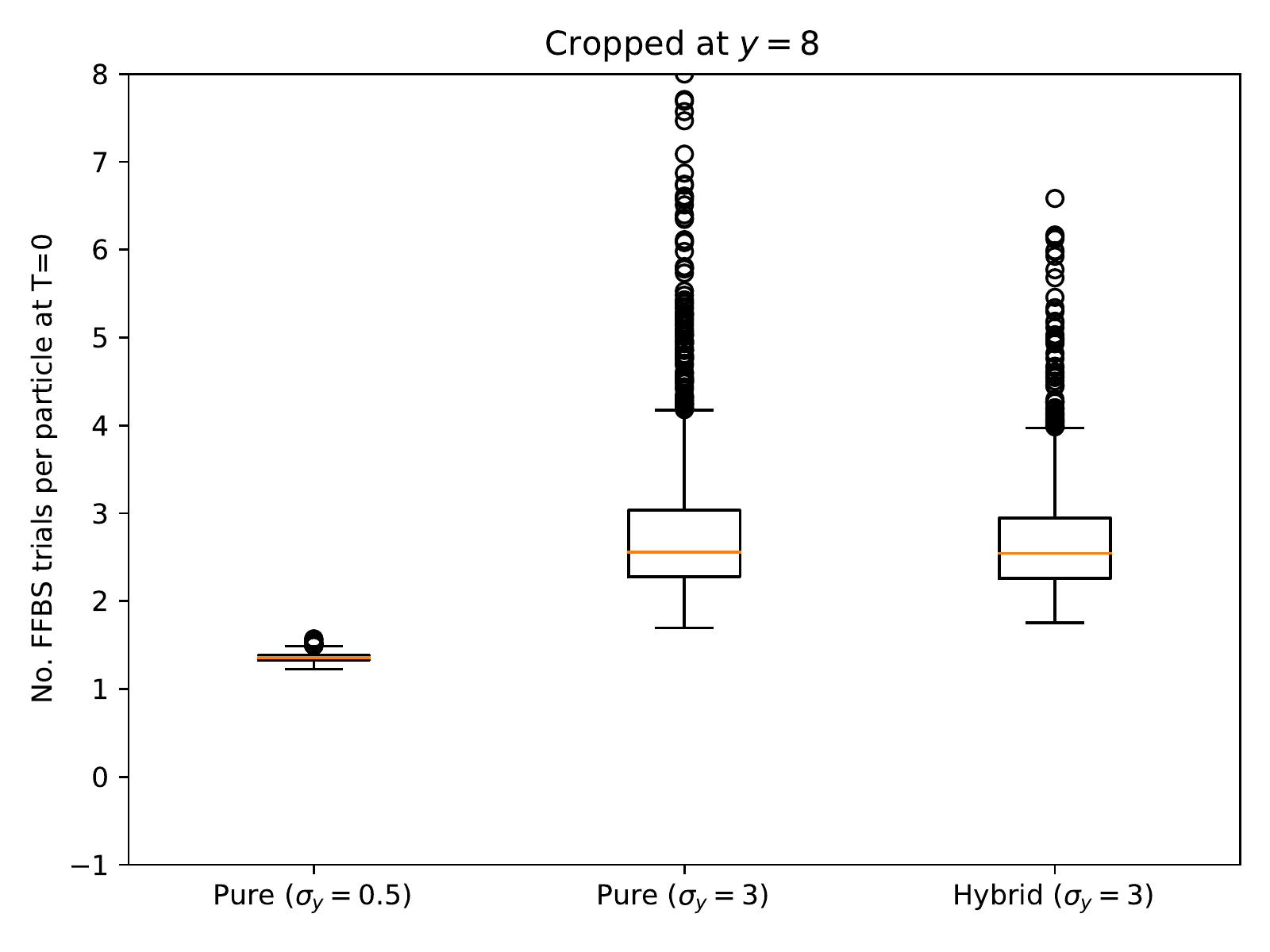}
	\caption{Zoom of Figure~\ref{fig:small_eg_ffbs_reject} to $0\leq y\leq8$}
	\label{fig:small_eg_ffbs_reject_closeup}
\end{figure}

\subsection{Execution time for hybrid rejection sampling}
Formula \eqref{eq:cond_dist_tau_ffbs} suggests defining the limit distribution $\tauinfffbs$ as
\[ \tauinfffbs\ | \ \xinfffbs \sim \operatorname{Geo}\pr{\frac{r_t(\xinfffbs)}{\mbhigh}}\]
where $\xinfffbs \sim \Q_T(\dd x_t)$ and $r_t$ given in
Definition~\ref{def:rt}. These quantities provide the following
characterisation of the cumulative execution time for the hybrid FFBS algorithm (proved in Section~\ref{ap:proof:collective_ffbs_hybrid}).
\begin{thm}
	\label{thm:ffbs_hybrid_exec}
	Under Assumptions~\ref{asp:Gbound} and~\ref{asp:Ct} and the setting of Section~\ref{apx:ffbs_exec_result_framework}, we have
	\[\frac{\sum_{n=1}^N\min(\tau_t^{n, \mrffbs}, N)}{N} = \bigOproba \pr{\E[\min(\tauinfffbs, N)]} \]
	where the notation $\bigOproba$ is defined in Supplement~\ref{sec:apx_notations}.
\end{thm}
This theorem admits the following corollary for linear Gaussian models (also proved in Section~\ref{ap:proof:collective_ffbs_hybrid}).
\begin{corollary}
	\label{cor:gaussian_hybrid_ffbs}
	For linear Gaussian models (Supplement~\ref{apx:linear_gaussian_models}), if smoothing is performed using the hybrid rejection version of the FFBS algorithm, the mean execution time per particle at time step $t$ is $\bigOproba(\log^{d_t/2} N)$ where $d_t$ is the dimension of $\mathcal X_t$.
\end{corollary}
The bound $\bigOproba(\log^{d_t/2} N)$ is actually quite conservative. For instance, with either $\sigma_y=0.5$ or $\sigma_y=3$, the model considered in Example~\ref{eg:simple_ffbs} admits $\E[\tauinfffbs] < \infty$. (Gaussian dynamics can be handled using exact analytic calculations and enables to verify the claim straightforwardly.) Theorem~\ref{thm:ffbs_hybrid_exec} then gives an execution time per particle of order $\bigOproba(1)$ for hybrid FFBS, which is better than the $\bigOproba(\sqrt{\log N})$ predicted by Corollary~\ref{cor:gaussian_hybrid_ffbs}. Yet another unsatisfactory point of the result is its failure to make sense of the spectacular improvement brought by hybrid rejection sampling over the ordinary procedure in the $\sigma_y=3$ case (see Figure~\ref{fig:small_eg_ffbs_reject}). As explained in Example~\ref{eg:simple_ffbs}, this is connected to the variance of $\E[\tau_t^{1, \mrffbs}]$ and not merely the expectation; so a study of second order properties of $N^{-1}\sum_n \min(\tau_t^{n, \mrffbs}, N)$ would be desirable.

\section{Conditionally-correlated versions of particle algorithms}
\label{apx:intractable_pratice}
\subsection{Alternative resampling schemes}
\label{apx:alternative_resampling_schemes}
In Algorithm~\ref{algo:bootstrap}, the indices $A_t^{1:N}$ are drawn conditionally i.i.d. from the multinomial distribution $\mathcal M(W_{t-1}^{1:N})$. They satisfy
\[\CE{\sum_{j=1}^N \ind_{A_t^j = i}}{\mathcal F_{t-1}} = NW_{t-1}^i \]
for any $i = 1, \ldots, N$. There are other ways to generate $A_t^{1:N}$ from $W_{t-1}^{1:N}$ that still verify this identity. We call them \textit{unbiased resampling schemes}, and the natural one used in Algorithm~\ref{algo:bootstrap} \textit{multinomial resampling}.

The main motivation for alternative resampling schemes is performance. We refer
to \cite{Chopin:CLT, Douc2005, gerber2019negative_association} for more
details, but would like to mention that the theoretical studies of particle
algorithms using other resampling schemes are more complicated since
$X_{t}^{1:N}$ are no longer i.i.d. given $\mathcal F_{t-1}$. We use systematic
resampling \citep{CarClifFearn} in our experiments. See Algorithm~\ref{algo:systematic_resampling} for a succinct description and \citet[Chapter 9]{SMCbook} for efficient implementations in $\mathcal O(N)$ running time.

\begin{algo}{Systematic resampling}\label{algo:systematic_resampling}
	\KwIn{Weights $W_{t-1}^{1:N}$ summing to $1$}
	Generate $U \sim \operatorname{Uniform}[0,1]$\;
	\For{$n \gets 1$ \KwTo $N$}{
		Set $A_t^n$ to the unique index $k$ satisfying
		\[W_1 + \cdots + W_{k-1} \leq \frac{n-1+u}{N} < W_1 + \cdots + W_k \]
	}
	\KwOut{Resampled indices $A_t^{1:N}$}
\end{algo}

\subsection{Conditionally-correlated version of Algorithm~\ref{algo:intractable}}

In this part, we present an alternative version of
Algorithm~\ref{algo:intractable} that does not create conditionally i.i.d.\
particles at each time step. The procedure is detailed in
Algorithm~\ref{algo:intractable_practice}. It creates on the fly backward
kernels $\btn{ITRC}$ (for ``intractable, conditionally correlated''). It
involves a resampling step which can be done in principle using any unbiased
resampling scheme. Following the intuitions of
Subsection~\ref{subsect:good_ancestor_couplings} and the notations of
Algorithm~\ref{algo:intractable_practice}, we want a scheme such that in most
cases, $A_t^{2k-1} \neq A_t^{2k}$ but the Euclidean distance between
$X_{t-1}^{A_t^{2k-1}}$ and $X_{t-1}^{A_t^{2k}}$ is small.
Algorithm~\ref{algo:adjacent_resampler} proposes such a method (which we name
the Adjacent Resampler). It can run in $\mathcal O(N)$ time using a suitably
implemented linked list.

\begin{algo}{Conditionally-correlated version of Algorithm~\ref{algo:intractable}}
	\label{algo:intractable_practice}
	\KwIn{Feynman-Kac model~\eqref{eq:fkmodel}, particles $X_{t-1}^{1:N}$ and weights $W_{t-1}^{1:N}$ that approximate $\Q_{t-1}(\dd x_{t-1})$}
	\underline{Resample} $A_t^{1:N}$ from $\px{1, 2, \ldots, N}$ with weights $W_{t-1}^{1:N}$ using any resampling scheme (such as the Adjacent Resampler in Algorithm~\ref{algo:adjacent_resampler}) \;
	\For{$k \gets 1$ \KwTo $N/2$}{
		\underline{Move}. Simulate $X_t^{2k-1}$ and $X_t^{2k}$ such that marginally, $X_t^{2k-1} \sim M_t(X_{t-1}^{A_t^{2k-1}}, \cdot)$ and $X_t^{2k} \sim M_t(X_{t-1}^{A_t^{2k}}, \cdot)$\;
		\underline{Calculate backward kernel}.\;
		\If{$X_t^{2k-1} = X_t^{2k}$}{
			Set $\btn{ITRC}(2k-1, \cdot) \gets \pr{\delta\px{A_t^{2k-1}} + \delta\px{A_t^{2k}}}/2$\;
			Set $\btn{ITRC}(2k, \cdot) \gets \pr{\delta\px{A_t^{2k-1}} + \delta\px{A_t^{2k}}}/2$
		}
		\Else{
			Set $\btn{ITRC}(2k-1, \cdot) \gets \delta\px{A_t^{2k-1}}$\;
			Set $\btn{ITRC}(2k, \cdot) \gets \delta\px{A_t^{2k}}$\;
		}
	}
	\reweight
	\KwOut{Particles $X_t^{1:N}$ and weights $W_t^{1:N}$ that approximate $\Q_t(\dd x_t)$; backward kernel $\btn{ITRC}$ for use in Algorithms~\ref{algo:offline_generic} and~\ref{algo:online_generic}}
\end{algo}

\begin{algo}{The Adjacent Resampler}
	\label{algo:adjacent_resampler}
	\KwIn{Particles $X_{t-1}^{1:N}$, weights $W_{t-1}^{1:N}$}
	Sort the particles $X_{t-1}^{1:N}$ using the Hilbert curve. Let $s\gets [s_1 \ldots s_N]$ be the corresponding \textit{indices}\;
	Resample from $\px{1, \ldots, N}$ with weights $W_{t-1}^{1:N}$ using systematic resampling \citep{CarClifFearn, gerber2019negative_association}, then let $f: \px{1,\ldots, N} \to \mathbb Z$ be the function defined by $f(i)$ being the number of times the index $s_i$ was resampled. Obviously $\sum_{i=1}^N f(i) = N$\;
	Initialise $i \gets 1$\;
	\For{$n\gets 1$ \KwTo $N$}{
		Set $A_t^n \gets s_i$\;
		Update $f(i) \gets f(i) - 1$\;
		Let $\Omega_1$ be the set $\px{\min\px{\ell > i \mid f_\ell > 0}}$ (which has one element if the minimum is well-defined and zero element otherwise)\;
		Let $\Omega_2$ be the set $\px{\max\px{\ell < i \mid f_\ell > 0}}$ (which has one element if the maximum is well-defined and zero element otherwise)\;
		If $\Omega_1 \cup \Omega_2$ is not empty, update $i \gets \operatorname{argmax} f|_{\Omega_1 \cup \Omega_2}$ (see section~\ref{sec:apx_other_notations} for the restriction notation). If there is more than one argmax, pick one randomly\;
	}
	\KwOut{Resampled indices $A_t^{1:N}$}
\end{algo}

\section{Additional information on numerical experiments}

\subsection{Offline smoothing in linear Gaussian models}\label{numexp_linear_gaussian_offline}

In this section, we study offline smoothing for the linear Gaussian model
specified in Section~\ref{numexp_linear_gaussian}. Since offline processing
requires storing particles at all times $t$ in the memory, we use $T=500$ here
instead of $T=3000$. Apart from that, the algorithmic and benchmark settings
remain the same.

Figure~\ref{fig:lg_offline_small_rate} plots the squared interquartile range of
the estimators $\Q_T(\varphi_t)$ with respect to $t$, for different algorithms.
For small $t$, the function $\varphi_t$ only looks at states close to time $0$,
whereas for bigger $t$, recent states less affected by degeneracy are also
taken into account. In all cases though, we see that MCMC and rejection-based
smoothers have superior performance.

\begin{figure}
	\centering
	\includegraphics[scale=0.5]{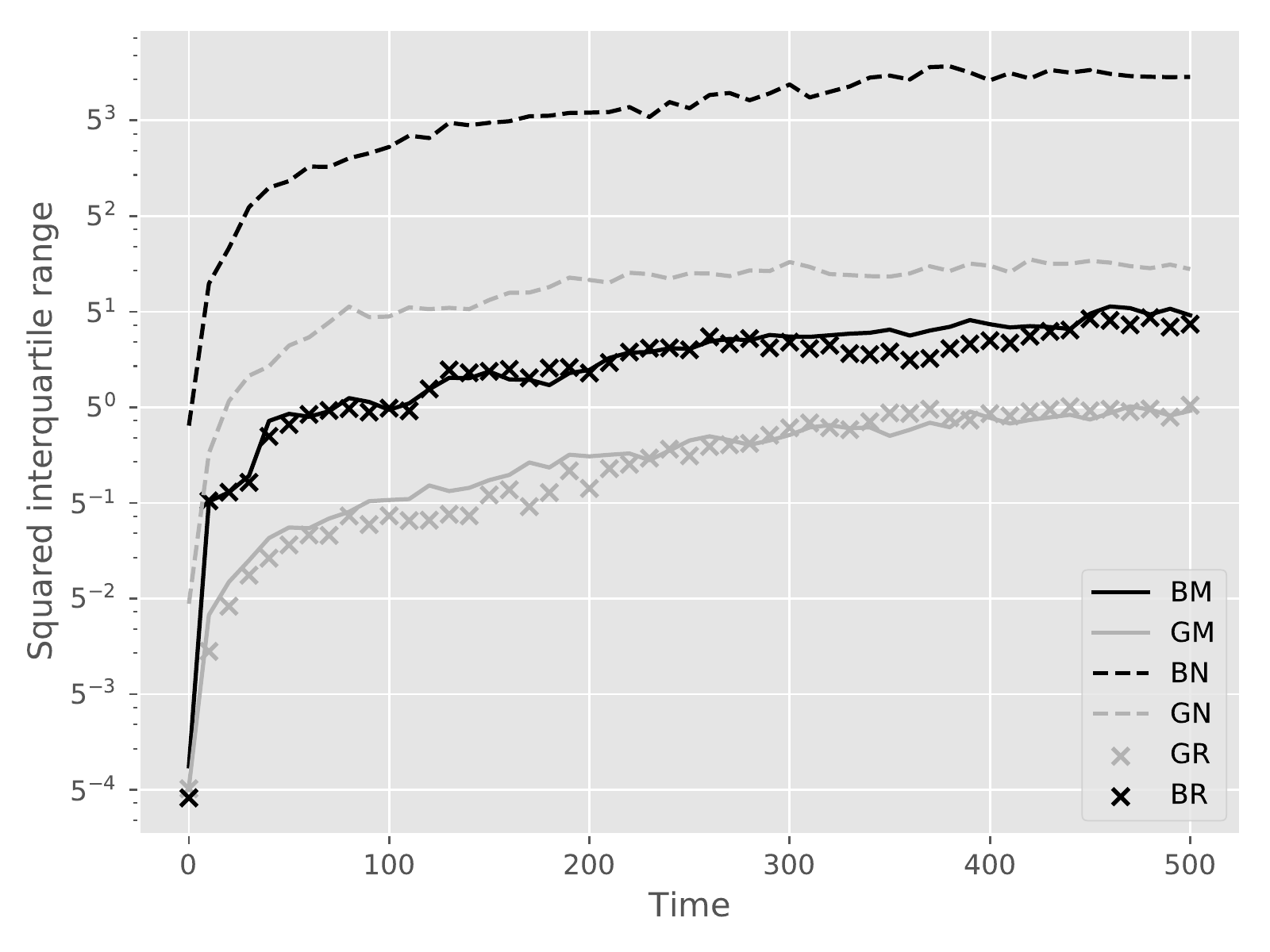}
	\caption{Squared interquartile range of the estimators of $\Q_T(\varphi_t)$ with respect 
    to $t$, for different algorithms applied to the model of 
    Section~\ref{numexp_linear_gaussian_offline}. See Section~\ref{numexp_linear_gaussian} 
   for the meaning of the acronyms in the legend.}
	\label{fig:lg_offline_small_rate}
\end{figure}

Figure~\ref{fig:lg_offline_small_exec} shows box plots of the averaged execution times (per particle $N$ per time $t$) based on $150$ runs. The observations are comparable to those in Section~\ref{numexp_linear_gaussian}. We see a performance difference between the rejection-based smoothers using the bootstrap and the guided filters. Both have an execution time that is much more variable than hybrid rejection algorithms. The latter still need around $10$ times more CPU load than MCMC smoothers, for essentially the same precision.

\begin{figure}
	\centering
	\includegraphics[scale=0.5]{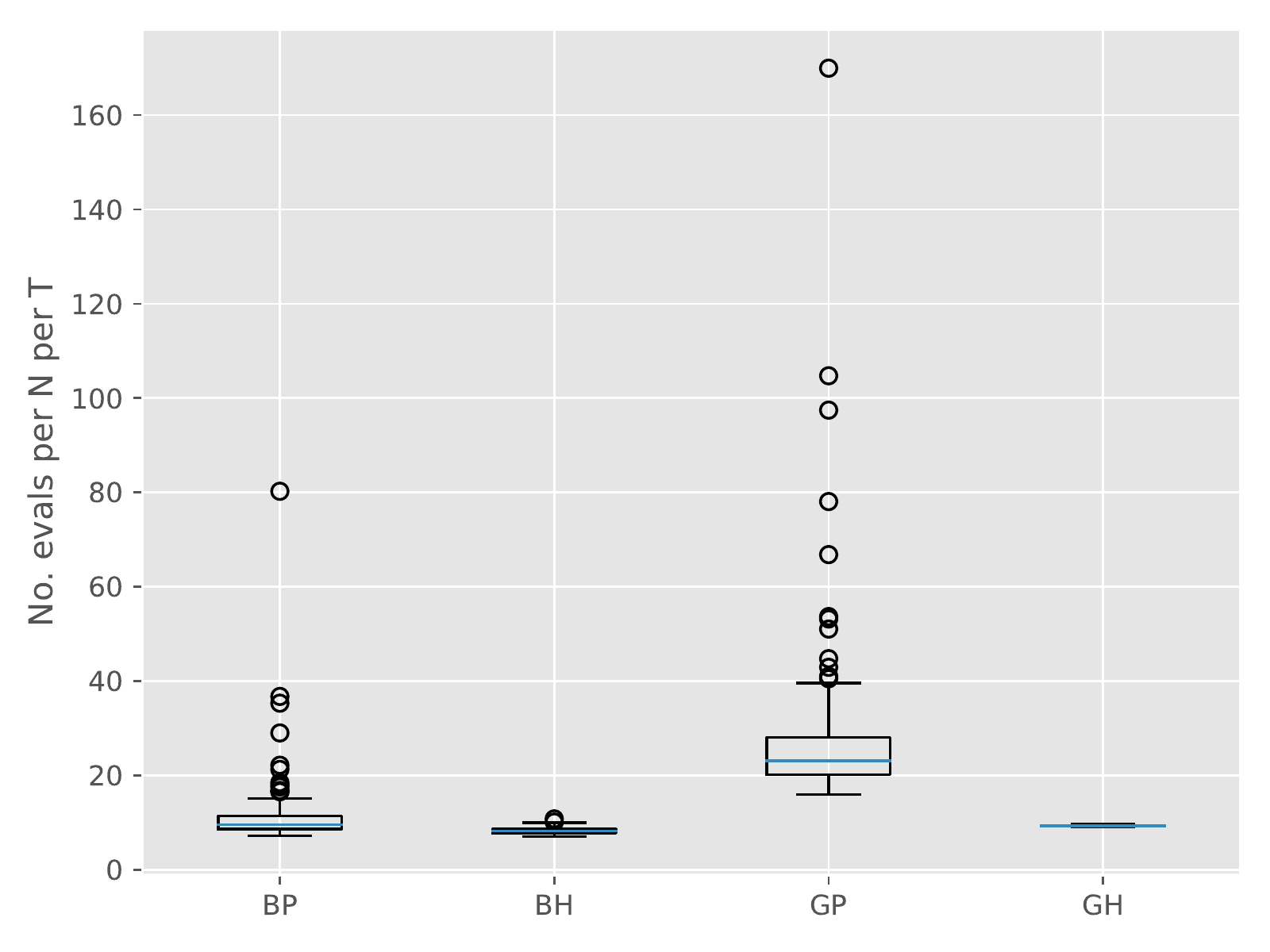}
	\caption{Box plots of the number of transition density evaluations divided by $NT$ for different algorithms in the offline linear Gaussian model of Section~\ref{numexp_linear_gaussian_offline}.}
	\label{fig:lg_offline_small_exec}
\end{figure}

We now take a closer look at the reason behind the performance difference between the bootstrap filter and the guided one when pure rejection sampling is used. Figure~\ref{fig:lg_offline_small_ess} shows the effective sample size (ESS) of both filters as a function of time. We can see that there is an outlier in the data around time $t=40$. Figure~\ref{fig:lg_offline_strange_t40} box-plots the execution times divided by $N$ at $t=40$ for the pure rejection sampling algorithm, whereas Figures~\ref{fig:lg_offline_strange_t40_before} and~\ref{fig:lg_offline_strange_t40_after} do the same for $t = 38$ and $t=42$. The root of the problem is now clear: at most times $t$ there is very few difference between the execution times of the bootstrap and the guided filters. However, if an outlier is present in the data, the guided filter suddenly requires a very high number of transition density evaluation in the rejection sampler. This gives yet another reason to avoid using pure rejection sampling.

\begin{figure}
	\centering
	\includegraphics[scale=0.5]{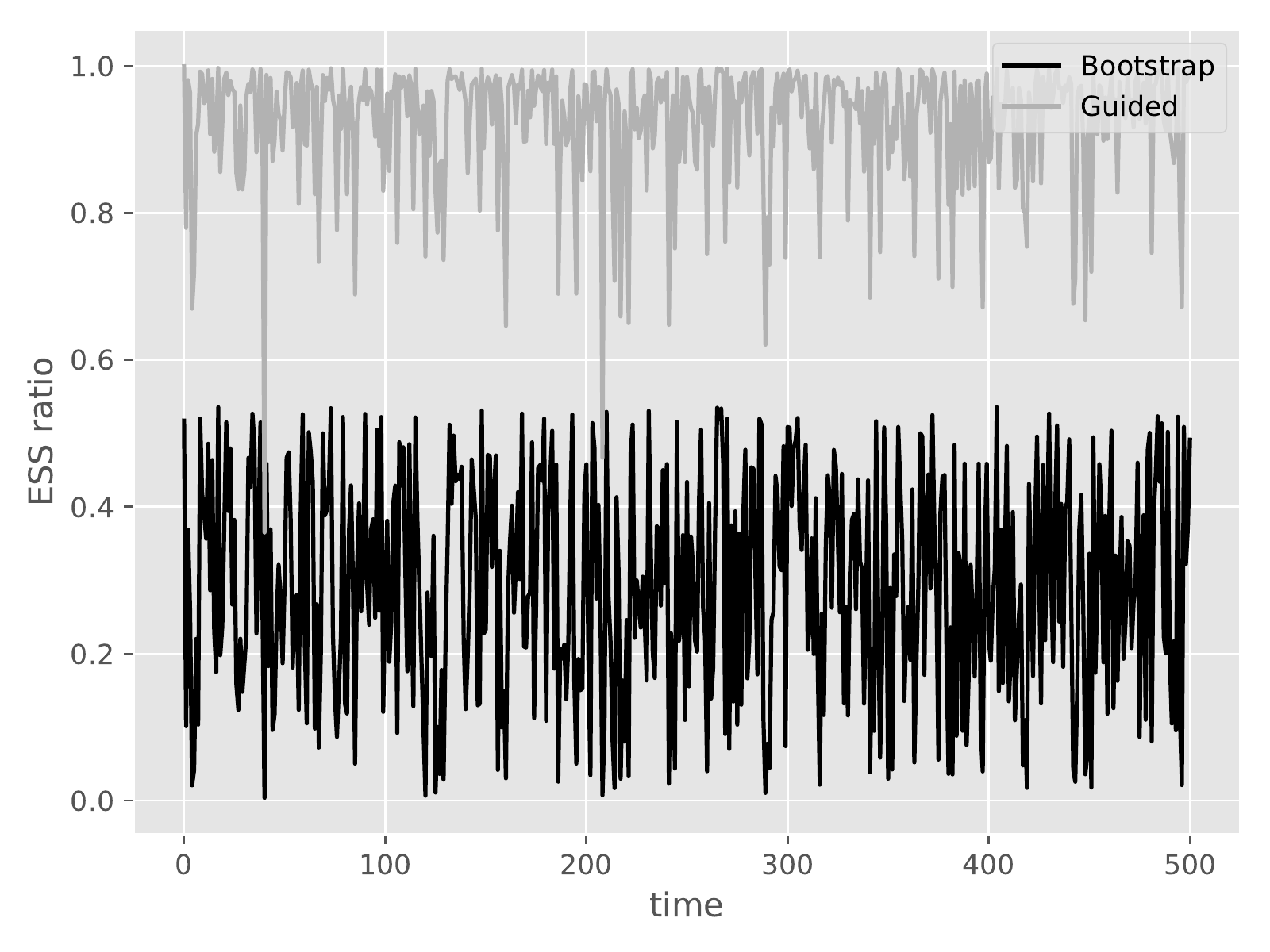}
	\caption{Evolution of the ESS for the linear Gaussian model of Section~\ref{numexp_linear_gaussian_offline}.}
	\label{fig:lg_offline_small_ess}
\end{figure}

\begin{figure}
	\centering
	\includegraphics[scale=0.5]{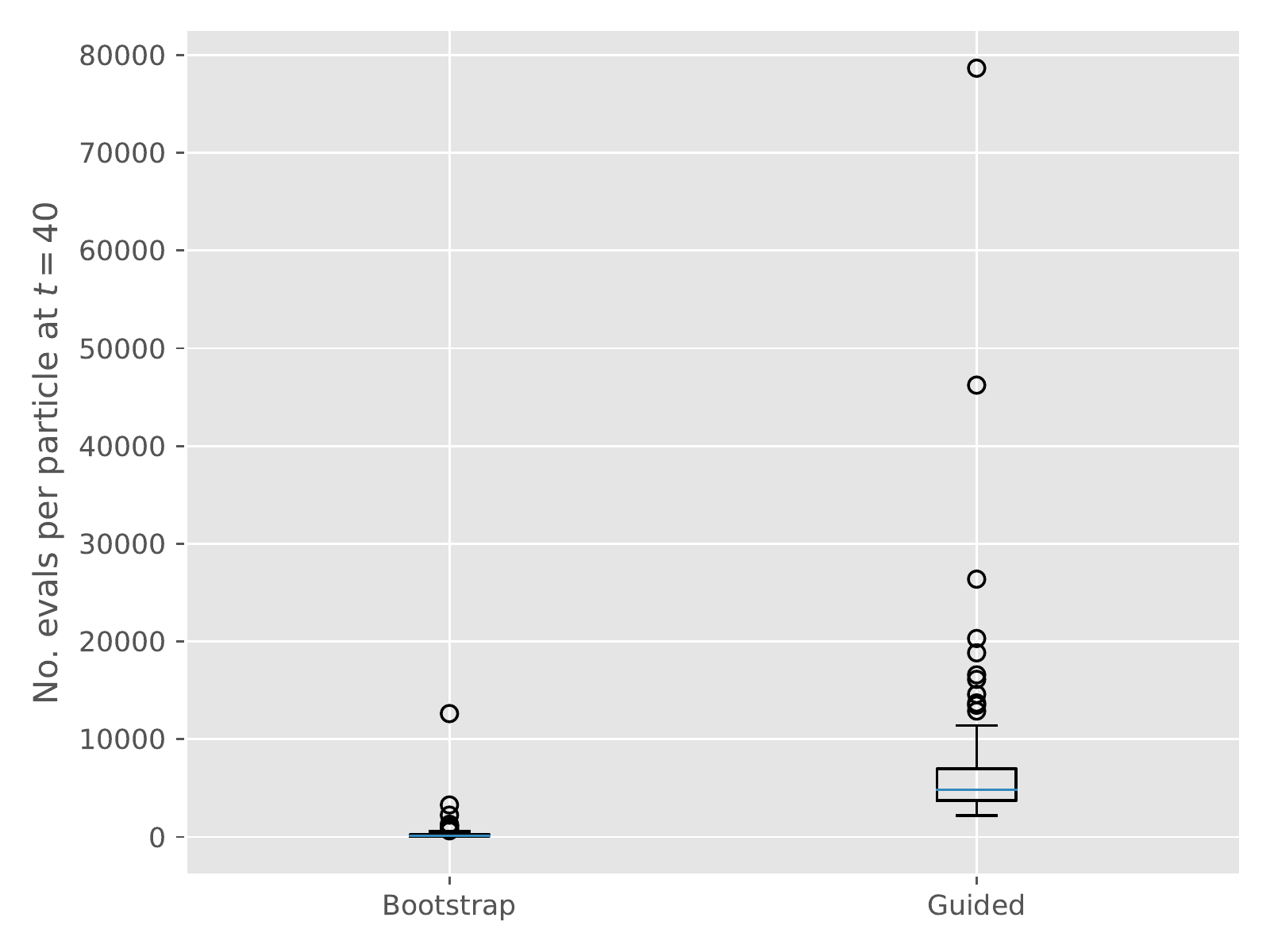}
	\caption{Box plots of the average execution time per particle for the pure rejection algorithm at time $t = 40$. Figure produced based on $150$ independent runs of the model described in Section~\ref{numexp_linear_gaussian_offline}.}
	\label{fig:lg_offline_strange_t40}
\end{figure}

\begin{figure}
	\centering
	\includegraphics[scale=0.5]{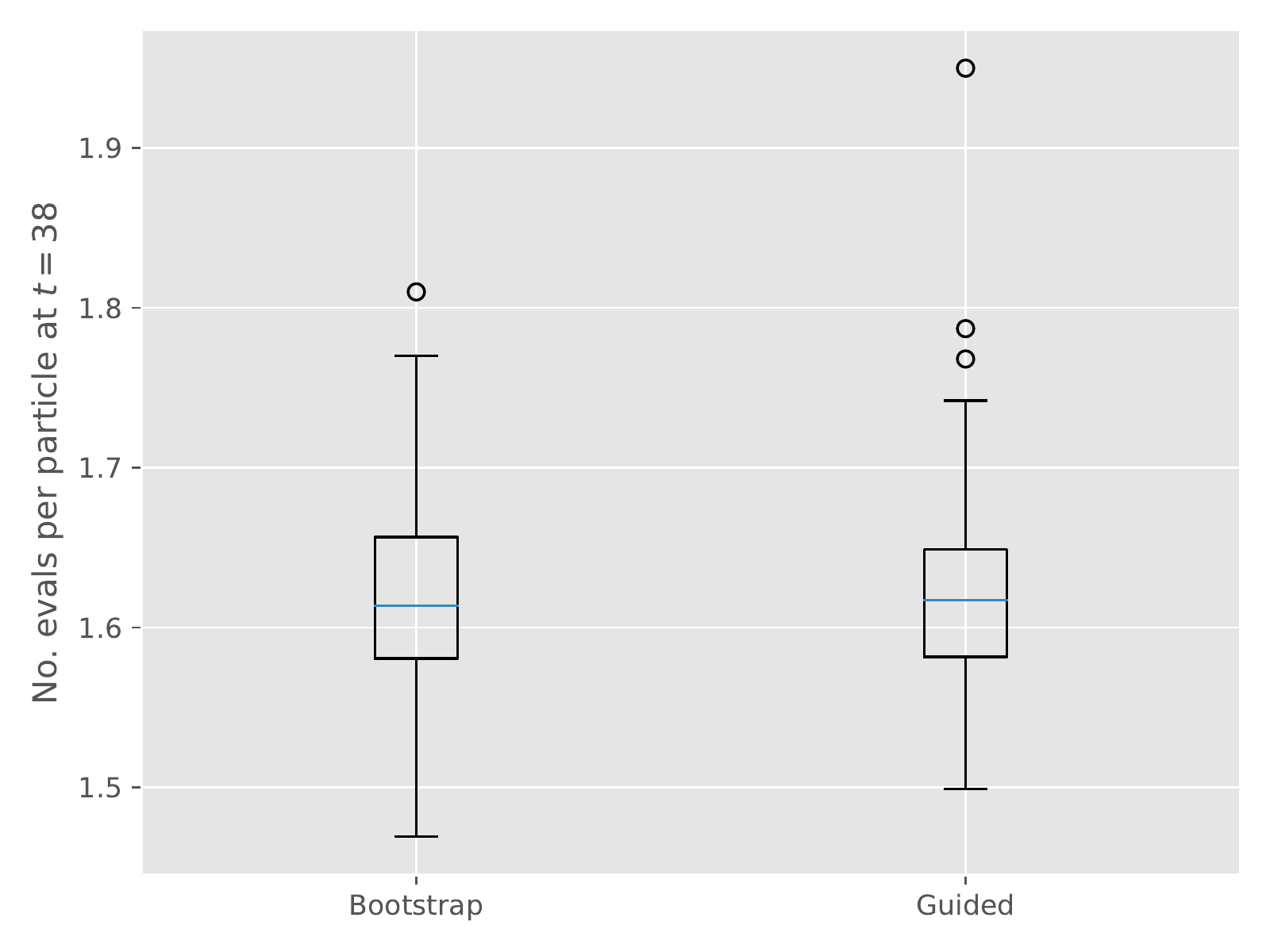}
	\caption{Same as Figure~\ref{fig:lg_offline_strange_t40}, but for $t=38$.}
	\label{fig:lg_offline_strange_t40_before}
\end{figure}

\begin{figure}
	\centering
	\includegraphics[scale=0.5]{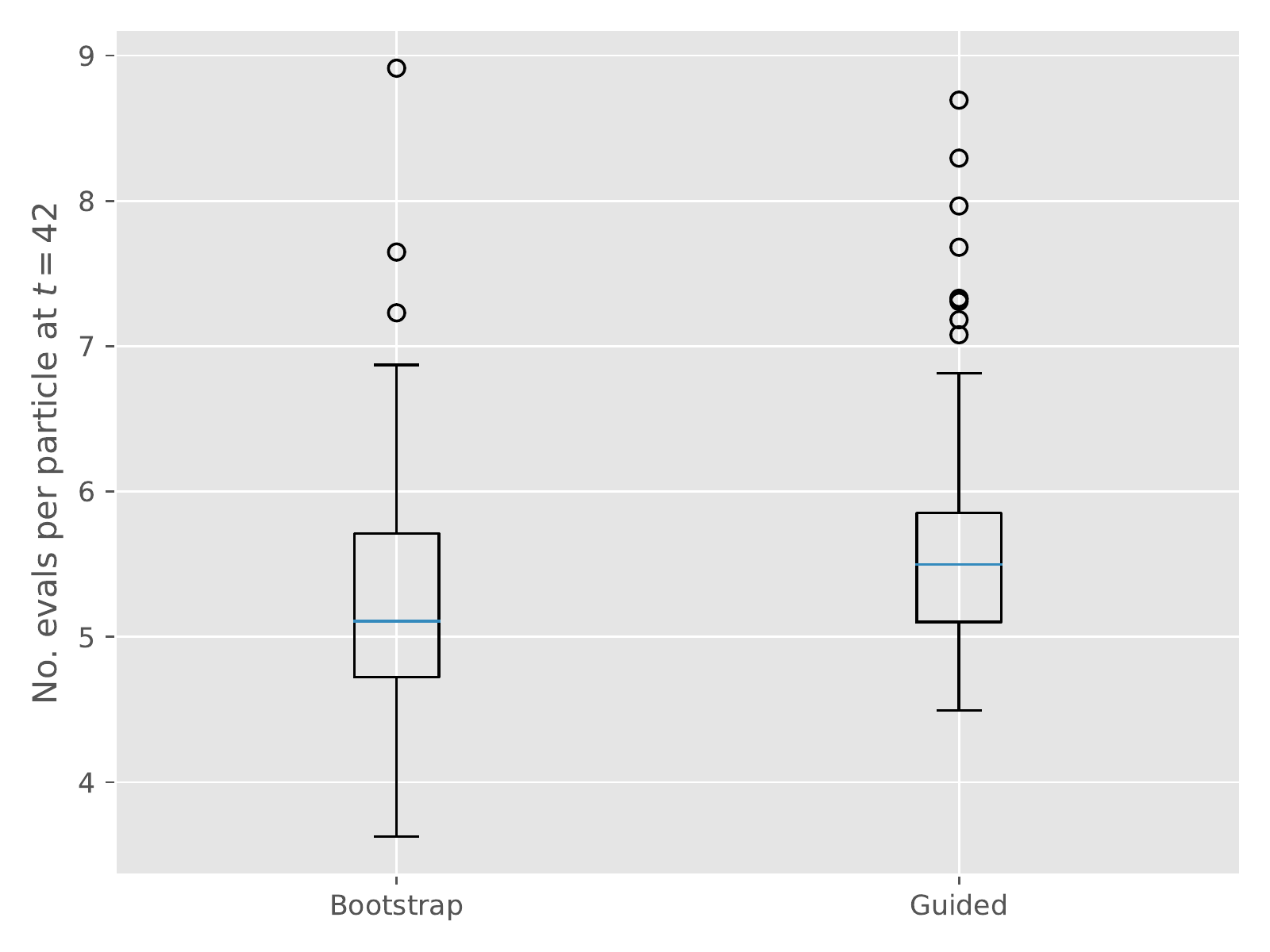}
	\caption{Same as Figure~\ref{fig:lg_offline_strange_t40}, but for $t=42$.}
	\label{fig:lg_offline_strange_t40_after}
\end{figure}

\subsection{Lotka-Volterra SDE}
\subsubsection{Coupling of Euler discretisations}
\label{apx:coupling_euler}
Consider the SDE
\begin{equation}
\label{eq:sde_sample}
\dd X_t = b(X_t) \dd t + \sigma(X_t) \dd W_t
\end{equation}
and two starting points $X_0\expa$ and $X_0\expb$ in $\mathbb R^d$. We wish to simulate $X_1\expa$ and $X_1\expb$ such that the transitions from $X_0\expa$ to $X_1\expa$ and $X_0\expb$ to $X_1\expb$ both follow the Euler-discretised version of the equation, but $X_1\expa$ and $X_1\expb$ are correlated in a way that increases, as much as we can, the probability that they are equal. Algorithm~\ref{algo:coupling_two_euler} makes it clear that it all boils down to the coupling of two Gaussian distributions.
\begin{algo}{Coupling of two Euler discretisations}
	\label{algo:coupling_two_euler}
	\KwIn{Functions $b: \mathbb R^d \to \mathbb R^d$ and $\sigma: \mathbb R^d \to \mathbb R^{d \times d}$, two starting points $X_0\expa$ and $X_0\expb$ at time $0$, number of discretisation step $\ndist$}
	Initialise $X\expa \gets X_0\expa$\;
	Initialise $X\expb \gets X_0\expb$\;
	Set $\delta \gets 1/\ndist$\;
	\For{$i \gets 1$ \KwTo $\ndist$}{
		Simulate $(\tilde X\expa, \tilde X\expb)$ from a coupling of
		\[\mathcal N( X\expa + \delta b( X\expa), \delta \sigma(X\expa) \sigma( X\expa)^\top) \]
		and
		\[\mathcal N( X\expb + \delta b( X\expb), \delta \sigma(X\expb) \sigma( X\expb)^\top), \]
		such as Algorithm~\ref{algo:coupling_two_gaussian}\;
		Update $(X\expa, X\expb) \gets (\tilde X\expa, \tilde X\expb)$\;
	}
	Set $(X_1\expa, X_1\expb) \gets (X\expa, X\expb)$\;
	\KwOut{Two endpoints $X_1\expa$ and $X_1\expb$ at time $1$, obtained by passing $X_0\expa$ and $X_0\expb$ in a correlated manner through a discretised version of \eqref{eq:sde_sample}}  
\end{algo}

\citet{lindvall1986coupling} propose the following construction: if two diffusions $X_t\expa$ and $X_t\expb$ both follow the dynamics of \eqref{eq:sde_sample}, that is,
\begin{align*}
\dd X_t\expa &= b(X_t\expa) \dd t + \sigma(X_t\expa) \dd W_t\expa\\
\dd X_t\expb &= b(X_t\expb) \dd t + \sigma(X_t\expb) \dd W_t\expb
\end{align*}
and the two Brownian motions are correlated via
\begin{equation}
\label{eq:lindvall}
\dd W_t\expb = [\operatorname{Id} - 2u(X\expa, X\expb) u(X\expa, X\expb)^\top] \dd W_t\expa 
\end{equation}
where $\operatorname{Id}$ is the identity matrix and the vector $u$ is defined by
\[u(x,x') = \frac{\sigma(x')^{-1}(x-x')}{\norm{\sigma(x')^{-1}(x-x')}_2}, \]
then under some regularity conditions, the two diffusions meet almost surely. (Note two special features of \eqref{eq:lindvall}: it is valid because the term in the square bracket is an orthogonal matrix; and it ceases to be well-defined once the two trajectories have met.) Simulating the meeting time $\tau$ turns out to be very challenging. The Euler discretisation (Algorithm~\ref{algo:coupling_two_euler} + Algorithm~\ref{algo:coupling_two_gaussian_lindvall}) has a fixed step size $\delta$, and there is zero probability that $\tau$ is of the form $k \delta$ for some integer $k$. Since the coupling transform is deterministic, the two Euler-simulated trajectories will \textit{never} meet. Figure~\ref{fig:lindvall_naive_brownian} depicts this difficulty in the special case of two Brownian motions in dimension $1$ (i.e.\ $b(x) \equiv 0$ and $\sigma \equiv 1$). Under this setting, \eqref{eq:lindvall} means that the two Brownian increments are symmetric with respect to the midpoint of the segment connecting their initial states. Note that the two dashed lines do cross at two points, but using them as meeting points is invalid: since they are not part of the discretisation but the result of some heuristic ``linear interpolation'', it would change the distribution of the trajectories.

\begin{algo}{Lindvall-Rogers coupling of two Gaussian distributions}
	\label{algo:coupling_two_gaussian_lindvall}
	\KwIn{Two vectors $\mu\expa, \mu\expb$ in $\mathbb R^d$ and two $d \times d$ matrices $\sigma\expa$ and $\sigma\expb$}
	Calculate $u \gets (\sigma\expb)^{-1} (\mu\expa - \mu\expb)$\;
	Normalise $u \gets u/\norm{u}_2$\;
	Simulate $W\expa \sim \mathcal N(0, \operatorname{Id})$\;
	Set $W\expb \gets (\operatorname{Id} - 2uu^\top) W\expa$\;
	Set $X\expa \gets \mu\expa + \sigma\expa W\expa$\;
	Set $X\expb \gets \mu\expb + \sigma\expb W\expb$\;
	\KwOut{Two correlated points $X\expa$ and $X\expb$ marginally distributed according to $\mathcal N(\mu\expa, \sigma\expa (\sigma\expa)^\top)$ and $\mathcal N(\mu\expb, \sigma\expb (\sigma\expb)^\top)$ respectively}
\end{algo}

\begin{figure}
	\includegraphics[scale=0.5]{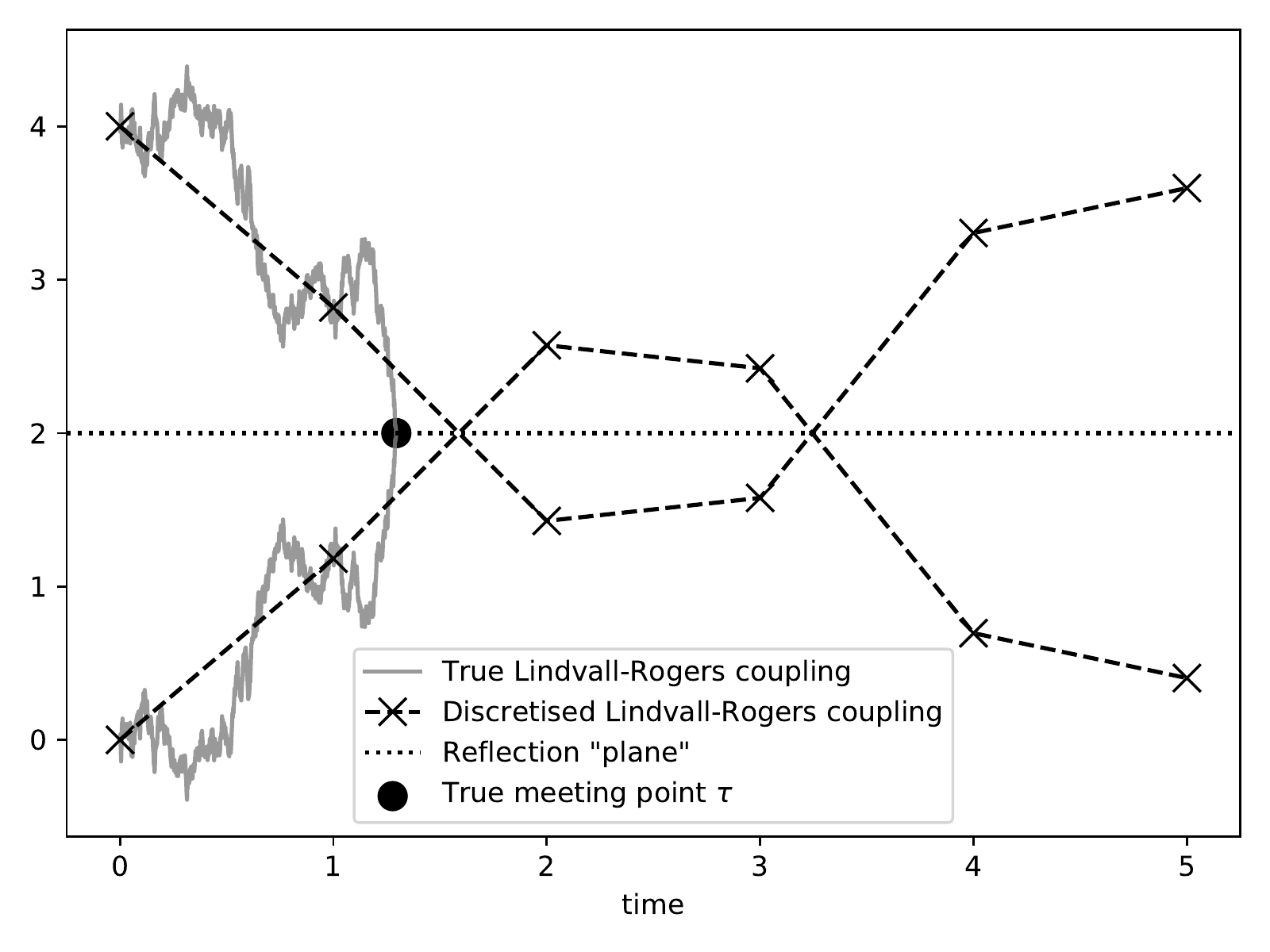}
	\caption{Coupling of two Brownian motions in $\mathbb R$ starting from $0$ and $4$ respectively. The true Lindvall-Rogers coupling \eqref{eq:lindvall} is represented by the continuous grey lines. The dicretised simulation (Algorithm~\ref{algo:coupling_two_euler} + Algorithm~\ref{algo:coupling_two_gaussian_lindvall}) is shown by the dashed lines. The discretised trajectories not only miss the true meeting point $\tau$ but also never meet afterwards (see text).}
	\label{fig:lindvall_naive_brownian}
\end{figure}

We therefore need some coupling that has a non-zero meeting probability at each $\delta$-step. This can be achieved by the rejection maximal coupling (Algorithm~\ref{algo:maximal_coupling}, see also, e.g.\ \citealp{roberts2004general}) as well as the recently proposed coupled rejection sampler \citep{corenflos2022coupledrejection}. However, they all make use of rejection sampling in one way or another, which renders the execution time random. We wish to avoid this if possible. The reflection-maximal coupling \citep{bourabee2020coupling, jacob_unbiased_mcmc} has deterministic cost and optimal meeting probability, but is only applicable for two Gaussian distributions of the same covariance matrix, which is not our case.

\begin{algo}{Rejection maximal coupler for two distributions}
	\label{algo:maximal_coupling}
	\KwIn{Two probability distributions $f\expa$ and $f\expb$}
	Simulate $X\expa \sim f\expa$\;
	Simulate $U\expa \sim \operatorname{Uniform}[0, f\expa(X\expa)]$\;
	\If{$U\expa \leq f\expb(X\expa)$}{
		Set $X\expb \gets X\expa$\;
	}
	\Else{
		\Repeat{$U\expb > f\expa(X\expb)$}{
			Simulate $X\expb \sim f\expb$\;
			Simulate $U\expb \sim \operatorname{Uniform}[0, f\expb(X\expb)]$\;
		}
	}
	\KwOut{Two maximally-coupled realisations $X\expa$ and $X\expb$, marginally $f\expa$-distributed and $f\expb$-distributed respectively}
\end{algo}

As suggested by Figure~\ref{fig:lindvall_naive_brownian}, the discretised Lindvall-Rogers coupling (Algorithm~\ref{algo:coupling_two_gaussian_lindvall}) is actually great for bringing together two faraway trajectories. Only when they start getting closer that it misses out. At that moment, the two distributions corresponding to the next $\delta$-step have non-negligible overlap and would preferably be coupled in the style of Algorithm~\ref{algo:maximal_coupling}. We propose a modified coupling scheme that acts like Algorithm~\ref{algo:coupling_two_gaussian_lindvall} when the two trajectories are at a large distance and behaves as Algorithm~\ref{algo:maximal_coupling} otherwise.

The idea is to preliminarily generate a uniform draw in the ``overlapping
zone'' of the two distributions (if they are close enough to make that easy). Next, we perform Algorithm~\ref{algo:coupling_two_gaussian_lindvall} and then, any of the two simulations belonging to the overlapping zone will be replaced by the aforementioned preliminary draw (if it is available). The precise mathematical formulation is given in Algorithm~\ref{algo:coupling_two_gaussian} and the proof in Supplement~\ref{subsect:validity_lindvall_rogers_coupler}.

\begin{algo}{Modified Lindvall-Rogers (MLR) coupler of two Gaussian distributions}
	\label{algo:coupling_two_gaussian}
	\KwIn{Two vectors $\mu\expa$ and $\mu\expb$ in $\mathbb R^d$, two $d\times d$ matrices $\sigma\expa$ and $\sigma\expb$}
	Let $f\expa$ and $f\expb$ be respectively the probability densities of $\mathcal N(\mu\expa, \covmata)$ and $\mathcal N(\mu\expb, \covmatb)$\;
	Simulate $X\expa$ and $X\expb$ from Algorithm~\ref{algo:coupling_two_gaussian_lindvall}\;
	Simulate $U \sim \operatorname{Uniform}[0,1]$\;
	Set $U\expa\gets Uf\expa(X\expa)$ and $U\expb\gets Uf\expb(X\expb)$\;
	Simulate $Y \sim f\expa$ and $V \sim \operatorname{Uniform}[0, f\expa(Y)]$\;
	\If{$V\leq f\expb(Y)$}{
		if $U\expa \leq f\expb(X\expa)$ then update $(X\expa, U\expa) \gets (Y, V)$\;
		if $U\expb \leq f\expa(X\expb)$ then update $(X\expb, U\expb) \gets (Y, V)$\; 
	}
	\KwOut{Two correlated random vectors $X\expa$ and $X\expb$, distributed marginally according to $\mathcal N(\mu\expa, \covmata)$ and $\mathcal N(\mu\expb, \covmatb)$}
\end{algo}

Algorithm~\ref{algo:coupling_two_gaussian} has a deterministic execution time,
but it does not attain the optimal coupling rate. Yet, as $\delta \to 0$, we
see empirically that it still recovers the oracle coupling time defined 
by~\eqref{eq:lindvall} (although we did not try to prove this formally). In
Figure~\ref{fig:brownian_delta0_meeting}, we couple two standard Brownian
motions starting from $a=0$ and $b=1.5$ using
Algorithm~\ref{algo:coupling_two_gaussian} with different values of $\delta$.
It is known, by a simple application of the reflection principle
(\citealp{levy1940certains}; see also Chapter 2.2 of
\citealp{morters2010brownian}), that the reflection
coupling~\eqref{eq:lindvall} succeeds after a
$\operatorname{Levy}(0,(b-a)^2/4)$-distributed time. We therefore have to deal
with a heavy-tailed distribution and restrict ourselves to the interval
$[0,5]$. We see that the law of the meeting time is stable and convergent as
$\delta \to 0$. Thus, at least empirically,
Algorithm~\ref{algo:coupling_two_gaussian} does not suffer from the instability
problem as $\delta \to 0$, contrary to a naive path space augmentation approach
(see \citealp{yonekura2022online_smoothing} for a discussion).

\begin{figure}
	\centering
	\includegraphics[scale=0.5]{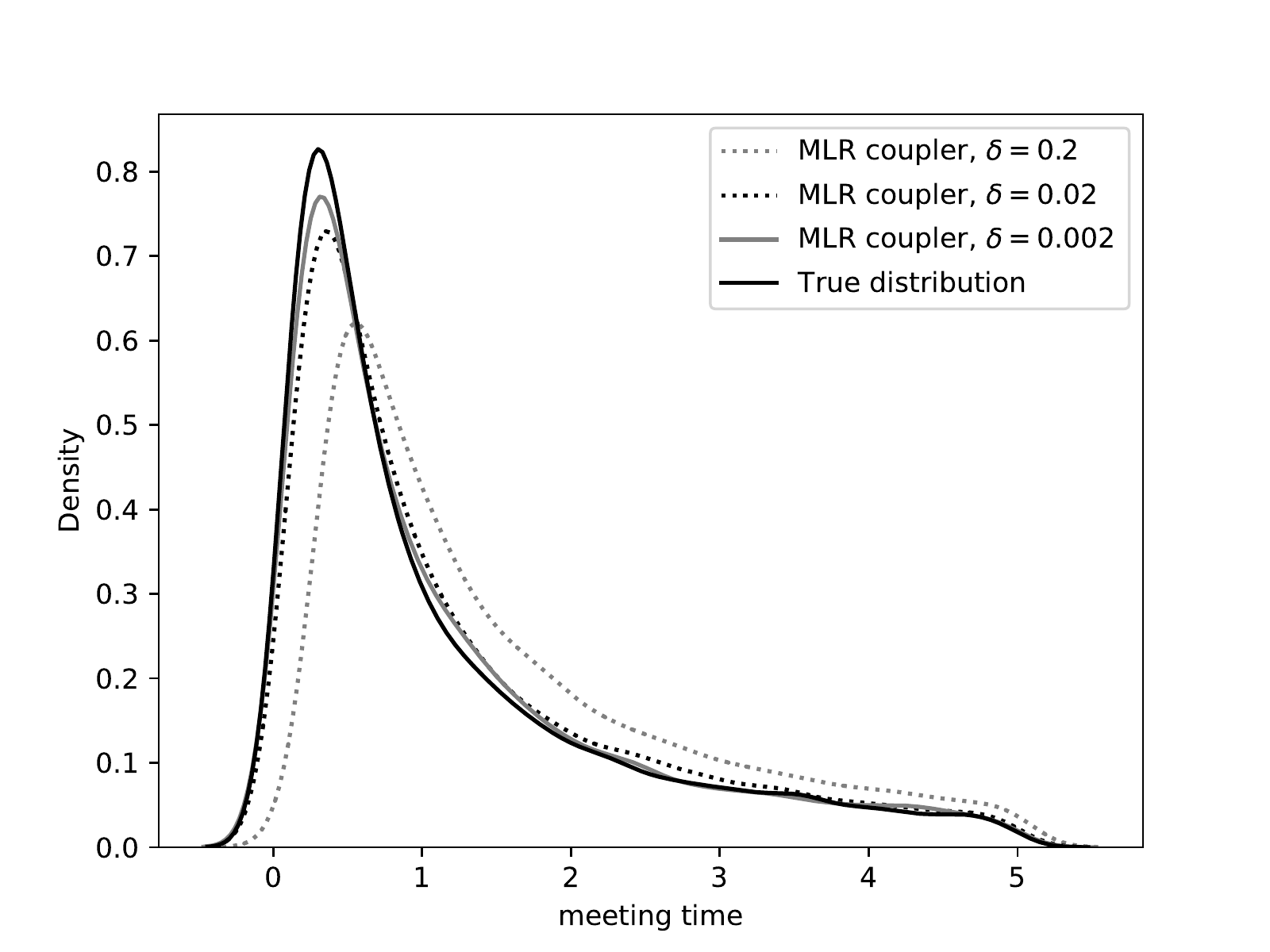}
	\caption{Densities of the meeting times restricted to $[0,5]$ for two
    Brownian motions started from $0$ and $1.5$. The curves are drawn using
$20\ 000$ simulations from either a Levy distribution (for the ``True
distribution'' curve) or Algorithm~\ref{algo:coupling_two_gaussian} (for the MLR ones). The boundary effect of kernel density estimators causes spills beyond $0$ and $5$.\label{fig:brownian_delta0_meeting}}
\end{figure}

\subsubsection{Supplementary figures}\label{more_figures_sde}

Figure~\ref{fig:sde_small_realisation} plots a realisation of the states and data with parameters given in Subsection~\ref{numexp_lk_sde}, for a relatively small scale dataset ($T=50$). While the periodic trait seen in classical deterministic Lotka-Volterra equations is still visible (with a period of around $20$), it is clear that here random perturbations have added considerable chaos to the system. Figures~\ref{fig:sde_small_naive} and~\ref{fig:sde_small_intractable} show respectively the performances of the naive genealogy tracking smoother and ours (Algorithm~\ref{algo:intractable_practice}) on the dataset of Figure~\ref{fig:sde_small_realisation}. Our smoother has successfully prevented the degeneracy phenomenon, particularly for times close to $0$. Figure~\ref{fig:sde_big_ess} shows, in two different ways, the properties of effective sample sizes (ESS) in the $T=3000$ scenario (see Section~\ref{numexp_lk_sde}). 
\begin{figure}
	\centering
	\includegraphics[scale=0.5]{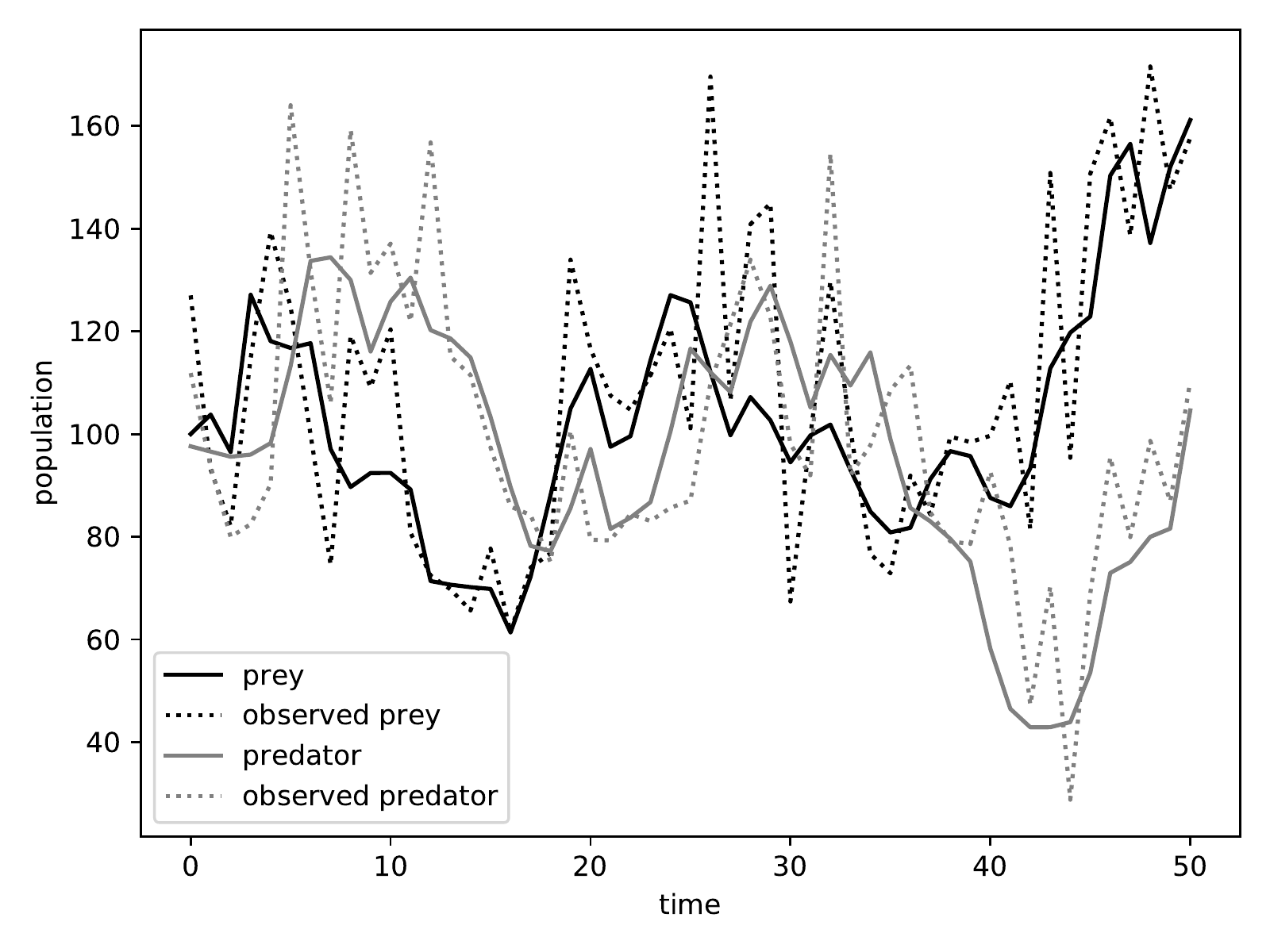}
	\caption{A realisation of the Lotka-Volterra SDE with parameters described in Section~\ref{numexp_lk_sde}. The stationary point of the system is $[100, 100]$.}
	\label{fig:sde_small_realisation}
\end{figure}

\begin{figure}
	\centering
	\includegraphics[scale=0.5]{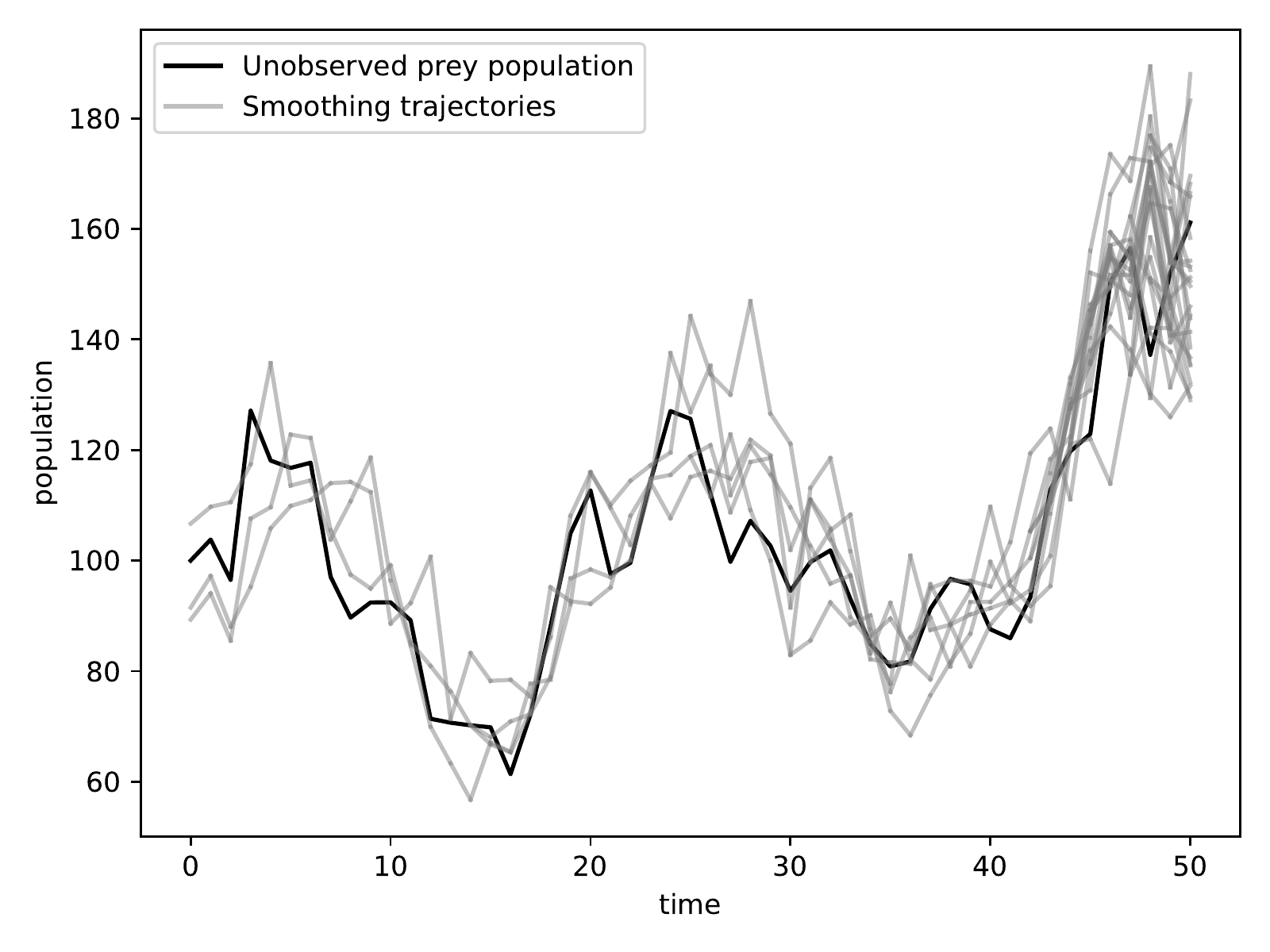}
	\caption{Smoothing trajectories for the dataset of Figure~\ref{fig:sde_small_realisation} using the naive genealogy tracking smoother ($\btn{GT}$ kernels) with systematic resampling (see Section~\ref{apx:alternative_resampling_schemes}). We took $N=100$ and randomly plotted $30$ smoothing trajectories.}
	\label{fig:sde_small_naive}
\end{figure}

\begin{figure}
	\centering
	\includegraphics[scale=0.5]{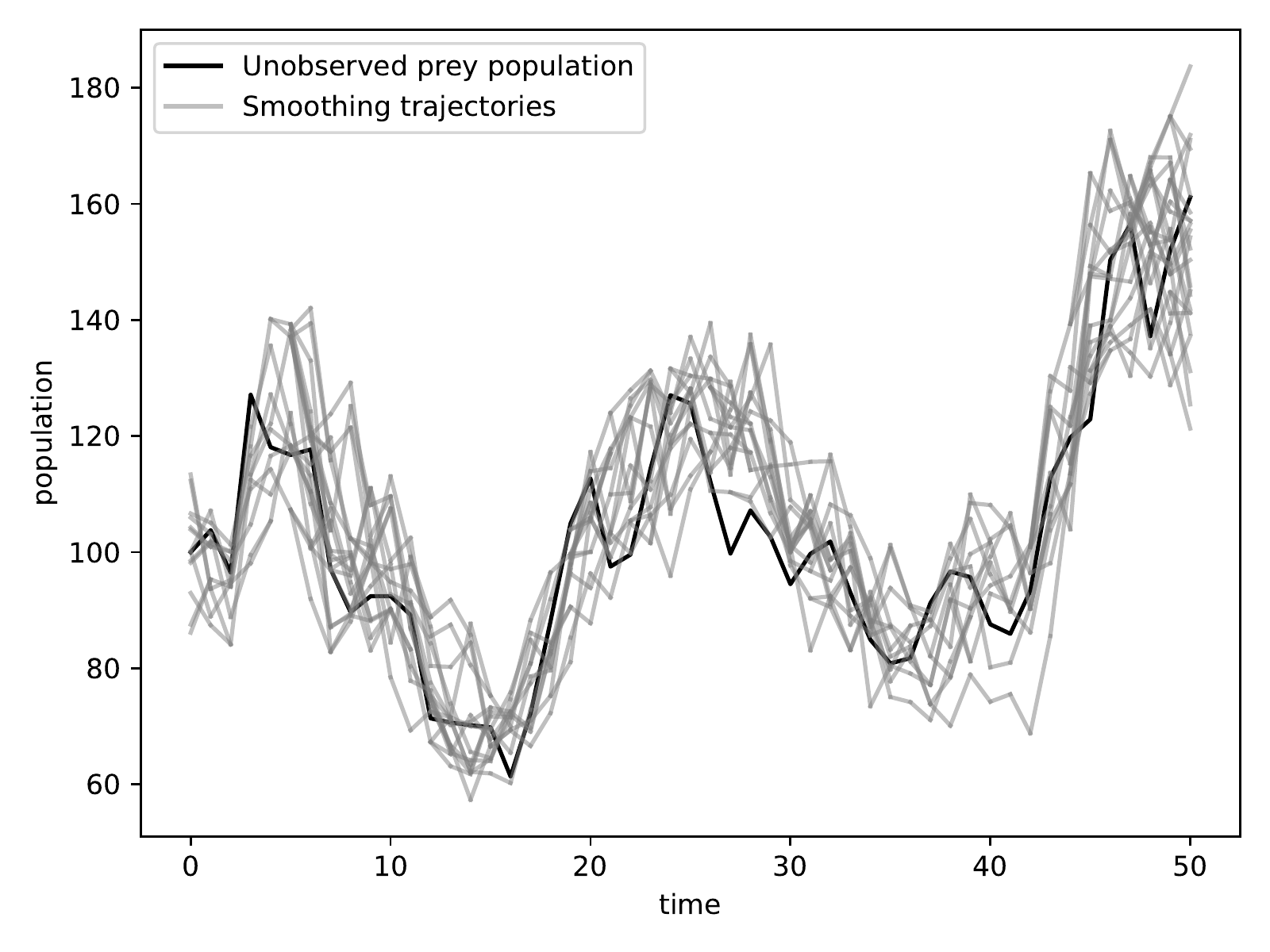}
	\caption{Same as Figure~\ref{fig:sde_small_naive}, but smoothing was done using Algorithm~\ref{algo:intractable_practice} instead.}
	\label{fig:sde_small_intractable}
\end{figure}

\begin{figure}
	\centering
	\includegraphics[scale=0.5]{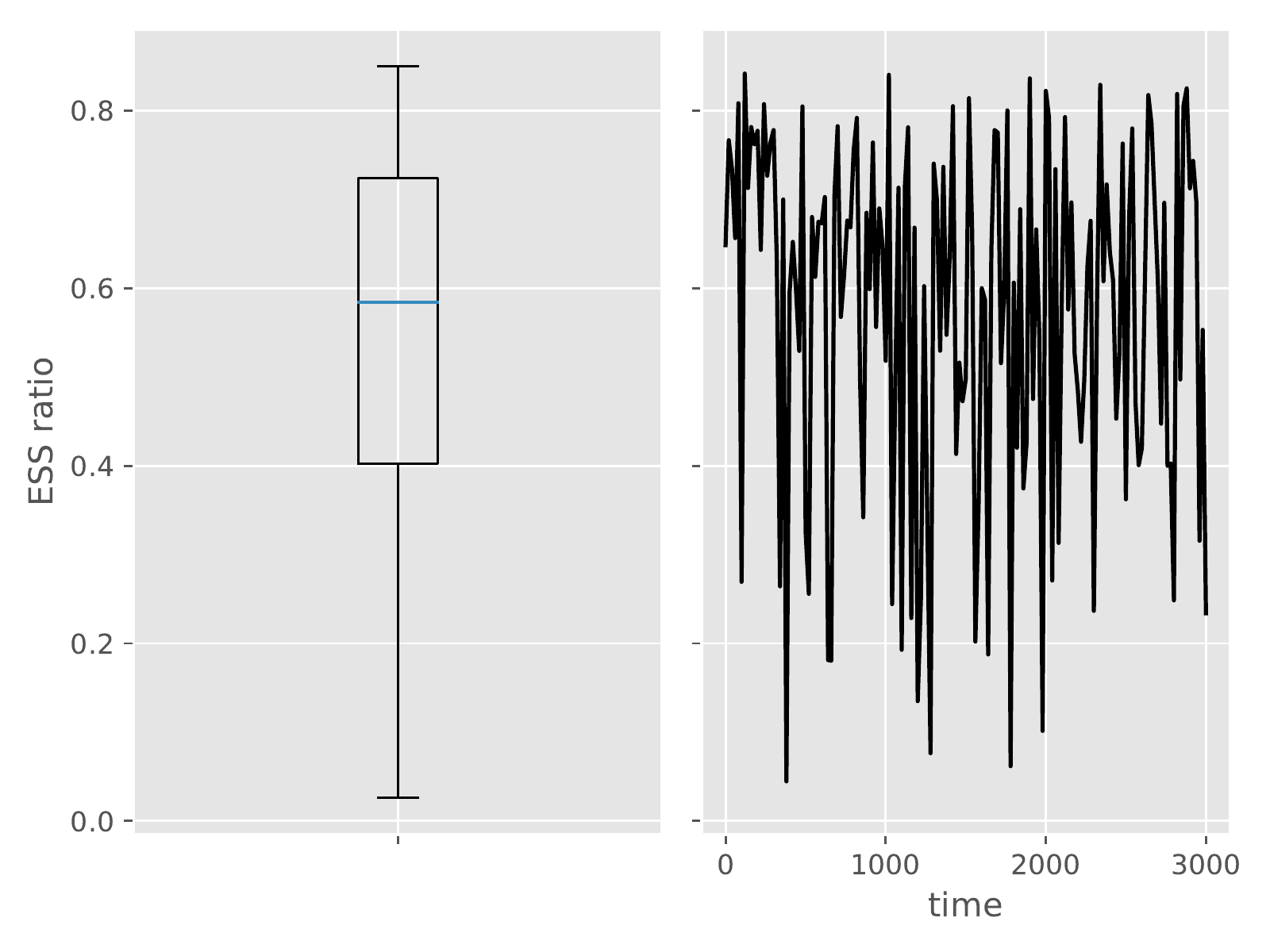}
	\caption{Effective sample size (ESS) for the Lotka-Voltterra SDE model with $T = 3000$ (Section~\ref{numexp_lk_sde}). The left pane draws the box plot of the collection of all estimated ESS for $t=0, \ldots, 3000$. The right pane plots the evolution of ESS with time. The quantity changes so chaotically that the curve only plots one value every $20$ time steps for readability.}
	\label{fig:sde_big_ess}
\end{figure}

\section{Proofs}

\subsection{\prot~\ref{thm:convergence_mcmc} (general convergence theorem)}
\label{ap:proof:cvg}
In line with~\eqref{eq:joint_empirical_smoothing}, we define the distribution $\Q_t^N(\dd x_{0:t})$ for $t<T$ as the $x_{0:t}$ marginal of the joint distribution
\begin{equation}
\label{eq:def_joint_qt}
\bar \Q_t^N(\dd x_{0:t}, \dd i_{0:t}) \eqdef \mathcal{M}(W_t^{1:N})(\dd i_t) \ps{\prod_{s=t}^1 B_s^N(i_s, \dd i_{s-1})} \ps{\prod_{s=t}^0 \delta_{X_s^{i_s}}(\dd x_s)}.
\end{equation}
The proof builds up on an inductive argument which links $\Q_t^N$ with $\Q_{t-1}^N$ through new innovations at time $t$. More precisely, we have the following fundamental proposition, where $\FTP t$ is defined as the smallest $\sigma$-algebra containing $\mathcal F_t$ and $\hat B_{1:t}^N$.

\begin{prop}\label{prop:fundamental}
	$\Q_t^N$ is a mixture distribution that admits the representation
    \begin{equation}
        \label{eq:mixture_rep_qtn}
        \Q_t^N(\dd x_{0:t}) = \nmt \sum_n G_t(x_t) K_t^N(n, \dd x_{0:t}) 
    \end{equation}
	where $\ell_t^N$ is defined in Algorithm~\ref{algo:bootstrap} and $K_t^N(n, \dd x_{0:t})$ is a certain probability measure satisfying
	\begin{equation}\label{eq:link_qt_qtm1} 
        \ceftmone{K_t^N(n, \dd x_{0:t})} = \Q_{t-1}^N(\dd x_{0:t-1}) M_t(x_{t-1}, \dd x_t). 
    \end{equation}
In other words, for any (possibly random) function $\varphi_t^N: \mathcal X_0
\times \cdots \times \mathcal X_t \to \mathbb R$ such that $\varphi_t^N(x_{0:t})$
is $\FTP{t-1}$-measurable, we have
	\[\ceftmone{\int K_t^N(n, \dd x_{0:t}) \varphi_t^N(x_{0:t})} = \int \Q_{t-1}^N(\dd x_{0:t-1}) M_t(x_{t-1}, \dd x_t) \varphi_t^N(x_{0:t}). \]
	Moreover, $\int K_t^N(n, \dd x_{0:t}) \varphi_t^N(x_{0:t})$, for $n=1,
    \ldots, N$ are i.i.d.\ given $\FTP{t-1}$.
\end{prop}

The proof is postponed until the end of this subsection. This proposition gives
the expression~\eqref{eq:mixture_rep_qtn} for $\Q_t^N$, which is easier to
manipulate than~\eqref{eq:def_joint_qt} and which highlights,
through~\eqref{eq:link_qt_qtm1}, its connection to $\Q_{t-1}^N$. To further simplify
the notations, let us define, following \citet{Douc2011}, the kernel
$L_{t_1:t_2}$, for $t_1 \leq t_2$, as
\begin{equation}
\label{eq:def:l_kernels}
L_{t_1:t_2}(x^\star_{0:t_1}, \dd x_{0:t_2}) \eqdef \delta_{x^\star_{0:t_1}}(\dd x_{0:t_1}) \prod_{s=t_1+1}^{t_2} M_s(x_{s-1}, \dd x_s) G_s(x_s). 
\end{equation}
In other words, for real-valued functions $\varphi_{t_2}=\varphi_{t_2}(x_0, \ldots, x_{t_2})$, we have
\[L_{t_1:t_2}(x^\star_{0:t_1}, \varphi_{t_2}) = \int \varphi_{t_2}(x_0^\star, \ldots, x_{t_1}^\star, x_{t_1+1}, \ldots, x_{t_2}) \prod_{s=t_1+1}^{t_2} M_s(x_{s-1}, \dd x_s) G_s(x_s).\]
The usefulness of these kernels will come from the simple remark $\Q_{t_2} \propto \Q_{t_1} L_{t_1:t_2}$. We also see that
\[\infnorm{L_{t_1:t_2}\varphi_{t_2}} \leq \infnorm{\varphi_{t_2}} \prod_{s=t_1+1}^{t_2} \infnorm{G_s}, \]
which gives $\infnorm{L_{t_1:t_2}} < \infty$, where the norm of a kernel is defined in Subsection~\ref{sec:apx_notations}. We are now in a position to state an importance sampling-like representation of $\Q_t^N$.
\begin{corollary}\label{corol:fundamental}
	Let $\varphi_t^N: \mathcal X_0 \times \cdots \times \mathcal X_t \to \mathbb R$ be a (possibly random) function such that $\varphi_t^N(x_{0:t})$ is $\FTP{t-1}$-measurable. Suppose that $\varphi_t^N$ is either uniformly non-negative (i.e.\ $\varphi_t^N(x_{0:t}) \geq 0$ almost surely) or uniformly bounded (i.e.\ there exists a deterministic $C$ such that $|\varphi_t^N(x_{0:t})| \leq C$ almost surely). Then
	\[\Q_t^N \varphi_t^N = \frac{N^{-1} \sum_n \tilde K_t^N(n, \varphi_t^N)}{N^{-1} \sum_n \tilde K_t^N(n, \mathbbm 1)}, \]
	where $\tilde K_t^N(n, \cdot)$ is a certain random kernel such that
	\begin{itemize}
		\item $\ceftmone{\tilde K_t^N(n, \varphi_t^N)} = (\Q_{t-1}^N L_{{t-1}:t}) \varphi_t^N$;
		\item $N^{-1} \sum_n \tilde K_t^N(n, \mathbbm 1) = \ell_t^N$;
		\item $\left(\tilde K_t^N(n, \varphi_t^N)\right)_{n=1,\ldots,N}$ are i.i.d.\ given $\FTP{t-1}$;
		\item almost surely, $\abs{\tilde K_t^N(n, \varphi_t^N)} \leq \infnorm{\varphi_t^N} \infnorm{G_t}$ if $\varphi_t^N$ is uniformly bounded and $\tilde K_t^N(n, \varphi_t^N) \geq 0$ if $\varphi_t^N$ is uniformly non-negative.
	\end{itemize}
	These statements are valid for $t=0$ under the convention $\Q_{-1}^N L_{-1:0} = \Q_{-1}L_{-1:0} = \mathbb M_0$ and $\mathcal F_{-1}$ being the trivial $\sigma$-algebra.
\end{corollary}
\begin{proof}
	Put $\tilde K_t^N(n, \varphi_t^N) := \int G_t(x_t) K_t^N(n, \ddxz t) \varphi_t^N(x_{0:t})$ where $K_t^N$ is defined in Proposition~\ref{prop:fundamental}. Then
	\begin{equation*}
	\begin{split}
	\Q_t^N(\varphi_t^N) &= \frac{N^{-1} \sum_n \int G_t(x_t) K_t^N(n, \ddxz t) \varphi_t^N(x_{0:t})}{\ell_t^N} \\
	&= \frac{N^{-1} \sum_n \tilde K_t^N(n, \varphi_t^N)}{\ell_t^N}.
	\end{split}
	\end{equation*}
	Since $\Q_t^N$ is a probability measure, applying this identity twice yields
	\[\Q_t^N(\varphi_t^N) = \frac{\Q_t^N(\varphi_t^N)}{\Q_t^N(\mathbbm 1)} = \frac{N^{-1} \sum_n \tilde K_t^N(n, \varphi_t^N)}{N^{-1} \sum_n \tilde K_t^N(n, \mathbbm 1)}. \]
	The remaining points are simple consequences of the definition of $\tilde K_t^N$ and $L_{t-1:t}$.
\end{proof}
The corollary hints at a natural induction proof for Theorem~\ref{thm:convergence_mcmc}.
\begin{proof}[Proof of Theorem~\ref{thm:convergence_mcmc}]
	The following calculations are valid for all $T \geq 0$, under the
    convention defined at the end of Corollary~\ref{corol:fundamental}. They
    will prove~\eqref{eq:in_thm_convg} for $T=0$ and, at the same time, prove
    it for any $T \geq 1$ under the hypothesis that it already holds true for
    $T-1$. Let $\varphi_T = \varphi_T(x_0, \ldots, x_T)$ be a real-valued
    function on $\mathcal X_0 \times \cdots \times \mathcal X_T$. Write 
    \begin{equation}
		\label{eq:delta_tn}
		\sqrt N (\Q_T^N \varphi_T - \Q_T \varphi_T) = \sqrt N \pr{\frac{\inv N \sum_n \tilde K_T^N(n, \varphi_T)}{\inv N \sum_n \tilde K_T^N(n, \mathbbm 1)} - \frac{\Q_{T-1}L_{T-1:T}\varphi_T}{\Q_{T-1}L_{T-1:T}\mathbbm 1}}
	\end{equation}
	where the rewriting of $\Q_T \varphi_T$ is a consequence of $Q_T \propto \Q_{T-1} L_{T-1:T}$. We will bound this difference by Hoeffding's inequalities for ratios (see Supplement~\ref{ap:hoeffding} for notations, including the definition of sub-Gaussian variables that we shall use below). We have
	\begin{itemize}
\item that $\sqrt N(\inv N \sum \tilde K_T^N(n, \varphi_T) - \Q_{T-1}^N L_{T-1:T}\varphi_T)$ is $(1, \infnorm{\varphi_T} \infnorm{G_T})$-sub-Gaussian conditioned on $\FTP{t-1}$ because of Theorem~\ref{thm:hoeffing_apx} (and thus unconditionally, by the law of total expectation);
\item and that $\sqrt N(\inv N \Q_{T-1}^N L_{T-1:T}\varphi_T - \Q_{T-1}L_{T-1:T}\varphi_T)$ is sub-Gaussian with parameters \[(C_{T-1}, S_{T-1}\infnorm{L_{T-1:T}} \infnorm{\varphi_T})\]
if $T \geq 1$ by induction hypothesis. The quantity is equal to $0$ if $T=0$.
	\end{itemize}
This permits to apply Lemma~\ref{lem:hoeffding_apx}, which results in the sub-Gaussian properties of
\begin{itemize}
	\item the quantity $\sqrt N(\inv N \sum \tilde K_T^N(n, \varphi_T) - \Q_{T-1}L_{T-1:T}\varphi_T)$, with parameters $(1+C_{T-1}, S_{T-1}' \infnorm{\varphi_T})$, for a certain constant $S_{T-1}'$;
	\item and the quantity $\sqrt N(\inv N \sum \tilde K_T^N(n, \mathbbm1) - \Q_{T-1}L_{T-1:T}\mathbbm1)$, which is a special case of the former one, with parameters $(1+C_{T-1}, S_{T-1}')$.
\end{itemize}
Finally, we invoke Proposition~\ref{prop:apx:hoeffding} and deduce the
sub-Gaussian property of~\eqref{eq:delta_tn} with parameters
\[\left(2+2C_{T-1}, 2\frac{S'_{T-1}\infnorm{\varphi_T}}{\Q_{T-1}L_{T-1:T}\mathbbm1}\right) \]
which finishes the proof.
\end{proof}

\begin{proof}[Proof of Proposition~\ref{prop:fundamental}]
	From~\eqref{eq:def_joint_qt}, we have
	\begin{equation*}
		\begin{split}
		\Q_t^N(\dd x_{0:t}) &= \sum_{i_t} \bar \Q_t^N(\dd i_t) \bar \Q_t^N(\dd x_{0:t}|i_t) \\
		&= \nmt \sum_{i_t} G_t(X_t^{i_t}) \bar \Q_t^N(\dd x_{0:t}|i_t) \\
		&= \nmt \sum_{i_t} G_t(x_t) \bar\Q_t^N(\dd x_{0:t}|i_t)
		\end{split}
	\end{equation*}
	since $\bar \Q_t^N(dx_{0:t}|i_t)$ has a $\delta_{X_t^{i_t}}(\dd x_t)$ term. In fact, the identity
	\[\qtnb t (\ddxz t, \dd i_{t-1}|i_t)= \delta_{X_t^{i_t}}(\dd x_t) B_t^N(i_t, \dd i_{t-1}) \qtnb{t-1}(\ddxz{t-1}|i_{t-1}) \]
	follows directly from the backward recursive nature of Algorithm~\ref{algo:offline_generic}, and thus
	\begin{equation}
		\label{eq:backward_qbar}
		\qtnb t (\ddxz t|i_t)= \delta_{X_t^{i_t}}(\dd x_t) \int_{i_{t-1}} B_t^N(i_t, \dd i_{t-1}) \qtnb{t-1}(\ddxz{t-1}|i_{t-1}).
	\end{equation}
	The $\Q_{t-1}^N(\ddxz{t-1}|i_{t-1})$ term is $\FTP{t-1}$-measurable. We shall calculate the expectation of $\delta_{X_t^{i_t}}(\dd x_t) B_t^N(i_t, \dd i_{t-1})$ given $\FTP{t-1}$. The following arguments are necessary for formal verification, but the result~\eqref{eq:proof_fundamental_expectation_nuclear} is natural in light of the ancestor regeneration intuition explained in Section~\ref{sec:validity}.
	
	Let $f_t^N: \px{1, \ldots, N} \times \mathcal X_t \to \mathbb R$ be a (possibly random) function such that $f_t^N(i_{t-1}, x_t)$ is $\FTP{t-1}$-measurable. Let $J_t^{i_t}$ be a random variable such that given $\FTP{t-1}$, $X_t^{i_t}$ and $\hat B_t^N(i_t, \cdot)$, $J_t^{i_t}$ is $B_t^N(i_t, \dd i_{t-1})$-distributed. This automatically makes $J_t^{i_t}$ satisfy the second hypothesis of Theorem~\ref{thm:convergence_mcmc}. Additionally, by virtue of its first hypothesis, the distribution of $(J_t^{i_t}, A_t^{i_t})$ is the same given either $\FTP{t-1}$ or $X_{t-1}^{1:N}$ (see also Figure~\ref{fig:thm1:variables}). We can now write
	\begin{equation*}
		\begin{split}
		&\ceftmone{\int f_t^N(i_{t-1}, x_t) \delta_{X_t^{i_t}}(\dd x_t) B_t^N(i_t, \dd i_{t-1})} \\
		=& \ceftmone{\int f_t^N(i_{t-1}, X_t^{i_t}) B_t^N(i_t, \dd i_{t-1})}\\
		=& \ceftmone{\CE{f_t^N(J_t^{i_t}, X_t^{i_t})}{\FTP{t-1}, X_t^{i_t}, \hat B_t^N(i_t, \cdot)}} \\
		=& \ceftmone{f_t^N(J_t^{i_t}, X_t^{i_t})} \text{ by the law of total expectation}\\
		=& \ceftmone{f_t^N(A_t^{i_t}, X_t^{i_t})} \text{ by the second hypothesis of Theorem~\ref{thm:convergence_mcmc}} \\
		=& \int f_t^N(i_{t-1}, x_t) \mathcal M(W_{t-1}^{1:N})(\dd i_{t-1}) M_t(X_{t-1}^{i_{t-1}}, \dd x_t).
		\end{split}
	\end{equation*}
	This equality means that
	\begin{equation}
		\label{eq:proof_fundamental_expectation_nuclear}
		\ceftmone{\delta_{X_t^{i_t}}(\dd x_t) B_t^N(i_t, \dd i_{t-1})} = \mathcal M(W_{t-1}^{1:N})(\dd i_{t-1}) M_t(X_{t-1}^{i_{t-1}}, \dd x_t),
	\end{equation}
	 Now, put
	\[K^N(i_t, \ddxz t)\eqdef \qtnb t(\ddxz t|i_t). \]
	From~\eqref{eq:backward_qbar} and~\eqref{eq:proof_fundamental_expectation_nuclear}, we have
	\begin{equation*}
		\begin{split}
		\ceftmone{K^N(i_t, \dd x_{0:t})} &= \int _{i_{t-1}} \mathcal M(W_{t-1}^{1:N})(\dd i_{t-1}) M_t(X_{t-1}^{i_{t-1}}, \dd x_t) \qtnb{t-1}(\ddxz{t-1}|i_{t-1}) \\
		&= M_t(x_{t-1}, \dd x_t) \int_{i_{t-1}}\mathcal M(W_{t-1}^{1:N}) (\dd i_{t-1}) \qtnb{t-1}(\ddxz{t-1}|i_{t-1})\\
		&\text{since } \qtnb{t-1}(\ddxz{t-1}|i_{t-1}) \text{ has a } \delta_{X_{t-1}^{i_{t-1}}(\dd x_{t-1})} \text{ term} \\
		&= M_t(x_{t-1}, \dd x_t) \Q_{t-1}^N(\dd x_{0:{t-1}})
		\end{split}
	\end{equation*}
	which finishes the proof.
\end{proof}

\subsection{\proe  \eqref{eq:generic_online_recursion_hard} (online smoothing recursion)}
\label{proof:full_online_recursion}

\begin{proof}
	Using~\eqref{eq:joint_empirical_smoothing} and the matrix notations, the distribution $\bar \Q_t^N(\dd i_s)$ can be represented by the $1 \times N$ vector
	\[\hat q_{s|t}^N := [W_t^1 \ldots W_t^N] \hat B_t^N \ldots \hat B_{s+1}^N.\]
	Defining the $N \times N$ matrix $\hat \psi_s^N$ as 
    \[\hat \psi_s^N [i_{s-1}, i_s] := \psi_s(X_{s-1}^{i_{s-1}}, X_s^{i_s}),\]
	we have
	\begin{align*}
		\E_{\Q_t^N}[\psi_s(X_{s-1}, X_s)] &= \sum_{i_s, i_{s-1}} \hat q_{s|t}^N[1, i_s] \hat B_s^N[i_s, i_{s-1}] \hat \psi_s^N [i_{s-1}, i_s] \\
		&= \sum_{i_s} \hat q_{s|t}^N[1, i_s] (\hat B_s^N \hat \psi_s^N) [i_s, i_s]\\
		&= \hat q_{s|t}^N \operatorname{diag}(\hat B_s^N \hat \psi_s^N).
	\end{align*}
	Therefore,
	\[\Q_t^N \varphi_t = \sum_{s=0}^t [W_t^1 \ldots W_t^N] \hat B_t^N \ldots \hat B_{s+1}^N \operatorname{diag}(\hat B_s^N \hat \psi_s^N) \]
	from which follows the recursion
	\begin{equation*}
	\begin{cases}
		\Q_t^N \varphi_t &\ = [W_t^1 \ldots W_t^N] \hat S_t^N, \\
		\hat S_t^N &:= \hat B_t^N \hat S_{t-1}^N + \operatorname{diag}(\hat B_t^N \hat \psi_t^N).
	\end{cases}
	\end{equation*}
	This is exactly~\eqref{eq:generic_online_recursion_hard}.
\end{proof}

\subsection{\prot~\ref{thm:stability} (general stability theorem)}
\label{proof:thm:stability}
The following lemma describes the simultaneous backward construction of two trajectories $\mathcal I_{0:T}^1$ and $\mathcal I_{0:T}^2$ given $\mathcal F_T^-$.
\begin{lem}\label{lem:two_backward_traj}
	We use the same notations as in Algorithms~\ref{algo:bootstrap} and~\ref{algo:offline_generic}. Suppose that the hypotheses of Theorem~\ref{thm:convergence_mcmc} are satisfied. Then, given $\mathcal I_{t:T}^1$, $\mathcal I_{t:T}^2$ and $\mathcal F_T^-$,
	\begin{itemize}
\item if $\mathcal I_t^1 \neq \mathcal I_t^2$, the two variables $\mathcal I_{t-1}^1$ and $\mathcal I_{t-1}^2$ are conditionally independent and their marginal distributions are respectively $B_t^{N, \mathrm{FFBS}} (\mathcal I_t^1, \cdot)$ and $B_t^{N, \mathrm{FFBS}}(\mathcal I_t^2, \cdot)$;
\item if $\mathcal I_t^1 = \mathcal I_t^2$, under the aforementioned
    conditioning, the two variables $\mathcal I_{t-1}^1$ and $\mathcal
    I_{t-1}^2$ are both marginally distributed according to
    $\btn{FFBS}(\mathcal I_t^1, \cdot)$. Moreover, if~\eqref{eq:support_cond} holds, we have 
\begin{equation}
\label{eq:backward_different_proba_lower_bound}
\CProb{\mathcal I_{t-1}^1 \neq \mathcal I_{t-1}^2}{\mathcal I_{t:T}^{1,2}, \mathcal F_T^-} \mathbbm 1_{\mathcal I_t^1 = \mathcal I_t^2} \geq \epss \  \mathbbm 1_{\mathcal I_t^1 = \mathcal I_t^2}.
\end{equation}
\end{itemize} 
In particular, the sequence of variables $(\mathcal I_{T-s}^1, \mathcal I_{T-s}^2)_{s=0}^T$ is a Markov chain given $\FTM$.
\end{lem}
\begin{proof}
	To simplify the notations, let $\tilde b_t^n$ denote the $\mathbb R^n$ vector $\hat B_t^N(n, \cdot)$. The relation between variables generated by Algorithm~\ref{algo:offline_generic} is depicted as a graphical model in Figure~\ref{fig:alg2:variables}. We consider
\begin{figure}
\centering
	\begin{tikzpicture}[node distance = {13mm}]
		\node (xtm2) {$x_{t-2}^{1:N}$};
		\node (xtm1) [below of=xtm2] {$x_{t-1}^{1:N}$};
		\node (xt) [below of=xtm1] {$x_t^{1:N}$};
		\node (xtp1) [below of=xt] {$x_{t+1}^{1:N}$};
		\node (xldots) [below of=xtp1] {$\ldots$};
		\node (xT) [below of=xldots] {$x_T^{1:N}$};
		
		\node (bldots) [right of=xtm2] {$\ldots$};
		\node (btm1) [below of=bldots] {$\tilde b_{t-1}^{1:N}$};
		\node (bt) [below of=btm1] {$\tilde b_{t}^{1:N}$};
		\node (btp1) [below of=bt] {$\tilde b_{t+1}^{1:N}$};
		\node (bldots2) [below of=btp1] {$\ldots$};
		\node (bT) [below of=bldots2] {$\tilde b_{T}^{1:N}$};
		
		\node (iT) [below of=bT] {$i_T^{1:2}$};
		
		\node (ildots) [right of=bldots] {$\ldots$};
		\node (itm2) [below of=ildots] {$i_{t-2}^{1:2}$};
		\node (itm1) [below of=itm2] {$i_{t-1}^{1:2}$};
		\node (it) [below of=itm1] {$i_{t}^{1:2}$};
		\node (ildots2) [below of=it] {$\ldots$};
		\node (iTm1) [below of=ildots2] {$i_{T-1}^{1:2}$};
		
		\draw[->] (xtm2) -- (xtm1);
		\draw[->] (xtm1) -- (xt);
		\draw[->] (xt) -- (xtp1);
		\draw[->] (xtp1) -- (xldots);
		\draw[->] (xldots) -- (xT);
		
		\draw[->] (xtm2) -- (bldots);
		\draw[->] (bldots) -- (ildots);
		\draw[->] (xtm1) -- (btm1);
		\draw[->] (btm1) -- (itm2);
		\draw[->] (xt) -- (bt);
		\draw[->] (bt) -- (itm1);
		\draw[->] (xtp1) -- (btp1);
		\draw[->] (btp1) -- (it);
		\draw[->] (xldots) -- (bldots2);
		\draw[->] (bldots2) -- (ildots2);
		\draw[->] (xT) -- (bT);
		\draw[->] (bT) -- (iTm1);
		
		\draw[->] (xtm2) -- (btm1);
		\draw[->] (xtm1) -- (bt);
		\draw[->] (xt) -- (btp1);
		\draw[->] (xtp1) -- (bldots2);
		\draw[->] (xldots) -- (bT);
		
		\draw[->] (iTm1) -- (ildots2);
		\draw[->] (ildots2) -- (it);
		\draw[->] (it) -- (itm1);
		\draw[->] (itm1) -- (itm2);
		\draw[->] (itm2) -- (ildots);
		
        \draw[->] (xT) -- (iT); \draw[->] (iT) -- (iTm1);
    \end{tikzpicture}
    \caption{Directed graph representing the relations between variables
        generated in Algorithm~\ref{algo:offline_generic}. Only those necessary
    for the proof of Lemma~\ref{lem:two_backward_traj} are included.}
    \label{fig:alg2:variables}
\end{figure}
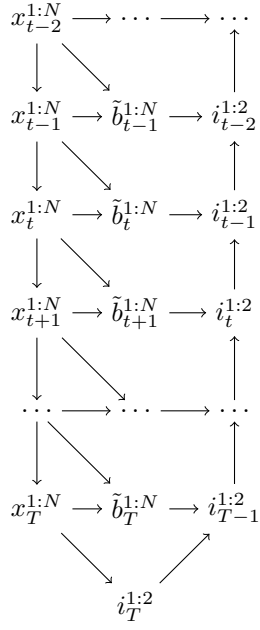
\begin{equation}
    \label{eq:2:bwk:idx}
    \begin{split}
        p(\tilde b_t^{1:N}, i_{t-1}^{1:2} | \mathcal F_T^-, i_{t:T}^{1:2}) &= p(\tilde b_t^{1:N} | \mathcal F_T^-, i_{t:T}^{1:2}) \  p(i_{t-1}^{1:2}| \tilde b_t^{1:N}, \mathcal F_T^-, i_{t:T}^{1,2})  \\
                                                                           &= p(\tilde b_t^{1:N} | x_{t-1}^{1:N}, x_t^{1:N})\  p(i_{t-1}^{1:2} | \tilde b_t^{1:N}, i_t^{1:2})  \\
                                                                           &\textrm{(by properties of graphical models, see
                                                           Figure~\ref{fig:alg2:variables})}  \\
                                                                           &= \ps{\prod_n p(\tilde b_t^n | x_{t-1}^{1:N}, x_t^n)} \tilde b_t^{i_t^1}(i_{t-1}^1) \tilde b_t^{i_t^2}(i_{t-1}^2). 
    \end{split}
\end{equation}
The distribution of $i_{t-1}^1$ given $\mathcal F_T^-$ and $i_{t:T}^{1:2}$ is
thus the $i_{t-1}^1$-marginal of
\begin{equation}
    \label{eq:2traj:hidden_jt}
    p(\tilde b_t^{i_t^1} | x_{t-1}^{1:N}, x_t^{i_t^1}) \tilde
    b_t^{i_t^1}(i_{t-1}^1),
\end{equation}
which is exactly the distribution of $p(j_t^{i_t^1}|x_{t-1}^{1:N},
x_t^{i_t^1})$ where the $J$'s are defined in the statement of
Theorem~\ref{thm:convergence_mcmc}. By the second hypothesis of that theorem,
the aforementioned distribution is equal to $p(a_t^{i_t^1} | x_{t-1}^{1:N},
x_t^{i_t^1})$, which is in turn no other than 
$B_t^{N, \mathrm{FFBS}}(i_t^1, \cdot)$. Moreover, if $i_t^1 \neq i_t^2$, \eqref{eq:2:bwk:idx}
straightforwardly implies the conditional independence of $i_{t-1}^1$ and
$i_{t-1}^2$. When $i_t^1 = i_t^2$, the distribution of $i_{t-1}^{1:2}$ given
$\FTM$ and $i_{t:T}^{1:2}$ is the $i_{t-1}^{1:2}$-marginal of
\[p(\tilde b_t^{i_t^1} | x_{t-1}^{1:N}, x_t^{i_t^1}) \tilde
b_t^{i_t^1}(i_{t-1}^1) \tilde b_t^{i_t^1}(i_{t-1}^2). \] Thus, we can 
apply~\eqref{eq:support_cond} for $n=i_t^1$, where $i_{t-1}^{1:2}$
here plays the role of $J_t^{1:2}$ there. 
Equation~\eqref{eq:backward_different_proba_lower_bound} is now proved.
\end{proof}

As Lemma~\ref{lem:two_backward_traj} describes the distribution of two
trajectories, it immediately gives the distribution of a single trajectory.
\begin{corollary}\label{cor:one_backward_traj}
    Under the same settings as in Lemma~\ref{lem:two_backward_traj}, given $\mathcal
    F_T^-$, the distribution of $\mathcal I_{0:T}^1$ is
    \[\mathcal M(W_{T}^{1:N})(\dd i_T) B_T^{N, \mathrm{FFBS}}(i_T, \dd i_{T-1})
    \ldots B_1^{N, \mathrm{FFBS}}(i_1, \dd i_0). \]
\end{corollary}
Note that the corollary applies even if the backward kernel used in
Algorithm~\ref{algo:offline_generic} is \emph{not} the FFBS one. This is due to the
conditioning on $\mathcal F_T^-$ and the second hypothesis of
Theorem~\ref{thm:convergence_mcmc}.

\begin{proof}[Proof of Theorem~\ref{thm:stability}]

First of all, we remark that as per Algorithm~\ref{algo:offline_generic}, using
index variables $\mci_{0:T}^{1:N}$ adds a level of Monte Carlo
approximation to $\Q_T^N(\dd x_{0:T})$. Therefore
\begin{align}
    \nonumber
    \E\ps{(\Q_T^N(\varphi_T) - \Q_T(\varphi_T))^2} 
    & = \E\ps{\pr{ \frac 1N \sum_{n=1}^N \varphi_T(X_0^{\mci_0^n}, \ldots, X_T^{\mci_T^n})
        - \Q_T(\varphi_T)}
    ^2} \\
    \label{eq:stb_err_decomp}
    & = \E\ps{(\Q_T^{N, \mathrm{FFBS}}(\varphi_T) - \Q_T(\varphi_T))^2} + \\
    \nonumber
    & \quad + \E\ps{\CVar{\frac 1N \sum_{n=1}^N \varphi_T(X_0^{\mci_0^n}, \ldots, X_T^{\mci_T^n})}{\FTM}}
\end{align}
where the ultimate inequality is justified by the law of total expectation and
Corollary~\ref{cor:one_backward_traj}. (Note that $(\mci_{0:T}^n)_{n=1}^N$ are
identically distributed but \textit{not} necessarily independent given $\FTM$.)
Using Lemma~\ref{lem:cond_covar_two_traj} (stated and proved below) and putting 
$\rho\eqdef 1-\mblow/\mbhigh$, we have
\begin{equation}\label{eq:stb_cond_error}
\begin{aligned}
\MoveEqLeft\CVar{\frac 1N \sum_{n=1}^N \varphi_T(X_0^{\mci_0^n}, \ldots, X_T^{\mci_T^n})}{\FTM} \\
& = \CVar{\frac 1N \sum_{n=1}^N \sum_{t=0}^T \addpsi{t}{n}}{\FTM} \\
& =	 \frac{1}{N^2} \sum_{n,m\leq N} \sum_{s,t\leq T} \CCov{\addpsi{t}{n}}{\addpsi{s}{m}}{\FTM} \\
& \leq \frac{2}{N^2} \sum_{\substack{n,m \leq N\\n=m}} \sum_{s, t \leq T} \infnorm{\psi_t} \infnorm{\psi_s} \rho^{\abs{t-s}-1} + \\
& \quad + \frac{4}{N^2} \sum_{\substack{n,m\leq N\\n\neq m}} \sum_{s,t\leq T} \frac{\tilde C}{N}\infnorm{\psi_t}\infnorm{\psi_s}\rho^{\abs{t-s}-1} \\
& = \pr{\sum_{s,t\leq T} 2\infnorm{\psi_t}\infnorm{\psi_s} \rho^{\abs{t-s} - 1}} \frac{(2\tilde C+1)N-2\tilde C}{N^2} \\
&\leq  \ps{\sum_{s,t\leq T} \pr{\infnorm{\psi_t}^2 + \infnorm{\psi_s}^2} \rho^{\abs{t-s}-1}} \frac{2\tilde C+1}{N} \leq \frac{4(2\tilde C+1)}{N\rho (1-\rho)} \sum \infnorm{\psi_t}^2.
	\end{aligned}
\end{equation}
We now look at the first term of \eqref{eq:stb_err_decomp}. In the fixed marginal smoothing case, for any $s \in \mathbb Z_+$, $s\leq T$ and any function $\phi_s: \mathcal X_s \to \mathbb R$, \citet{Douc2011} proved that
\[\P\pr{\abs{\Q_T^{N, \mathrm{FFBS}}(\varphi_T) - \Q_T(\varphi_T)} \geq \varepsilon} \leq B'e^{-C'N\varepsilon^2/\norminf{\phi_s}^2} \]
for $\varphi_T(x_{0:T}) = \phi_s(x_s)$ and constants $B'$ and $C'$ not depending on $T$. Using $\E[\Delta^2]=\int_0^\infty \P(\Delta^2 \geq t) \dd t$, the inequality implies
\begin{equation}
	\label{eq:previous_ffbs_bound_1}
	\E\ps{\abs{\qtnffbs(\varphi_T) - \Q_T(\varphi_T)}^2} \leq \frac{B'\infnorm{\phi_s}^2}{C'N}
\end{equation}
for $\varphi_T(x_{0:T})=\phi_s(x_s)$. In the additive smoothing case, \citet{DubarryLeCorff2011} proved that, for $T \geq 2$,
\begin{equation}
	\label{eq:previous_ffbs_bound_2}
	\E\ps{\abs{\qtnffbs(\varphi_T) - \Q_T(\varphi_T)}^2} \leq \frac{C'}{N}\pr{\sum_{t=0}^T \norminf{\psi_t}^2} \pr{1 + \sqrt{\frac TN}}^2.
\end{equation}
Equations~\eqref{eq:previous_ffbs_bound_1}, \eqref{eq:previous_ffbs_bound_2}, \eqref{eq:stb_cond_error} and \eqref{eq:stb_err_decomp} conclude the proof.
\end{proof}

The following lemma quantifies the backward mixing property induced by Assumption~\ref{asp:mt_2ways_bound}.
\begin{lem}\label{lem:backward_mixing}
	Under the same setting as Theorem~\ref{thm:stability}, we have
	\[\operatorname{TV}\pr{\btn{FFBS}(m, \cdot), \btn{FFBS}(n, \cdot)} \leq 1 - \frac{\mblow}{\mbhigh} \]
	for all $m, n \in \px{1, \ldots, N}$ and $t \in \px{1, \ldots, T}$.
\end{lem}
\begin{proof} We have
	\begin{align*}
\MoveEqLeft 1 - \operatorname{TV}\pr{\btn{FFBS}(m, \cdot), \btn{FFBS}(n, \cdot)} \\
& = \left[ \sum_{i=1}^N \min \left( \bwdist{m}, \right. \right. \\
& \qquad \left. \left. \bwdist{n} \right) \right] \text{ by
Lemma~\ref{lem:properties_TV} (Supplement~\ref{apx:tv})} \\
& \geq \ps{\sum_{i=1}^N \frac{G_t(X_{t-1}^i) \mblow}{\sum_{j=1}^N G_t(X_{t-1}^j) \mbhigh}}  \text{ by Assumption~\ref{asp:mt_2ways_bound}}  \\
& = (\mblow/\mbhigh).
\end{align*}
\end{proof}

\begin{lem}\label{lem:cond_covar_two_traj}
	 Under the same settings as in Theorem~\ref{thm:stability}, for any $m,n \in \px{1, \ldots, N}$ and $s, s' \in \px{0, \ldots, T}$, we have
	 \begin{multline}\label{eq:small_corr_lemma_statement}
	\ccovftm{\psi_s(X_{s-1}^{\mci_{s-1}^m}, X_s^{\mci_s^m})}{\psi_{s'}(X_{s'-1}^{\mci_{s'-1}^n}, X_{s'}^{\mci_{s'}^n})} \\
	\leq \quad
	2\pr{1 - \frac{\mblow}{\mbhigh}}^{\abs{s-s'}-1} \infnorm{\psi_s}\infnorm{\psi_{s'}} \times
	\begin{cases}
	 	2 \tilde C / N &\text{ if } m\neq n\\
	 	1 &\text{ if } m = n
	\end{cases}
	\end{multline}
    where $\tilde C=\tilde C(\mblow, \mbhigh, \gblow, \gbhigh, \epss)$ is a
    constant that does not depend on $T$ (and which arises in the formulation
    of Lemma~\ref{lem:I_s_star_likely_same}). If $s$ or $s'$ is equal to $0$,
    we adopt the natural convention $\psi_0(x_{-1}, x_0) \eqdef \psi_0(x_0)$.
\end{lem}
\begin{proof}
	We first handle the case $m \neq n$. Without loss of generality, assume that $m=1$, $n=2$ and $s \geq s'$. The covariance bound of Lemma~\ref{lem:properties_TV} yields
	\begin{multline}\label{eq:small_corr_to_small_tv}
		\ccovftm{\psi_s(X_{s-1}^{\mcio_{s-1}}, X_s^{\mcio_s})}{\psi_{s'}(X_{s'-1}^{\mcit_{s'-1}}, X_{s'}^{\mcit_{s'}})}  \\
		\leq  2\infnorm{\psi_s} \infnorm{\psi_{s'}} \operatorname{TV}\pr{(\mcio_{s-1:s}, \mcit_{s'-1:s'})|\FTM, (\mcio_{s-1:s}|\FTM) \otimes (\mcit_{s'-1:s'}|\FTM)}.
	\end{multline}
	We shall bound this total variation distance via the coupling inequality of
    Lemma~\ref{lem:properties_TV} (Supplement~\ref{apx:tv}). The idea is to construct, in addition to
    $\mcio_{0:T}$ and $\mcit_{0:T}$, two trajectories $\mciso_{0:T}$ and
    $\mcist_{0:T}$ i.i.d.\ given $\FTM$ such that each of them is conditionally distributed according to $\mcio_{0:T}$ (cf. Corollary~\ref{cor:one_backward_traj}). To make the coupling inequality efficient, it is desirable to make $\mcio_{0:T}$ and $\mciso_{0:T}$ as similar as possible (same thing for $\mcit_{0:T}$ and $\mcist_{0:T}$).
	
The detailed construction of the four trajectories $\mcio_{0:T}$, $\mcit_{0:T}$, $\mciso_{0:T}$ and $\mcist_{0:T}$ \textit{given} $\FTM$ is described in Algorithm~\ref{algo:four_trajs}. In particular, we ensure that $\forall t \geq s-1$, we have $\mcio_t = \mciso_t$. For $t \leq s-1$, if $\mcit_t = \mcist_t$, it is guaranteed that $\mcit_\ell = \mcist_{\ell}$ holds $\forall \ell \leq t$. The rationale for different coupling behaviours between the times $t \geq s-1$ and $t \leq s-1$ will become clear in the proof: the former aim to control the correlation between two different trajectories $m=1$ and $n=2$ and result in the $\tilde C/N$ term of \eqref{eq:small_corr_lemma_statement}; the latter are for bounding the correlation between two times $s$ and $s'$ and result in the $(1-\mblow/\mbhigh)^{|s-s'| - 1}$ term of the same equation.

\begin{algo}{Sampler for the variables $\mcio_{0:T}$, $\mcit_{0:T}$, $\mciso_{0:T}$ and $\mcist_{0:T}$ (see proof of Lemma~\ref{lem:cond_covar_two_traj})}
		\label{algo:four_trajs}
		\KwIn{Feynman-Kac model~\eqref{eq:fkmodel}, variables $X_{0:T}^{1:N}$ from the output of Algorithm~\ref{algo:bootstrap}, integer $s \geq 0$ (see statement of Lemma~\ref{lem:cond_covar_two_traj})}
		Sample $\mcio_T, \mcit_T \iid \mathcal M(W_T^{1:N})$\;
		Set $\mciso_T \gets \mcio_T$ and $\mcist_T \gets \mcit_T$\;
		\For{$t \gets T$ \KwTo $1$}{
			\If{$\mcio_t \neq \mcit_t$}{
				\For{$k \in \px{1,2}$}{
					Sample $(\mci_{t-1}^k, \mci^{*k}_{t-1})$ from any maximal coupling of $B_t^{N, \mathrm{FFBS}}(\mci_t^k, \cdot)$ and $\btn{FFBS}(\mci_t^{*k}, \cdot)$ (cf. Lemma~\ref{lem:properties_TV})
				}
			}
			\Else{
				Sample the $\mathbb R^N$ vector $\hat B_t^N(\mci_t^1, \cdot)$ from $p(\hat b_t^N(i_t^1, \cdot) | x_{t-1}^{1:N}, x_t^{i_t^1})$\;
				Sample $\mci_{t-1}^1, \mcit_{t-1} \iid \hat B_t^N(\mcio_t, \cdot)$\;
				Set $k\gets 1, \ell \gets 2$ if $t \geq s$ and $k \gets 2, \ell \gets 1$ otherwise\;
				Sample $\mci_{t-1}^{*k} \sim \btn{FFBS}(\mci_t^{*k}, \cdot)$ such that $(\mci_{t-1}^{*k}, \mci_{t-1}^k)$ is any maximal coupling of $\btn{FFBS}(\mci_t^{*k},\cdot)$ and $\btn{FFBS}(\mci_t^k, \cdot)$ given $\mci_{t:T}^{1:2}$, $\mci_{t:T}^{*1:2}$ and $\FTM$ ($(\star)$ - see text for validity of this step)\;
				Sample $\iitmos \ell \sim \btn{FFBS}(\iits \ell, \cdot)$\;
			}
		}
	\KwOut{Four trajectories $\mci_{0:T}^1$, $\mcit_{0:T}$, $\mciso_{0:T}$, $\mcist_{0:T}$ to be used in the proof of Lemma~\ref{lem:cond_covar_two_traj}}
	\end{algo}

The correctness of Algorithm~\ref{algo:four_trajs} is asserted by
Lemma~\ref{lem:two_backward_traj}. Step $(\star)$  is valid because that
lemma states that the distribution of $\iitmo k$ given $\FTM$,
$\mci_{t:T}^{1,2}$ and $\mci_{t:T}^{*1,2}$ is $\btn{FFBS}(\iit k, \cdot)$.
Furthermore, we note that $(R_{T-t})_{t=0}^T$ where 
\[R_t \eqdef (\mci_{t}^1, \mci_{t}^2, \mci^{*1}_{t}, \mci_{t}^{*2}),\]  
is a Markov chain given $\FTM$.

From \eqref{eq:small_corr_to_small_tv}, applying the coupling inequality of Lemma~\ref{lem:properties_TV} gives
\begin{multline}\label{eq:cov_to_diff_stb_prf}
	\ccovftm{\psi_s(X_{s-1}^{\mcio_{s-1}}, X_s^{\mcio_s})}{\psi_{s'}(X_{s'-1}^{\mcit_{s'-1}}, X_{s'}^{\mcit_{s'}})} \\
	\leq 2\infnorm{\psi_s} \infnorm{\psi_{s'}} \pftm{(\mcio_{s-1:s}, \mcit_{s'-1:s'}) \neq (\mciso_{s-1:s}, \mcist_{s'-1:s'})} \\
	= 2\infnorm{\psi_s} \infnorm{\psi_{s'}} \pftm{\mcit_{s'-1:s'} \neq \mcist_{s'-1:s'}}
\end{multline}
where the last equality results from the construction of Algorithm~\ref{algo:four_trajs}. The sub-case $s=s'$ following directly from Lemma~\ref{lem:I_s_star_likely_same}, we now focus on the sub-case $s \geq s'+1$. For all $t \leq s-1$,
\begin{equation}\label{eq:two_twostar_diff_bw}
	\begin{aligned}
	\MoveEqLeft\pftm{\trajtwo_{t-1} \neq \trajtwos_{t-1}} \\
    & = \pftm{\trajtwo_{t-1} \neq \trajtwos_{t-1}, \trajtwo_t \neq \trajtwos_t}\\ &\text{ by construction of Algorithm~\ref{algo:four_trajs}} \\
	&=\CE{\prtf{\trajtwo_{t-1}\neq \trajtwos_{t-1}, \trajtwo_t \neq \trajtwos_t}}{\FTM} \\ 
	&\text{ by the law of total expectation}\\
	&= \CE{\operatorname{TV}\pr{\btn{FFBS}(\trajtwo_t, \cdot), \btn{FFBS}(\trajtwos_t, \cdot)} \ind\px{\trajtwo_t \neq \trajtwos_t}}{\FTM} \\
	&\leq \oneminusmlmh \pftm{\trajtwo_t \neq \trajtwos_t} \text{ by Lemma~\ref{lem:backward_mixing}}.
	\end{aligned}
\end{equation}

Thus \begin{align*}
	\MoveEqLeft\pftm{\trajtwo_{s'-1:s'} \neq \trajtwos_{s'-1:s'}} \\
	&=\pftm{\trajtwo_{s'} \neq \trajtwos_{s'}} \text{ by construction of Algorithm~\ref{algo:four_trajs}} \\
    & \leq \oneminusmlmh^{s-s'-1}\pftm{\trajtwo_{s-1} \neq \trajtwos_{s-1}} \text{ by applying \eqref{eq:two_twostar_diff_bw} recursively}\\
	&\leq \oneminusmlmh^{s-s'-1} \frac{\tilde C}{N} \text{ by Lemma~\ref{lem:I_s_star_likely_same},}
\end{align*}
which, combined with \eqref{eq:cov_to_diff_stb_prf} finishes the proof for the current sub-case $s \geq s'+1$. It remains to show \eqref{eq:small_corr_lemma_statement} when $m=n$. The proof follows the same lines as in the case $m\neq n$, although we shall briefly outline some arguments to show how the factor $\tilde C/N$ disappeared. The case $s=s'$ being trivial, suppose that $s\geq s'+1$ and without loss of generality that $m=n=3$. To use the coupling tools of Lemma~\ref{lem:properties_TV}, we construct trajectories $\trajthree_{0:T}$, $\trajthrees_{0:T}$ and $\trajfours_{0:T}$ via Algorithm~\ref{algo:three_trajs} and write, in the spirit of \eqref{eq:cov_to_diff_stb_prf}:

\begin{algo}{Sampler for the variables $\trajthree_{0:T}$, $\trajthrees_{0:T}$ and $\trajfours_{0:T}$ (see proof of Lemma~\ref{lem:cond_covar_two_traj})}
	\label{algo:three_trajs}
		\KwIn{Feynman-Kac model \eqref{eq:fkmodel}, variables $X_{0:T}^{1:N}$ from the output of Algorithm~\ref{algo:bootstrap}, integer $s \geq 0$ (see statement of Lemma~\ref{lem:cond_covar_two_traj})}
		Sample $\trajthrees_T, \trajfours_T \iid \mathcal{M}(W_T^{1:N})$\;
		Set $\trajthree_T \gets \trajthrees_T$\;
		\For{$t \gets T$ \KwTo 1}{
			\If{$t \geq s$}{
				Sample $\trajthrees_{t-1} \sim \btn{FFBS}(\trajthrees_t, \cdot)$ and $\trajfours_{t-1} \sim \btn{FFBS}(\trajfours_t, \cdot)$\;
				Set $\trajthree_{t-1} \gets \trajthrees_{t-1}$\;			
			}
			\Else{
				Sample $(\trajthree_{t-1}, \trajfours_{t-1})$ from a maximal coupling of $\btn{FFBS}(\trajthree_t, \cdot)$ and $\btn{FFBS}(\trajfours_t, \cdot)$\;
				Sample $\trajthrees_{t-1} \sim \btn{FFBS}(\trajthrees_t, \cdot)$
			}
		}
		\KwOut{Three trajectories $\trajthree_{0:T}$, $\trajthrees_{0:T}$ and $\trajfours_{0:T}$ to be used in the proof of Lemma~\ref{lem:cond_covar_two_traj}}
\end{algo}

\begin{multline}\label{eq:cov_via_tv_easy}
	\ccovftm{\psi_s(X_{s-1}^{\trajthree_{s-1}}, X_s^{\trajthree_s})}{\psi_{s'}(X_{s'-1}^{\trajthree_{s-1}}, X_{s'}^{\trajthree_{s'}})} \\
	\leq 2 \infnorm{\psi_s} \infnorm{\psi_{s'}} \pftm{(\trajthree_{s-1:s}, \trajthree_{s'-1:s'}) \neq (\trajthrees_{s-1:s}, \trajfours_{s'-1:s'})}\\
	= 2 \infnorm{\psi_s} \infnorm{\psi_{s'}} \pftm{\trajthree_{s'} \neq \trajfours_{s'}}
\end{multline}
where the last equality follows from the construction of Algorithm~\ref{algo:three_trajs} and the hypothesis $s \geq s'+1$. For all $t\leq s-1$, the inequality
\begin{equation}
	\label{eq:bw_easy_stb}
	\pftm{\trajthree_{t-1} \neq \trajfours_{t-1}} \leq \oneminusmlmh \pftm{\trajthree_t \neq \trajfours_t}
\end{equation}
can be proved using the same techniques as those used to prove \eqref{eq:two_twostar_diff_bw}: applying Lemma~\ref{lem:backward_mixing} given $(\trajthree_t, \trajthrees_t, \trajfours_t)$ then invoking the law of total expectation. Repeatedly instantiating \eqref{eq:bw_easy_stb} gives
\begin{align*}
	\pftm{\trajthree_{s'} \neq \trajfours_{s'}} &\leq \oneminusmlmh^{s-s'-1} \pftm{\trajthree_{s-1} \neq \trajfours_{s-1}}\\
	&\leq \oneminusmlmh^{s-s'-1}
\end{align*}
which, when plugged into \eqref{eq:cov_via_tv_easy}, finishes the proof.
\end{proof}

\begin{lem}
	\label{lem:I_s_star_likely_same}
	For $\mci_s^2$ and $\mci_s^{2*}$ defined by the output of Algorithm~\ref{algo:four_trajs}, we have
	\begin{align*} 
	\P\pr{\mcit_s \neq \mcist_s | \mathcal{F}_T^-} &\leq \tilde C/N, \text{ and } \\
	\pftm{\mcit_{s-1} \neq \mcist_{s-1}} &\leq \tilde C/N, \text{ if } s \geq 1, 
	\end{align*}
	for some constant $\tilde C=\tilde C(\mblow, \mbhigh, \gblow, \gbhigh, \epss)$.
\end{lem}
\begin{proof}
	Define $A_t \eqdef \ind\px{\iit 1 \neq \iit 2}$, $B_t \eqdef \ind\px{\iit 2
    = \iits 2}$ and $\Gamma_t \eqdef A_tB_t$ and recall that $R_t \eqdef (\iit
    1, \iit 2, \iits 1, \iits 2)$. The sequence $(R_{T-\ell})_{\ell=0}^T$ is a
    Markov chain given $\FTM$, but this is not necessarily the case for the
    sequence $(\Gamma_{T-\ell})_{\ell=0}^T$ of Bernoulli random variables.
    Nevertheless, Lemma~\ref{lem:stability_pseudo_markov} below shows that one can
    get bounds on two-step ``transition probabilities'' for
    $(\Gamma_{T-\ell})$, i.e.\ the probabilities under $\FTM$ that $\Gamma_{t-2}
    = 1$ given $\Gamma_t$ \textit{and} $R_t$. This motivates our following
    construction of actual Markov chains approximating the dynamic of $\Gamma_t$. Let $\gammastar_T$ and $\gammastar_{T-1}$ be two independent Bernoulli random variables given $\FTM$ such that
	\begin{equation}
	\label{eq:gammastar_initial}
	\begin{aligned}
		\CProb{\gammastar_T = 1}{\FTM} &= \CProb{\Gamma_T = 1}{\FTM} \\
		\CProb{\gammastar_{T-1} = 1}{\FTM} &=
		\CProb{\Gamma_{T-1} = 1}{\FTM}.
	\end{aligned}
	\end{equation}
	Let $\gammastar_T, \gammastar_{T-2}, \gammastar_{T-4}, \ldots$ and $\gammastar_{T-1}, \gammastar_{T-3}, \ldots$ be two homogeneous Markov chains given $\FTM$ with the same transition kernel $\gdbackward$ defined by
	\begin{equation}
	\label{eq:gammastar_transition}
	\begin{alignedat}{3}
		&\P_{\FTM}(\gammastar_{t-2}=1 | \gammastar_t = 1) &&= 1 - \frac 2N \hlbr &&=: \gdbackward_{11} \\
		&\P_{\FTM}(\gammastar_{t-2} = 1 | \gammastar_t=0) &&= \frac{\mblow \epss}{2\mbhigh} &&=: \gdbackward_{01}
	\end{alignedat}
	\end{equation}
	where for two events $E_1$, $E_2$, the notation $\P_{\FTM}(E_1 | E_2)$ is the ratio between $\CProb{E_1, E_2}{\FTM}$ and $\CProb{E_2}{\FTM}$. We shall now prove by backward induction the following statement:
	\begin{equation}
		\label{eq:pseudo_markov_bound_via_true_markov}
		\CProb{\Gamma_t=1}{\FTM} \geq \CProb{\gammastar_t = 1}{\FTM}, \forall t \geq s-1.
	\end{equation}
	Firstly, \eqref{eq:pseudo_markov_bound_via_true_markov} holds for $t=T$ and $t=T-1$. Now suppose that it holds for some $t \geq s+1$ and we wish to justify it for $t-2$. By Lemma~\ref{lem:stability_pseudo_markov},
	\begin{align*}
		\prtf{\Gamma_{t-2}=1} \ind_{\Gamma_t=1} &\geq \gdbackward_{11} \ind_{\Gamma_t=1}\\
		\prtf{\Gamma_{t-2}=1} \ind_{\Gamma_t=0} &\geq \gdbackward_{01} \ind_{\Gamma_t=0}.
	\end{align*}
	Applying the law of total expectation gives
	\begin{align*}
		\CProb{\Gamma_{t-2}=1}{\FTM} &\geq \gdbackward_{11} \CProb{\Gamma_t=1}{\FTM} + \gdbackward_{01} \CProb{\Gamma_t=0}{\FTM} \\
		&= \pr{\gdbackward_{11} - \gdbackward_{01}} \CProb{\Gamma_t = 1}{\FTM} + \gdbackward_{01}\\
		&\geq \pr{\gdbackward_{11} - \gdbackward_{01}} \CProb{\gammastar_t=1}{\FTM} + \gdbackward_{01} \\
		&\text{ if } N \text{ is large enough, by induction hypothesis}\\
		&= \CProb{\gammastar_{t-2}=1}{\FTM}
	\end{align*}
	and \eqref{eq:pseudo_markov_bound_via_true_markov} is now proved. To finish the proof of the lemma, it is necessary to lower bound its right hand side. We start by controlling the distribution $\gammastar_t$ for $t = T$ and $t = T-1$. We have
	\begin{equation}
			\label{eq:fpseudo_L1}
	\begin{aligned}
		\pftm{\gammastar_T=1} &= \pftm{\Gamma_T=1} \text{ by \eqref{eq:gammastar_initial}}  \\
		&= 1 - \pftm{A_T = 0} \text{ as } B_T = 1 \text{ by Algorithm~\ref{algo:four_trajs}} \\
		&= 1 - \sum_{i=1}^N \pftm{\mci_T^1 = \mci_T^2=i} \\
		&= 1 - \sum_{i=1}^N \pr{\frac{G(X_T^i)}{\sum_{j=1}^N G(X_T^j)}}^2  \\
		&\geq 1 - \frac 1N \pr{\frac{\gbhigh}{\gblow}}^2 \text{ by Assumption~\ref{asp:g_2ways_bound}}
	\end{aligned}
	\end{equation}
	and
	\begin{equation}
			\label{eq:fpseudo_L2}
	\begin{aligned}
		\pftm{\gammastar_{T-1}=1} &\geq \pftm{\Gamma_T = 1, \Gamma_{T-1}=1}  \\
		&= \CE{\CProb{\Gamma_T=1, \Gamma_{T-1}=1}{R_T, \FTM}}{\FTM} \\ &\text{ by the law of total expectation} \\
		&= \CE{\CProb{\Gamma_{T-1}=1}{R_T, \FTM} \ind_{\Gamma_T=1}}{\FTM}\\
		&\geq \ps{1 - \frac 1N \hlbr} \pftm{\Gamma_T=1} \text{ via \eqref{ieq_c_4trajs_backward_reg_bound}} \\
		&\geq \ps{1 - \frac 1N \hlbr} \ps{1-\frac 1N \pr{\frac{\gbhigh}{\gblow}}^2}.
	\end{aligned}
	\end{equation}
	The contraction property of Lemma~\ref{lem:properties_TV} makes it possible
    to relate the intermediate distributions $\gammastar_t|\FTM$ to the end
    point ones $\gammastar_{T-1}|\FTM$ and $\gammastar_T|\FTM$. More
    specifically, \eqref{eq:gammastar_transition} and
    Lemma~\ref{lem:properties_TV} lead to
	\begin{equation}
	\label{eq:pseudo_end_to_mid}
		\operatorname{TV}(\gammastar_t | \FTM, \mu^*) \leq \max\pr{\operatorname{TV}(\gammastar_T|\FTM, \mu^*), \operatorname{TV}(\gammastar_{T-1}|\FTM, \mu^*)}
	\end{equation}
	where $\mu^*$ is the invariant distribution of a Markov chain with transition matrix $\gdbackward$, namely
	\begin{equation}
	\label{eq:fpseudo_L3}
	\begin{cases}
		\mu^*(\px{0}) &= \frac{\gdbackward_{10}}{\gdbackward_{01}+\gdbackward_{10}} \\
		\mu^*(\px{1}) &= 1 - \mu^*(\px{0}).
	\end{cases}
	\end{equation}
	Furthermore, an alternative expression of the total variation distance given in Lemma~\ref{lem:properties_TV} implies that the total variation distance between two Bernoulli distributions of parameters $p$ and $q$ is $\abs{p-q}$. Combining this with \eqref{eq:pseudo_end_to_mid}, the triangle inequality and the rough estimate $\max(a,b) \leq a+b \ \forall a,b \geq 0$, we get
	\begin{equation*}
		\pftm{\gammastar_t=0} \leq 3\mu^*(\px{0}) + \pftm{\gammastar_T=0} + \pftm{\gammastar_{T-1}=0} \leq \tilde C/N
	\end{equation*}
	where $\tilde C=\tilde C(\mblow, \mbhigh, \gblow, \gbhigh, \epss)$. The last inequality is straightforwardly derived by plugging respectively \eqref{eq:fpseudo_L3}, \eqref{eq:fpseudo_L1} and \eqref{eq:fpseudo_L2} into the three terms of the preceding sum. This combined with \eqref{eq:pseudo_markov_bound_via_true_markov} finishes the proof.
\end{proof}

\begin{lem}
	\label{lem:stability_pseudo_markov}
	For $s$ defined in the statement of Lemma~\ref{lem:cond_covar_two_traj}; $A_t$, $B_t$ and $R_t$ defined in the proof of Lemma~\ref{lem:I_s_star_likely_same} and all $t \geq s + 1$, we have
\begin{align*}
\CProb{A_{t-2}B_{t-2}=1}{R_t, \FTM} \ind_{\atbt=1} & \geq \pr{1 - \frac 2N
\hlbr} \ind_{\atbt=1}; \\
\CProb{A_{t-2}B_{t-2}=1}{R_t, \FTM} & \geq \frac{\mblow \epss}{2 \mbhigh} \\
\end{align*} 
where the inequalities hold for $N$ large enough, i.e., $N \geq N_0 = N_0(\mblow, \mbhigh, \gblow, \gbhigh, \epss)$.
\end{lem}
\begin{proof}
	We start by showing the following three inequalities for all $t \geq s$ and $N$ sufficiently large:
	\begin{align}
		\label{ieq_a_4trajs_backward_reg_bound} \CProb{A_{t-1}=1}{R_t, \FTM} &\geq  \epss;\\
		\label{ieq_b_4trajs_backward_reg_bound}
		\CProb{A_{t-1}B_{t-1}=1}{R_t, \FTM} \ind_{A_t=1} &\geq (\mblow/2\mbhigh) \ind_{A_t=1};\\
		\label{ieq_c_4trajs_backward_reg_bound}
		\CProb{A_{t-1}B_{t-1}=1}{R_t, \FTM} \ind_{\atbt=1} &\geq \ps{1 - \frac 1N \pr{\frac{\gbhigh \mbhigh}{\gblow \mblow}}^2} \ind_{\atbt=1}.
	\end{align}
	For \eqref{ieq_a_4trajs_backward_reg_bound}, we have
	\begin{equation}
		\label{ieq_d_4trajs_bw_reg_bound}
		\CProb{A_{t-1}=1}{R_t, \FTM} \ind_{A_t \neq 1} = \CProb{\iitmo 1 \neq \iitmo 2}{R_t, \FTM} \ind_{\iit 1 = \iit 2} \geq \epss \ind_{A_t \neq 1}
	\end{equation}
	by Lemma~\ref{lem:two_backward_traj}. Next,
	\begin{equation}
	\begin{aligned}
		\MoveEqLeft\CProb{A_{t-1}=1}{R_t, \FTM} \ind_{A_t=1} \\
        & = \CProb{\iitmo 1 \neq \iitmo 2}{\rtftm} \ind_{\iit 1 \neq \iit 2}  \\
		& = \ps{1 - \sum_i \CProb{\iitmo 1 = \iitmo 2 = i}{\rtftm}} \ind_{\iit 1 \neq \iit 2}  \\
		& = \ps{1 - \sum_{i=1}^N 
			\prod_{k=1}^2
			 \bwdist{\iit k}
		} \ind_{\iit 1 \neq \iit 2} \text{ by Lemma~\ref{lem:two_backward_traj}}  \\
        & \geq \ps{1 - \frac 1N \hlbr} \ind_{A_t=1} \text{ by Assumptions~\ref{asp:mt_2ways_bound} and~\ref{asp:g_2ways_bound}.} \label{ieq_e_bw_reg}
	\end{aligned}
	\end{equation}
	Combining \eqref{ieq_d_4trajs_bw_reg_bound} and \eqref{ieq_e_bw_reg} yields \eqref{ieq_a_4trajs_backward_reg_bound} for $N$ large enough. To prove \eqref{ieq_b_4trajs_backward_reg_bound}, we write
	\begin{equation}\label{ieq_f_bw_reg_bound}
	\begin{aligned}
		\MoveEqLeft\CProb{A_{t-1}B_{t-1}=1}{\rtftm} \ind_{A_t=1}  \\
		&=\ps{1 - \CProb{A_{t-1}B_{t-1}=0}{\rtftm}} \ind_{A_t=1}  \\
        & \geq \ps{1 - \prtf{A_{t-1}=0} - \prtf{B_{t-1}=0}} \ind_{A_t=1}  \\
		&=\ps{\prtf{A_{t-1}=1} + \prtf{B_{t-1}=1} - 1} \ind_{A_t=1}. 
	\end{aligned}
	\end{equation}
	We analyse the second term in the above expression. We have
	\begin{equation}\label{ieq_g_bw_reg_bound}
	\begin{aligned}
		\MoveEqLeft\prtf{B_{t-1}=1} \ind_{A_t=1} \\
        & = \prtf{\iitmo 2 = \iitmos 2} \ind_{\iit 1 \neq \iit 2}  \\
		& = \ps{1 - \operatorname{TV}\pr{\btn{FFBS}(\iit 2, \cdot), \btn{FFBS}(\iits 2, \cdot)}} \ind_{A_t=1}  \\
		&\text{ by construction of Algorithm~\ref{algo:four_trajs}}  \\
        & \geq (\mblow/\mbhigh)\ind_{A_t=1} \text{ by Lemma~\ref{lem:backward_mixing}.} 
	\end{aligned}
	\end{equation}
	Plugging \eqref{ieq_e_bw_reg} and \eqref{ieq_g_bw_reg_bound} into \eqref{ieq_f_bw_reg_bound} yields
	\[\prtf{A_{t-1}B_{t-1}=1}\ind_{A_t=1} \geq \pr{-\frac 1N \hlbr + \frac{\mblow}{\mbhigh}} \ind_{A_t=1} \]
	and thus \eqref{ieq_b_4trajs_backward_reg_bound} follows if $N$ is large enough. The inequality \eqref{ieq_c_4trajs_backward_reg_bound} is justified by combining \eqref{ieq_e_bw_reg}, the simple decomposition $\ind_{\atbt=1} = \ind_{A_t=1} \ind_{B_t=1}$ and the fact that Algorithm~\ref{algo:four_trajs} guarantees $B_{t-1} = 1$ if $A_t = B_t = 1$.
	
	We can now deduce the two inequalities in the statement of the Lemma. The first one is a straightforward double application of \eqref{ieq_c_4trajs_backward_reg_bound}:
	\begin{align*}
		\MoveEqLeft\prtf{A_{t-2}B_{t-2}=1} \ind_{\atbt=1} \\
		&\geq \prtf{A_{t-2}B_{t-2}=1, A_{t-1}B_{t-1}=1} \ind_{\atbt=1} \\
		&=\CE{\CProb{A_{t-2}B_{t-2}=1, A_{t-1}B_{t-1}=1}{R_{t-1}, \rtftm}}{\rtftm} \ind_{\atbt=1}\\
		&\text{ by the law of total expectation}\\
		&=\CE{\CProb{A_{t-2}B_{t-2}=1}{R_{t-1}, \FTM}\ind_{A_{t-1}B_{t-1}=1}}{\rtftm} \ind_{\atbt=1} \\
		&\text{ since } (R_{T-\ell})_{\ell=0}^T \text{ is Markov given } \FTM \\
		&\geq \CE{\pr{1 - \frac 1N \hlbr} \ind_{A_{t-1}B_{t-1}=1}}{\rtftm} \ind_{\atbt=1} \\
		&\geq \ps{1 - \frac 1N \hlbr}^2 \ind_{\atbt=1} \geq \pr{1 - \frac 2N \hlbr} \ind_{\atbt=1}.
	\end{align*}
	Finally, we have
	\begin{align*}
		\MoveEqLeft\prtf{A_{t-2}B_{t-2} = 1} \\
        & \geq \prtf{A_{t-2}B_{t-2} = 1, A_{t-1} = 1}\\
		&=\CE{\CProb{A_{t-2}B_{t-2}=1}{R_{t-1}, \FTM} \ind_{A_{t-1}=1}}{\rtftm}\\ 
		&\mbox{ using law of total expectation and the Markov property as above}\\
		&\geq \frac{\mblow}{2\mbhigh} \prtf{A_{t-1}=1} \text{ by \eqref{ieq_b_4trajs_backward_reg_bound}}\\
		&\geq \frac{\mblow}{2\mbhigh} \epss \text{ by \eqref{ieq_a_4trajs_backward_reg_bound}}
	\end{align*}
	and the second inequality is proved.
\end{proof}

\subsection{\propp~\ref{prop:hybrid_validity} (hybrid rejection validity)}
\label{apx:proof:hybrid_validity}

\begin{proof}
Put $Z_n:=(X_n, U_nC\mu_0(X_n))$. Then $Z_n$ is uniformly distributed on
\[ \mathcal G_0:= \px{(x, y) \in \mathcal X \times \mathbb R_+, y \leq C \mu_0(x)}.\]
The proof would be done if one could show that, \textit{given} $K^* \leq K$, the variable $Z_{K^*}$ is uniformly distributed on
\[\mathcal G_1:= \px{(x, y) \in \mathcal X \times \mathbb R_+, y\leq \mu_1(x)}. \]
Note that $K^*$ is, by definition, the first time index where the sequence
$(Z_n)$ touches $\mathcal G_1$. Let $B$ be any subset of $\mathcal G_1$. We have
\begin{equation}\label{eq:main:proof:hybrid_validity}
	\begin{aligned}
	\MoveEqLeft\CProb{Z_{K^*} \in B}{K^* \leq K} \propto \P(Z_{K^*} \in B, K^* \leq K)\\
    & = \sum_{k^*=1}^\infty \P\pr{Z_{k^*} \in B, K^*=k^*, K \geq k^*} \\
    & = \sum_{k^*=1}^\infty \P\pr{Z_{k^*} \in B, Z_{1:k^*-1} \notin \mathcal G_1, K > k^*-1} \\
    & = \sum_{k^*=1}^{\infty} \P(Z_{k^*} \in B) \P\pr{Z_{1:k^*-1} \notin
    \mathcal G_1, K > k^*-1} \mbox{since $K$ stopping time}\\
    & = \P(Z_1 \in B)\sum_{k^*=1}^\infty \P\pr{Z_{1:k^*-1} \notin \mathcal G_1, K > k^*-1} \\
    & \propto \P(Z_1 \in B) \propto \CProb{Z_1 \in B}{Z_1 \in \mathcal G_1}.
	\end{aligned}
\end{equation}
By considering the special case $B = \mathcal G_1$, we see that the constant of proportionality between the first and the last terms of \eqref{eq:main:proof:hybrid_validity} must be $1$, from which the proof follows.
\end{proof}

\subsection{\prot~\ref{thm:intermediate_perf} (hybrid algorithm's intermediate complexity)}
From \eqref{eq:dist_tau_N}, one may have the correct intuition that as $N \to \infty$, $\tau_t^{1, \mrp}$ tends in distribution to that of the variable $\tau_t^{\infty, \mrp}$ defined as
\begin{equation}
\label{eq:dist_tau_infty}
\tau_t^{\infty, \mrp} \textrm{ } | \textrm{ } X_t^\pinfty \sim \operatorname{Geo}\pr{\frac{r_t(X_t^\pinfty)}{\mbhigh}}
\end{equation}
where $X_t^\pinfty \sim \Q_{t-1}M_t (\dd x_t)$ is distributed according to the predictive distribution of $X_t$ given $Y_{0:t-1}$ and $r_t$ is the density of $X_t^\pinfty$ with respect to the Lebesgue measure (cf. Definition~\ref{def:rt}).
The following proposition formalises the connection between $\tau_t^{1,\mrp}$ and $\tau_t^\pinfty$.
\begin{prop}
	\label{prop: tau_N_to_tau_inf}
	We have $\tau_t^{1, \mrp} \Rightarrow \tau_t^\pinfty$ as $N \to \infty$.
\end{prop}
\begin{proof}
	From \eqref{eq:dist_tau_N} and Definition~\ref{def:rt} one has
	\begin{equation}
	\label{eq:dist_tauN_with_r}
	\tau_t^{1, \mrp} \given X_t^1, \mathcal{F}_{t-1} \sim \operatorname{Geo}\pr{\frac{r_t^N(X_t^1)}{\mbhigh}}.
	\end{equation}
	In light of \eqref{eq:dist_tau_infty}, it suffices to establish that
	\begin{equation}
	\label{eq:suff_prop_dist_conv}
	\frac{r_t^N(X_t^1)}{\mbhigh} \Rightarrow \frac{r_t(X_t^\pinfty)}{\mbhigh}.
	\end{equation}
	Indeed, this would mean that for any continuous bounded function $\psi$, we have
	\begin{align*}
	\E[\psi(\tau_t^{1,\mrp})] = \E\ps{(\operatorname{Geo}^\star \psi)\pr{\frac{r_t^N (X_t^1)}{\mbhigh}}} 
	&\rightarrow \E\ps{(\operatorname{Geo}^\star \psi)\pr{\frac{r_t (X_t^\pinfty)}{\mbhigh}}} \\
	&= \E[\psi(\tau_t^\pinfty)]
	\end{align*}
	where $\operatorname{Geo}^\star$ is the geometric Markov kernel that sends each $\lambda$ to the geometric distribution of parameter $\lambda$, i.e.\ $\operatorname{Geo}^\star(\lambda, \dx) = \operatorname{Geo}(\lambda)$. To this end, write
	\begin{equation}
	\label{eq:law_cvg_rN_to_r}
	\begin{split}
	r_t^N(X_t^1) - r_t(X_t^1) 
    &= \frac{\sum_n G_{t-1}(X_{t-1}^n) m_t(X_{t-1}^n, X_t^1)}{\sum_n G_{t-1}(X_{t-1}^n)} - r_t(X_t^1) \\
    &= \frac{\sum_n N^{-1} G_{t-1}(X_{t-1}^n) \left[m_t(X_{t-1}^n, X_t^1) -
        r_t(X_t^1)\right]}{N^{-1} \sum_n G_{t-1}(X_{t-1}^n)}.
	\end{split}
	\end{equation}
	We study the mean squared error of the numerator:
	\begin{align*}
	\MoveEqLeft \E\px{\frac 1 N\sum_n  G_{t-1}(X_{t-1}^n) \ps{m_t(X_{t-1}^n, X_t^1) - r_t(X_t^1)}}^2 \\
    & =  \frac{1}{N} \E\px{ G_{t-1}(X_{t-1}^1)^2 \ps{m_t(X_{t-1}^1, X_t^1) - r_t(X_t^1)}^2} \\
    & \quad + \frac{N(N-1)}{N^2} \E\Big\{G_{t-1}(X_{t-1}^1) G_{t-1}(X_{t-1}^2) \ps{m_t(X_{t-1}^1, X_t^1) - r_t(X_t^1)} \\
     & \qquad \times \ps{m_t(X_{t-1}^2, X_t^1) - r_t(X_t^1)} \Big\}
	\end{align*} 
	where we have again used the exchangeability induced by step $(\star)$ of Algorithm~\ref{algo:paris_concrete}. The first term obviously tends to $0$ as $N \to \infty$ by Assumptions~\ref{asp:Ct} and~\ref{asp:Gbound}. The second term also vanishes asymptotically thanks to Lemma~\ref{lem:prop_chao} below and Assumption~\ref{asp:continuous}. Assumption~\ref{asp:Gbound} also implies that the denominator of \eqref{eq:law_cvg_rN_to_r} converges in probability to some constant, via the consistency of particle approximations, see e.g.\ \citet{DelMoral:book} or \citet{SMCbook}. Thus, $r_t^N(X_t^1) - r_t(X_t^1) \Rightarrow 0$ by Slutsky's theorem. Moreover, $r_t(X_t^1) \Rightarrow r_t(X_t^\pinfty)$ by the continuity of $r_t$ and the consistency of particle approximations. Using again Slutsky's theorem yields \eqref{eq:suff_prop_dist_conv}. 
\end{proof}
The following lemma is needed to complete the proof of Proposition~\ref{prop: tau_N_to_tau_inf} and is related to the propagation of chaos property, see \citet[Chapter 8]{DelMoral:book}.
\begin{lem}
	\label{lem:prop_chao}
	We have $(X_{t-1}^1, X_{t-1}^2, X_t^1) \Rightarrow \Q_{t-2}M_{t-1} \otimes \Q_{t-2}M_{t-1} \otimes Q_{t-1}M_t$.
\end{lem}

\begin{proof}
	For vectors $u$, $v$, and $w$, we have, by the symmetry of the distribution of particles:
	\begin{align*}
	\MoveEqLeft\E\ps{\exp\pr{iuX_{t-1}^1 + ivX_{t-1}^2 + iwX_t^1}} \\
    & = \E\ps{\pr{\Teiu} \pr{\Teiv} \pr{\Teiw}} \\
    & \quad - \frac{N}{N^2} \E\ps{e^{iuX_{t-1}^1} e^{ivX_{t-1}^1} \pr{\Teiw}} \\
    & \quad + \frac{N}{N^2} \E\ps{e^{iuX_{t-1}^1} e^{ivX_{t-1}^2} \pr{\Teiw}}.
	\end{align*}
	Note that
	\[\Teiu \overset{\textrm{a.s.}}{\longrightarrow} \Q_{t-2}M_{t-1}\pr{\exp(iu \bullet)} \]
	and
	\[\Teiw \overset{\textrm{a.s.}}{\longrightarrow} \Q_{t-1}M_t \pr{\exp(iw\bullet)}.\]
	The dominated convergence theorem, applicable since $\abs{e^{iu}} \leq 1$ for $u \in \mathbb{R}$, finishes the proof.
\end{proof}
\begin{proof}[Proof of Theorem~\ref{thm:intermediate_perf}]
	First of all,
	\begin{equation}
	\label{eq:proof_thm2_expectation_tauinf}
	\E[\tau_t^\pinfty] = \E\ps{\frac{\mbhigh}{r_t(X_t^\pinfty)}} = \int_{\mathcal{X}_t} \frac{\mbhigh}{r_t(x_t)} r_t(x_t) \dd x_t = \infty    
	\end{equation}
	by Assumption~\ref{asp:space}. Next, for any $x \in \mathbb{R} \setminus \mathbb{Z}$ and $N > x$,
	\[\P\pr{\tnn \leq x} = \P\pr{\tau_t^{1,\mrp} \leq x} \to \P\pr{\tau_t^\pinfty \leq x}\]
	by Proposition~\ref{prop: tau_N_to_tau_inf}. Thus, by Portmanteau theorem,
	\begin{equation}
	\label{eq:proof_thm2_tnn_to_tinf}
	\tnn \Rightarrow \tau_t^\pinfty.
	\end{equation}
	Altogether, we have
	\begin{align*}
	\liminf_{N \to \infty} \E\ps{\tnn} &= \liminf_{N\to\infty} \sum k \P\pr{\tnn = k} \\
	&\geq \sum \liminf_{N\to\infty} k \P\pr{\tnn = k} \textrm{ by Fatou's lemma} \\
	&= \sum k \P\pr{\tau_t^\pinfty = k} \textrm{ by \eqref{eq:proof_thm2_tnn_to_tinf}} \\
	&= \infty {\textrm{ by \eqref{eq:proof_thm2_expectation_tauinf}}}
	\end{align*}
	and
	\[\lim_{N \to \infty} \frac 1 N \E\ps{\tnn} = \lim_{N\to\infty} \E\ps{\min\pr{\frac{\tau_t^{1,\mrp}}{N}, 1}} \to 0 \]
	since $\tau_t^{1,\mrp} \Rightarrow \tau_t^\pinfty$ implies that the sequence of random variables \[\min\pr{\frac{\tau_t^{1,\mrp}}{N}, 1}\] converges to $0$ in distribution while being bounded between $0$ and $1$.
\end{proof}

\subsection{\prot~\ref{thm:near_linear} (hybrid PaRIS near-linear complexity)}
\label{proof:near_linear}
The following proposition shows that the real execution time for the hybrid algorithm is asymptotically at most of the same order as the ``oracle'' hybrid execution time.
\begin{prop}
	\label{prop:real_tnn_vs_oracle_tnn}
	We have \[\limsup_{N \to \infty} \frac{\E\ps{\tnn}}{\E\ps{\min(\tau_t^\pinfty, N)}} < \infty. \]
\end{prop}
\begin{proof}
	Put
	\begin{equation}
	\label{eq:def_zn}
	 z^N(\lambda) \eqdef \frac{1 - (1-\lambda)^N}{\lambda} = \sum_{n=0}^{N-1}
     (1-\lambda)^n.
	 \end{equation}
	One can quickly verify (using the memorylessness of the geometric distribution for example) that $z^N(\lambda) = \CE{\min(G,N)}{G \sim \operatorname{Geo}(\lambda)}$. It will be useful to keep in mind the elementary estimate $z^N(\lambda) \leq \min(N, \lambda^{-1})$. We can now write
	\begin{align*}
	\E\ps{\tnn} &= \E\ps{z^N\pr{\frac{r_t^N(X_t^1)}{\mbhigh}}} \textrm{ (by \eqref{eq:dist_tauN_with_r})} 
	= \E\ps{\CE{z^N\pr{\frac{r_t^N(X_t^1)}{\mbhigh}}}{\mathcal{F}_{t-1}}} \\
	&= \E\ps{\intxt{\znp r_t^N(x_t) \lambda_t(\dd x_t)}} \\
	&\leq c_t\pr{\intxt{\znpt r_t(x_t) \lambda_t(\dd x_t)} + b_t} \text{ by Lemma~\ref{lem:magic_jensen}}\\
	&= c_t\pr{\E[\min(\tau_t^{\pinfty}, N)] + b_t}
	\end{align*}
	from which the proposition is immediate.
	\end{proof}
\begin{lem}
	\label{lem:magic_jensen}
	In addition to notations of Algorithm~\ref{algo:bootstrap}, let the function $z^N$ be defined as in \eqref{eq:def_zn} and the functions $r_t$ and $r_t^N$ be defined as in Definition~\ref{def:rt}. Let $\phi_t: \mathcal X_t \to \mathbb R_{>0}$ be a bounded non-negative deterministic function. Then, under Assumptions~\ref{asp:Gbound} and~\ref{asp:Ct}, there exist constants $b_t$ and $c_t$ depending only on the model such that
	\[\E\ps{\intxt{\znp r_t^N \phi_t}} \leq c_t\pr{\intxt{\znpt r_t \phi_t} + b_t \infnorm{\phi_t}}\]
	where for brevity, we shortened the integration notation (e.g.\ dropping $\lambda_t(\dd x_t)$, dropping $x_t$ from $\phi(x_t)$, etc.) whenever there is no ambiguity.
\end{lem}

\begin{proof} 
    We have
	\begin{equation}
	\label{eq:magic_jensen}
	 \E\ps{\intxt{ z^N\pr{\frac{r_t^N(x_t)}{\mbhigh}} r_t^N \phi_t}}
	\leq \intxt{ z^N\pr{\E\ps{\frac{r_t^N(x_t)}{\mbhigh}}} \E\ps{r_t^N(x_t)} \phi_t}
	\end{equation}
    using Fubini's theorem and the concavity of $\lambda \mapsto \lambda
    z^N(\lambda)$ on $[0,1]$. By a well-known result on the bias of a particle
    filter (which is in fact the propagation of chaos in the special case of
    $q=1$ particle), we have:
	\begin{align*}
	\abs{\E\ps{r_t^N(x_t)} - r_t(x_t)} &= \abs{\E\ps{\sum W_{t-1}^n m_t(X_{t-1}^n, x_t)} - r_t(x_t)} \\
	&= \abs{\E\ps{m_t\pr{X_{t-1}^{A_t^1}, x_t}} - \Q_{t-1}\pr{m_t\pr{\bullet, x_t}}} \\
	&\leq \frac{b_t \mbhigh}{N}
	\end{align*}
    for some constant $b_t$. We next show that such a bias does not change the
    asymptotic behavior of $z^N$. More precisely,
	\begin{equation}\label{eq:bias_zn_est1}
        \begin{aligned}
	z^N\pr{\E\ps{\frac{r_t^N(x_t)}{\mbhigh}}} 
&\leq z^N\pr{\frac{r_t(x_t)}{\mbhigh} - \frac{b_t}{N}} \\
    & = \sum_{n=0}^{N-1} \pr{\frac{1 - {r_t(x_t)}/{\mbhigh} + \nfrac{b_t}{N}}{1 - \nfrac{r_t(x_t)}{\mbhigh}}}^n \pr{1 - \frac{r_t(x_t)}{\mbhigh}}^n \\
	&\leq \sum_{n=0}^{N-1} \pr{1 + \frac{b_t}{N\pr{1 - r_t(x_t)/\mbhigh}}}^N \pr{1 - \frac{r_t(x_t)}{\mbhigh}}^n \\
    &\leq \exp\pr{\frac{b_t}{1-r_t(x_t)/\mbhigh}}
    z^N\pr{\frac{r_t(x_t)}{\mbhigh}} \\
    & \leq e^{2b_t} z^N\pr{\frac{r_t(x_t)}{\mbhigh}}
	\end{aligned}
\end{equation}
    if $x_t$ is such that $r_t(x_t)/\mbhigh \leq 1 / 2$. In contrast,
    if $r_t(x_t)/\mbhigh \geq 1 / 2$, then provided that $N \geq 6b_t$,
    we have
	\begin{equation}
	\label{eq:bias_zn_est2}
	z^N\pr{\E\ps{\frac{r_t^N(x_t)}{\mbhigh}}} \leq z^N\pr{\frac{r_t(x_t)}{\mbhigh} - \frac{b_t}{N}} \leq z^N\pr{\frac 13} \leq 3 z^N\pr{\frac{r_t(x_t)}{\mbhigh}}.
	\end{equation}
	Putting together \eqref{eq:bias_zn_est1} and \eqref{eq:bias_zn_est2}, we have, for $N \geq 6b_t$,
	\[z^N\pr{\E\ps{\frac{r_t^N(x_t)}{\mbhigh}}} \leq \pr{e^{2b_t} + 3} z^N\pr{\frac{r_t(x_t)}{\mbhigh}}\]
	and so, by \eqref{eq:magic_jensen},
	\begin{align*}
	\E\ps{\intxt{\znp r_t^N \phi_t}} &\leq \pr{e^{2b_t} + 3} \intxt{ z^N\pr{\frac{r_t(x_t)}{\mbhigh}} \E\ps{r_t^N(x_t)} \phi_t}  \\
	&= \pr{e^{2b_t} + 3} \E\ps{z^N\pr{\frac{r_t(X_t^1)}{\mbhigh}} \phi_t(X_t^1)}.
	\end{align*}
	Again, using the result on the bias of a particle filter,
	\[\abs{\E\ps{z^N\pr{\frac{r_t(X_t^1)}{\mbhigh}} \phi_t(X_t^1)} - \intxt{\znpt r_t \phi_t}} \leq \frac{b_t \norminf{z^N} \norminf{\phi_t}}{N} = b_t \norminf{\phi_t} \]
	which, together with the previous inequality, implies the desired result.
\end{proof}
\begin{prop}
	\label{prop:oracle_linear}
	In linear Gaussian state space models, we have \[\E\ps{\min(\tau_t^\pinfty, N)} = \mathcal{O}\pr{(\log N)^{d_t/2}}.\]
\end{prop}
\begin{proof}
	Let $\mu_t$ and $\Sigma_t$ be such that $X_t^\pinfty \sim \mathcal{N}(\mu_t, \Sigma_t)$. Then \[\log(r_t(X_t^\pinfty)/\mbhigh) = b'_t - W_t\] where $b'_t$ is some constant and
	\[W_t \eqdef \frac{(X_t^\pinfty - \mu_t)^\top \Sigma_t^{-1} (X_t^\pinfty - \mu_t)}{2} \sim \operatorname{Gamma}\pr{\frac {d_t}{2}, 1}. \]
	We have
	\begin{align*}
	\E\ps{\min(\tau_t^\pinfty, N)} &= \E\ps{z^N\pr{\frac{r_t(X_t^\pinfty)}{\mbhigh}}} = \E\ps{z^N(e^{b'_t - W_t})} \\
	&= \int_0^\infty z^N(e^{b'_t-w}) \frac{w^{d_t/2 - 1}e^{-w}}{\Gamma(d_t/2)} \dd w \\
	&\leq \int_0^{\log N} e^{w-b'_t} \frac{w^{d_t/2 - 1}e^{-w}}{\Gamma(d_t/2)} \dd w + \int_{\log N}^\infty N \frac{w^{d_t/2 - 1}e^{-w}}{\Gamma(d_t/2)} \dd w
	\end{align*}
	using the bound $z^N(\lambda) \leq \min(N, 1/\lambda)$. The first term is of order $\mathcal{O}(\log^{d_t/2}N)$ by elementary calculus, while the second term is of order $\mathcal{O}(\log^{d_t/2 - 1} N)$ using asymptotic properties of the incomplete Gamma function, see \citet[Section 8.11]{nisthandbook}.
\end{proof}
\begin{proof}[Proof of Theorem~\ref{thm:near_linear}]
	The theorem is a straightforward consequence of Proposition~\ref{prop:real_tnn_vs_oracle_tnn} and Proposition~\ref{prop:oracle_linear}.
\end{proof}

\subsection{\prots~\ref{thm:ffbs_exec_infinite_general} and~\ref{thm:ffbs_exec_infinite_gaussian} (pure rejection FFBS complexity)}
\label{ap:proof:ffbs_exec_infinite}
We start with a useful remark linking the projection kernels $\Pi$ and the cost-to-go functions defined in Supplement~\ref{sec:apx_notations} with the L-kernels formulated in \eqref{eq:def:l_kernels}. The proof is simple and therefore omitted.
\begin{lem}
	\label{lem:l_and_proj}
	We have $L_{t:T}(x_{0:t}, \mathbbm 1) = H_{t:T}(x_t)$ for all $x_{0:t}$. Moreover, for any function $\phi_t: \mathcal X_t \to \mathbb R$, we have
	\[L_{t:T}\Proj^{0:T}_t \phi_t = \Proj^{0:t}_t(\phi_t \times H_{t:T}). \]
\end{lem}
Theorems~\ref{thm:ffbs_exec_infinite_general} and \ref{thm:ffbs_exec_infinite_gaussian} both rely on an induction argument wrapped up in the following proposition.
\begin{prop}
	\label{prop:ffbs_exec_induction}
	We use the notations of Algorithm~\ref{algo:offline_generic}. Let $\Q_t^N$ be defined as in \eqref{eq:def_joint_qt}, where the $B_s^N$ kernels can be $B_s^{N, \mathrm{FFBS}}$ or any other kernels satisfying the hypotheses of Theorem~\ref{thm:convergence_mcmc}. Suppose that Assumption~\ref{asp:Gbound} holds. Let $f_t^N: \mathcal X_t \to \mathbb R_{\geq 0}$ be a (possibly random) function such that $f_t^N(x_t)$ is $\mathcal F_{t-1}$-measurable. Then the following assertions are true:
	\begin{enumerate}
		\item[(a)] Suppose that $\E\ps{\int_{\mathcal X_t} \px{r_t^N \times f_t^N \times G_t \times H_{t:T}} (x_t) \lambda_t(\dd x_t)} = \infty$, where $r_t^N$ and $\lambda_t$ are defined in Definition~\ref{def:rt}. Then
		\[\E\ps{\int \Q_T^N(\dd x_t) f_t^N(x_t)} = \infty. \]
		\item[(b)] Suppose that $\int_{\mathcal X_t} \px{r_t^N \times f_t^N \times G_t \times H_{t:T}} (x_t) \lambda_t(\dd x_t) \cvprob 0$. Then
		\[\int \Q_T^N(\dd x_t) f_t^N(x_t) \cvprob 0. \]
	\end{enumerate}
\end{prop}
\begin{proof}
	\textbf{Part (a).} We shall prove by induction the statement
	\[\E\ps{\Q_s^N L_{s:T} \Proj^{0:T}_t f_t^N} = \infty, \forall \ t-1 \leq s \leq T. \]
	For $s=t-1$, it follows from part (a)'s hypothesis and Lemma~\ref{lem:l_and_proj}. Indeed,
	\begin{align*}
		\MoveEqLeft\Q_{t-1}^N L_{t-1:T} \Proj_t^{0:T} f_t^N \\
		&= \Q_{t-1}^N L_{t-1:t} L_{t:T} \Proj_t^{0:T} f_t^N = \Q_{t-1}^N L_{t-1:t} \Proj_t^{0:t} (f_t^N \times H_{t:T}) \\
		&= \iint_{\mathcal X_{t-1} \times \mathcal X_t} \Q_{t-1}^N(\dd x_{t-1}) m_t(x_{t-1}, x_t) \lambda_t(\dd x_t) G_t(x_t) (f_t^N \times H_{t:T})(x_t) \\
		&= \int_{\mathcal X_t} \px{r_t^N \times f_t^N \times G_t \times H_{t:T}} (x_t) \lambda_t(\dd x_t).
	\end{align*}
	For $s \geq t$, we have
	\begin{align*}
		\E\ps{\Q_s^N L_{s:T} \Proj^{0:T}_t f_t^N} &= \E\ps{\frac{\inv N \sum \tilde K_s^N(n, L_{s:T}\Proj^{0:T}_t f_t^N)}{\ell_s^N}} \text{ by Corollary~\ref{corol:fundamental}} \\
		&\geq \frac{1}{\infnorm{G_s}} \E\ps{\inv N \sum \tilde K_s^N(n, L_{s:T}\Proj^{0:T}_t f_t^N)} \\
		&\text{by Assumption~\ref{asp:Gbound} and definition of } \ell_s^N \text{ (see Algorithm~\ref{algo:bootstrap})} \\
		&\geq \frac{1}{\infnorm{G_s}} \E\ps{\Q_{s-1}^N L_{s-1:s} L_{s:T} \Proj^{0:T}_t f_t^N} \\
		&\text{by Corollary~\ref{corol:fundamental} and law of total expectation}\\
		&= \E\ps{\Q_{s-1}^N L_{{s-1}:T} \Proj^{0:T}_t f_t^N} = \infty \text{ (induction hypothesis).}
	\end{align*}
	\textbf{Part (b).} Similar to part (a), we shall prove by induction the statement
	\[\Q_s^N L_{s:T}  \Proj^{0:T}_t f_t^N \cvprob 0, \forall \ t-1 \leq s \leq T.\]
	Again, by Corollary~\ref{corol:fundamental}, this quantity is equal to
	\[\frac{\inv N \sum \tilde K_s^N(n, L_{s:T}\Proj^{0:T}_t f_t^N)}{\ell_s^N}, \]
	and the expectation of the numerator given $\mathcal F_{s-1}$ is
    $\Q_{s-1}^N L_{s-1:T} \Proj_t^{0:T} f_t^N$, which tends to $0$ in
    probability by induction hypothesis.
    Lemma~\ref{lem:cvg_conde_to_convg_prob} (see below at the end of the
    section), the classical result $\ell_s^N \cvprob \ell_s:=\Q_{s-1}M_s(G_s)$ and Stutsky's theorem concludes the proof. 
\end{proof}
\begin{proof}[Proof of Theorem~\ref{thm:ffbs_exec_infinite_general}]
	By \eqref{eq:cond_dist_tau_ffbs}, we have $\E[\tau_t^{1, \mrffbs}] = \E[\int \Q_T^{N, \mrffbs}(\dd x_t) f_t^N(x_t)] $ where
	\[f_t^N(x_t) = \frac{\mbhigh}{\sum W_{t-1}^n m_t(X_{t-1}^n, x_t)} = \frac{\mbhigh}{r_t^N} \]
	with $r_t^N$ given in Definition~\ref{def:rt}.
    Proposition~\ref{prop:ffbs_exec_induction}(a) gives a sufficient condition
    for $\E[\tau_t^{1, \mrffbs}]=\infty$ to hold, namely
	\[\int_{\mathcal{X}_t} (r_t^N \times f_t^N \times G_t \times H_{t:T})(x_t)\lambda_t(\dd x_t) = \infty, \]
	which is equivalent to the hypothesis of the theorem.
\end{proof}
\begin{proof}[Proof of Theorem~\ref{thm:ffbs_exec_infinite_gaussian}]
	We use notations from Definition~\ref{def:rt} and Supplement~\ref{apx:linear_gaussian_models}. We note $\mathcal N(x | \mu, \Sigma)$ the density of the specified normal distribution at point $x$. Using Lemma~\ref{lem:higher_moment_geometric}, Proposition~\ref{prop:ffbs_exec_induction} and \eqref{eq:cond_dist_tau_ffbs}, we have
	\begin{align*}
		\E[(\tau_t^{1, \mrffbs})] = \infty &\Leftrightarrow \E\ps{\frac{1}{r_t^N(X_t^{\mathcal I_t^1})^k}} = \infty \\
		&\Leftrightarrow \E\ps{\int \Q_T^N(\dd x_t) \frac{1}{r_t^N(x_t)^k}} = \infty \\
		&\Leftarrow \E\ps{\int_{\mathcal X_t} \frac{r_t^N G_t H_{t:T}}{(r_t^N)^k}(x_t) \lambda_t(\dd x_t)} = \infty \\
		&\Leftarrow \int_{\mathcal X_t} \frac{r_t G_t H_{t:T}}{(r_t^N)^{k-1} r_t}(x_t) \lambda_t(\dd x_t) = \infty \text{ almost surely} \\
		&\Leftrightarrow \int_{\mathcal X_t} \frac{\mathcal N(x_t | \mu_t^{\mathrm{smth}}, \sigmasmth_t)}{r_t^N(x_t)^{k-1} \mathcal N(x_t| \mu_t^{\mathrm{pred}}, \sigmapred_t)} \lambda_t(\dd x_t) = \infty \text{ a.s. }
	\end{align*}
	The theorem then follows from elementary arguments, by noting that $r_t^N$ is a mixture of $N$ Gaussian distributions with covariance matrix $C_X$.
\end{proof}
\begin{lem}
	\label{lem:higher_moment_geometric}
	Let $L$ be a $]0,1]$-valued random variable. Suppose $X$ is another random variable such that $X | L \sim \operatorname{Geo}(L)$. Then for any real number $k>0$,
	\[\E[X^k] = \infty \Leftrightarrow \E[L^{-k}]=\infty. \]
\end{lem}
\begin{proof}
	By the definition of $X$, we have
	\[\E[X^k] = \E\ps{\sum_{x=1}^\infty x^k(1-L)^{x-1}L}. \]
	A natural idea is then to approximate the sum by the integral $\int_0^\infty x^k (1-L)^{x-1} L \dd x$, from which one easily extracts the $L^{-k}$ factor. This is however technically laborious, since the function $x \mapsto x^k(1-L)^{x-1}L$ is not monotone on the whole real line. It is only so starting from a certain $x_0$ which itself depends on $L$. We would therefore rather write
	\begin{equation*}
		\begin{split}
		\E[X^k] &= \int_{0}^\infty \P(X^k \geq x) \dd x = \int_{0}^\infty \P(X \geq x^{1/k}) \dd x\\
		&= \int_0^\infty \E\ps{(1-L)^{\floor{x^{1/k}}}} \dd x \\
		&\text{where the two integrands are equal Lebesgue--almost-everywhere} \\
		&= \E\ps{\int_{0}^\infty \expfc{\floor} \dd x}
		\end{split}
	\end{equation*}
	with the natural interpretation of expressions when $L=1$. Using $u \sim v$ as a shorthand for ``$u$ and $v$ are either both finite or both infinite'', we have
	\begin{equation*}
		\begin{split}
		\E[X^k] &\sim  \E\ps{\int_0^\infty \expfc{} \dd x} \text{ by Lemma~\ref{lem:equivalent_in_01}} \\
		&= k\ \Gamma(k)\  \E\ps{\frac{1}{\abs{\log(1-L)}^k}} 
        \sim \E\left[\frac{1}{L^k}\right] \text{ by Lemma~\ref{lem:equivalent_in_01} again.}
		\end{split} 
	\end{equation*} 
% The proof is now completed.
\end{proof}

The following lemma is elementary. Its proof is therefore omitted.

\begin{lem}\label{lem:equivalent_in_01}
Let $L$ be a $]0,1]$-valued random variable and let $f_1$ and $f_2$ be two continuous functions from $]0,1]$ to $\mathbb R$. Suppose that $\limsup_{\ell \to 0^+} f_1(\ell)/f_2(\ell)$ and $\limsup_{\ell \to 0^+} f_2(\ell)/f_1(\ell)$ are both finite. Then $\E[f_1(L)]$ is finite if and only if $\E[f_2(L)]$ is so.
\end{lem}

\begin{lem}\label{lem:cvg_conde_to_convg_prob}
	Let $Z_1, Z_2, \ldots$ be non-negative random variables. Suppose that there exist $\sigma$-algebras $\mathcal F_1, \mathcal F_2, \ldots$ such that $\E[Z_n|\mathcal F_n] \cvprob 0$. Then $Z_n \cvprob 0$.
\end{lem}
\begin{proof}
	Fix $\varepsilon > 0$. By Markov's inequality, $\P(Z_n \geq \varepsilon |
    \mathcal F_n) \leq \inv \varepsilon \E[Z_n |\mathcal F_n]$. Therefore, the $[0,1]-$bounded random variable $\P(Z_n \geq \varepsilon | \mathcal F_n)$ tends to $0$ in probability, hence also in expectation. The law of total expectation then gives $\P(Z_n \geq \varepsilon) \to 0$, which, by varying $\varepsilon$, establishes the convergence of $Z_n$ to $0$ in probability.
\end{proof}

\subsection{\prot~\ref{thm:ffbs_hybrid_exec} and Corollary~\ref{cor:gaussian_hybrid_ffbs} (hybrid FFBS complexity)}
\label{ap:proof:collective_ffbs_hybrid}
\begin{proof}[Proof of Theorem~\ref{thm:ffbs_hybrid_exec}]
According to \citet[Lemma 3]{JansonBigO}, it is sufficient to show that
\[\frac{\sum_n \min(\tau_t^{n, \mrffbs}, N)}{N\alpha_N} \cvprob 0 \]
for any deterministic sequence $\alpha_N$ such that ${\alpha_N}/{\E[\min(\tauinfffbs, N)]} \to \infty$. By Lemma~\ref{lem:cvg_conde_to_convg_prob}, we can take expectation with respect to $\mathcal F_T$ to derive a sufficient condition, namely
\begin{equation*}
	\begin{split}
	&\hspace{0.5cm} \intxt \inv{\alpha_N} \znp \Q_T^N(\dd x_t) \cvprob 0 \text{ with } z^N \text{ defined in \eqref{eq:def_zn}} \\
	&\Leftarrow \intxt \inv{\alpha_N} \znp r_t^N \times G_t \times H_{t:T} \ \dd \lambda_t \cvprob 0 \text{ by Proposition~\ref{prop:ffbs_exec_induction}(b)}\\
	&\Leftarrow \E\ps{\intxt \inv{\alpha_N} \znp r_t^N \times G_t \times H_{t:T} \ \dd \lambda_t} \to 0 \\
	&\Leftarrow \intxt \inv{\alpha_N} \znpt r_t \times G_t \times H_{t:T} \ \dd \lambda_t \to 0 \text{ by Lemma~\ref{lem:magic_jensen}} \\
	&\Leftrightarrow \frac{\E[\min(\tauinfffbs,N)]}{\alpha_N} \to 0.
	\end{split}
\end{equation*}
The proof is now complete.
\end{proof}
\begin{proof}[Proof of Corollary~\ref{cor:gaussian_hybrid_ffbs}]
	We have, using the cost-to-go, the $z^N$ functions and the $\tau_t^{\infty, \mrp}$ distribution defined respectively in \eqref{eq:def-cost-to-go}, \eqref{eq:def_zn} and \eqref{eq:dist_tau_infty}:
	\begin{equation*}
		\begin{split}
		&\E[\min(\tauinfffbs, N)] = \intxt \znpt \Q_T(\dd x_t) \\
		= &\ps{(\Q_{t-1}M_t)(G_tH_{t:T})}^{-1} \intxt \znpt (G_tH_{t:T})(x_t)(\Q_{t-1}M_t)(\dd x_t)\\
		\leq& \infnorm{G_tH_{t:T}} \ps{(\Q_{t-1}M_t)(G_tH_{t:T})}^{-1} \intxt \znpt (\Q_{t-1}M_t)(\dd x_t)\\
		=& \infnorm{G_tH_{t:T}} \ps{(\Q_{t-1}M_t)(G_tH_{t:T})}^{-1} \E[\min(\tau_t^{\infty, \mrp}, N)].
		\end{split}
	\end{equation*}
	Proposition~\ref{prop:oracle_linear} then finishes the proof.
\end{proof}

\subsection{\propp~\ref{prop:mcmc_validity} (MCMC kernel properties)}
\label{apx:proof:mcmc_validity}
\newcommand{\tempxtmoa}{X_{t-1}^{1:N}}
\begin{proof}
	Let $J_t^n$ be such that $J_t^n|\tempxtmoa, X_t^n, \btnhat{IMH}(n, \cdot) \sim \btn{IMH}(n, \cdot)$. Moreover, let $K_t^n$ be such that
	$$K_t^n | \tempxtmoa, X_t^n, A_t^n, \btnhat{IMH}(n, \cdot) \sim \btn{IMH}(n, \cdot). $$
	Given $\tempxtmoa$, $X_t^n$ and $A_t^n$, the kernel $\btn{IMH}(n, \cdot)$ applies to $A_t^n$ one or more several MCMC steps keeping invariant $\btn{FFBS}(n, \cdot)$. Since $A_t^n | \tempxtmoa, X_t^n \sim \btn{FFBS}(n, \cdot)$, it follows that $K_t^n | \tempxtmoa, X_t^n \sim \btn{FFBS}(n, \cdot)$ too. On the other hand, $J_t^n$ and $K_t^n$ share the same distribution given $\tempxtmoa$, $X_t^n$ and $\btnhat{IMH}(n, \cdot)$. Hence they also do, given $\tempxtmoa$ and $X_t^n$ only. This implies that $J_t^n|\tempxtmoa, X_t^n \sim \btn{FFBS}(n, \cdot)$, which is the same as $A_t^n|\tempxtmoa, X_t^n$. Thus $(J_t^n, X_t^n)$ have the same distribution as $(A_t^n, X_t^n)$ given $\tempxtmoa$, as required by Theorem~\ref{thm:convergence_mcmc}. The arguments for the kernel $\btn{IMHP}$ are similar.
	
    To show that a certain kernel $B_t^N$ satisfies \eqref{eq:support_cond}, we
    look at two conditionally i.i.d. simulations $J_t^{n,1}$ and $J_t^{n,2}$
    from $B_t^N(n, \cdot)$ and lower bound the probability that they are
    different. For the kernel $\btn{IMH}$, the variables $J_t^{n,1}$ and
    $J_t^{n,2}$ both result from one step of MH applied to $A_t^n$. Let
    $J_t^{n,1*}$ and $J_t^{n,2*}$ be the corresponding MH proposals. A
    sufficient condition for $J_t^{n,1} \neq J_t^{n,2}$ is that $J_t^{n,1*}
    \neq J_t^{n,2*}$ and the two proposals are both accepted. The acceptance
    rate is at least $\mblow/\mbhigh$ by Assumption~\ref{asp:mt_2ways_bound}
    and the probability that $J_t^{n,1*} \neq J_t^{n,2*}$ is \[1 - \sum_{n=1}^N
    (W_{t-1}^n)^2 \geq 1 - \frac 1 N \pr{\frac{\gbhigh}{\gblow}}^2\] by
    Assumption~\ref{asp:g_2ways_bound}.  Thus \eqref{eq:support_cond} is
    satisfied for $\epss = \mblow/2\mbhigh$ for $N$
    large enough. Similarly, the probability that $J_t^{n,1} \neq J_t^{n,2}$ for
    the $\btn{IMHP}$ kernel with $\tilde N=2$ can be lower-bounded via the
    probability that $\tilde J_t^{n, 1} \neq \tilde J_t^{n,2}$ (where $\tilde
    J_t^{n,1}$ and $\tilde J_t^{n,2}$ are defined in \eqref{eq:mcmc_paris_kernel}).
    Thus using the same arguments,  \eqref{eq:support_cond} is satisfied here for
    $\epss = \mblow/4\mbhigh$.
\end{proof}

\subsection{Conditional probability of maximal couplings}
\label{conditional_proba_maximal_coupling}
In general, there exist multiple maximal couplings of two random distributions (i.e.\ couplings that maximise the probability of equality of the two variables). However, they all satisfy a certain conditional probability property stated in the following lemma. It is closely related to results on the coupling density on the diagonal (see e.g. \citealp[Lemma 2]{wang2021maximal} or \citealp[Theorem 19.1.6]{MR3889011}). Its statement, which we were unable to find in the literature in the exact form we need, is obvious in the discrete case but requires lengthier arguments in the continuous one.
\begin{prop}\label{conditional_proba_maximal_coupling_prop}
	Let $X_1$ and $X_2$ be two random variables with densities $f_1$ and $f_2$ with respect to some dominating measure defined on a space $\mathcal X$. Then, the following inequality holds almost surely:
	\begin{equation}
	\label{conditional_proba_maximal_coupling_ineq}
	\P(X_2 = X_1 | X_1) \leq 1 \land \frac{f_2(X_1)}{f_1(X_1)}.
	\end{equation}
	Moreover, the equality occurs almost surely if and only if $X_1$ and $X_2$ form a maximal coupling.
\end{prop}
\begin{proof}
	Let $h$ be any non-negative test function from $\mathcal X$ to $\mathbb R$. Putting
	\begin{align*}
		A_1 &\eqdef \px{x \in \mathcal X \mid f_1(x)\geq f_2(x)} \\
		A_2 &\eqdef \px{x \in \mathcal X \mid f_2(x) \geq f_1(x)}, 
	\end{align*}
	we have
	\begin{equation*}
	\begin{split}
	\E[\P(X_2=X_1|X_1) h(X_1)] &= \E[\mathbbm 1_{X_2 = X_1} h(X_1)] \\
	&= \E[\mathbbm 1_{X_2 = X_1} \mathbbm 1_{X_1 \in A_1} h(X_1)] + \E[\mathbbm{1}_{X_2 = X_1} \ind_{X_1 \in A_2} h(X_1)] \\
	&= \E[\mathbbm 1_{X_2 = X_1} \mathbbm 1_{X_2 \in A_1} h(X_2)] + \E[\mathbbm{1}_{X_2 = X_1} \ind_{X_1 \in A_2} h(X_1)] \\
	&\leq \E[\ind_{X_2 \in A_1} h(X_2)] + \E[\ind_{X_1 \in A_2} h(X_1)] \\
	&= \int h(x) f_2 \land f_1(x) \dd x = \E\ps{\pr{1 \land \frac{f_2(X_1)}{f_1(X_1)}} h(X_1)}.
	\end{split}
	\end{equation*}
	The inequality \eqref{conditional_proba_maximal_coupling_ineq} is now proved almost-surely. As a result, almost-sure equality occurs if and only if the expectation of the two sides of \eqref{conditional_proba_maximal_coupling_ineq} are equal, which means, via Lemma~\ref{lem:properties_TV}, that the two variables are maximally coupled.
\end{proof}
The following lemma establishes the symmetry of Assumption~\ref{asp:coupling:dynamics}. Again, its statement is obvious in the discrete case, though some work is needed to rigorously justify the continuous one.
\begin{lem}\label{lem:coupling_efficiency_symmetric}
	Let $X_1$ and $X_2$ be two random variables of densities $f_1$ and $f_2$
    w.r.t.\ some dominating measure defined on some space $\mathcal X$. Suppose
    that almost-surely
	\[\P(X_2 = X_1|X_1) \geq \varepsilon \pr{1 \land \frac{f_2(X_1)}{f_1(X_1)}} \]
	for some $\varepsilon > 0$. Then almost-surely,
	\[\P(X_1 = X_2|X_2) \geq \varepsilon \pr{1 \land \frac{f_1(X_2)}{f_2(X_2)}}. \]
\end{lem}
\begin{proof}
	We introduce a non-negative test function $h_2: \mathcal X \to \mathbb R$ and write
	\begin{equation*}
		\begin{split}
		\E[\P(X_1=X_2|X_2)h(X_2)] &= \E[\ind_{X_1=X_2} h(X_2)] = \E[\ind_{X_2=X_1} h(X_1)] \\
		&= \E[\P(X_2 = X_1|X_1)h(X_1)] \\
		&\geq \E\ps{\varepsilon\pr{1 \land \frac{f_2(X_1)}{f_1(X_1)}} h(X_1)} \\
		&= \int \varepsilon f_1 \land f_2(x) h(x) \dd x \\
		&= \E\ps{\varepsilon\pr{1 \land \frac{f_1(X_2)}{f_1(X_2)}} h(X_2)}
		\end{split}
	\end{equation*}
	which implies the desired result.
\end{proof}

\subsection{\prot~\ref{thm:intractable} (intractable kernel properties)}
\label{proof:thm:intractable}
\begin{proof}
	Let $J_t^n$ be a random variable such that 
    \[J_t^n | X_{t-1}^{1:N}, X_t^n, \hat B_t^{N, \mathrm{ITR}}(n, \cdot) \sim \btn{ITR}(n, \cdot).\]
	By construction of Algorithm~\ref{algo:intractable}, given $X_{t-1}^{1:N}$, the couple $(J_t^n, X_t^n)$ has the same distribution as $(A_t^{n,L}, X_t^{n,L})$. Thus, $\btn{ITR}$ satisfies the hypotheses of Theorem~\ref{thm:convergence_mcmc}. To verify \eqref{eq:support_cond}, we define the variables $J_t^{n, 1:2}$ accordingly and write:
	\begin{multline*}
		\CProb{J_t^{n,1} \neq J_t^{n,2}}{X_t^n=x_t, X_{t-1}^{1:N}=x_{t-1}^{1:N}} \\
		\begin{aligned}
		&= \frac 12 \CProb{X_t^{n,1} = X_t^{n,2}, A_t^{n,1} \neq A_t^{n,2}}{X_t^n=x_t, X_{t-1}^{1:N}=x_{t-1}^{1:N}} \\
		&= \CProb{X_t^{n,1}=X_t^{n,2}, A_t^{n,1} \neq A_t^{n,2}, L=1}{\defcond} \\
		&= \frac 12 \CProb{X_t^{n,1} = X_t^{n,2}, A_t^{n,1} \neq A_t^{n,2}}{L=1, \defcond} \text{(by symmetry)} \\
		&= \frac 12 \CProb{X_t^{n,1} = X_t^{n,2}, A_t^{n,1} \neq A_t^{n,2}}{\defcond[,1]} \\
		&= \frac 12 \begin{multlined}[t]
		 \sum_{a_t^{n,1} \neq a_t^{n,2}} \CProb{X_t^{n,2} = x_t}{A_t^{n,1} = a_t^{n,1}, A_t^{n,2} = a_t^{n,2}, \defcond[,1]} \times \\
		\times \CProb{A_t^{n,1} = a_t^{n,1}, A_t^{n,2} = a_t^{n,2}}{\defcond[,1]}
		\end{multlined} \\
		&\geq \frac 12
		\begin{multlined}[t]
		\epsd \frac{\mblow}{\mbhigh} \sum_{a_t^{n,1} \neq a_t^{n,2}} \CProb{A_t^{n,1} = a_t^{n,1}, A_t^{n,2} = a_t^{n,2}}{\defcond[,1]} \\
		\text{(by Assumptions~\ref{asp:coupling:dynamics} and~\ref{asp:mt_2ways_bound})}
		\end{multlined}\\
		&\geq \frac 12 \epsd \pr{\frac{\mblow}{\mbhigh}}^2 \epsa \text{ by Lemma~\ref{lem:modified_epsa}.}
		\end{aligned}
	\end{multline*}
	The proof is complete.
\end{proof}

\begin{lem}\label{lem:modified_epsa}
	We use the notations of Algorithm~\ref{algo:intractable}. Under Assumptions~\ref{asp:mt_2ways_bound} and~\ref{asp:coupling:ancestors}, we have
	\[\CProb{A_t^{n,1} \neq A_t^{n,2}}{X_t^{n,1}, X_{t-1}^{1:N}} \geq \frac{\mblow}{\mbhigh} \epsa. \]
\end{lem}
\begin{proof}
	We write (and define new notations along the way):
	\begin{equation*}
		\begin{split}
		\pi(a_t^{n,1}, a_t^{n,2}) &:= p(a_t^{n,1}, a_t^{n,2} | X_t^{n,1}, X_{t-1}^{1:N}) \\
		&\propto p(a_t^{n,1}, a_t^{n,2}|X_{t-1}^{1:N}) m_t(X_{t-1}^{a_t^{n, 1}}, X_t^{n,1}) \\
		&=: p(a_t^{n,1}, a_t^{n,2}|X_{t-1}^{1:N}) \phi(a_t^{n,1})\\
		&=: \pi_0(a_t^{n,1}, a_t^{n,2}) \phi(a_t^{n,1}).
		\end{split}
	\end{equation*}
	Thus
	\begin{equation*}
		\begin{split}
		\CProb{A_t^{n,1} \neq A_t^{n,2}}{X_t^{n,1}, X_{t-1}^{1:N}} &= \int \mathbbm 1\px{a_t^{n, 1}\neq a_t^{n,2}} \pi(a_t\eno, a_t\ent) \\
		&= \frac{\int \mathbbm 1 \px{a_t\eno \neq a_t\ent} \pi_0(a_t\eno, a_t\ent)\phi(a_t\eno)}{\int \pi_0(a_t\eno, a_t\ent) \phi(a_t\eno)} \\
		&\geq \int \mathbbm 1\px{a_t^{n,1}\neq a_t^{n,2}} \pi_0(a_t\eno, a_t\ent) \frac{\mblow}{\mbhigh}
		\end{split}
	\end{equation*}
	by the boundedness of the function $\phi$ between $\mblow$ and $\mbhigh$. From this,  we get the desired result by virtue of Assumption~\ref{asp:coupling:ancestors}.
\end{proof}

\subsection{Validity of Algorithm~\ref{algo:coupling_two_gaussian} (modified Lindvall-Rogers coupler)}
\label{subsect:validity_lindvall_rogers_coupler}
Recall that generating a random variable is equivalent to uniformly simulating under the graph of its density \citep[see e.g.][The Fundamental Theorem of Simulation, chapter 2.3.1]{RobCas}. Algorithm~\ref{algo:coupling_two_gaussian}'s correctness is thus a direct corollary of the following intuitive lemma.
\begin{lem}
	Let $S_A$ and $S_B$ be two subsets of $\mathbb R^d$ with finite Lebesgue measures. Let $A$ and $B$ be two not necessarily independent random variables distributed according to $\operatorname{Uniform}(S_A)$ and $\operatorname{Uniform}(S_B)$ respectively. Denote by $S_0$ the intersection of $S_A$ and $S_B$; and by $C$ a certain $\operatorname{Uniform}(S_A)$-distributed random variable that is independent from $(A,B)$. Define $A^\star$ and $B^\star$ as
	\[A^\star = \begin{cases}
	C &\text{ if } (A, C) \in S_0 \times S_0\\
	A &\text{ otherwise}
	\end{cases}
	 \]
	 and 
	 \[B^\star = \begin{cases}
	 C &\text{ if } (B, C) \in S_0 \times S_0 \\
	 B &\text{ otherwise}
 \end{cases} \]
	 Then $A^\star \sim \operatorname{Uniform}(S_A)$ and $B^\star \sim \operatorname{Uniform}(S_B)$.
\end{lem}
\begin{proof}
	Given $(A,C) \in S_0 \times S_0$, the two variables $A$ and $C$ have the same distribution (which is $\operatorname{Uniform}(S_0)$). Thus, the definition of $A^\star$ implies that $A$ and $A^\star$ have the same (unconditional) distribution. The same argument applies to $B$ and $B^\star$ notwithstanding the asymmetry in the definition of $C$.
\end{proof}

\subsection{Hoeffding inequalities}\label{ap:hoeffding}

This section proves a Hoeffding inequality for ratios, which helps us to bound~\eqref{eq:delta_tn}. It is essentially a reformulation of \citet[Lemma 4]{Douc2011} in a slightly more general manner.

\begin{defi}\label{def:sub_gaussian}
A real-valued random variable $X$ is called $(C, S)$-sub-Gaussian if 
\[\P\pr{\frac{\abs X}{S} > t} \leq 2 C e^{-t^2/2}, \forall\ t \geq 0. \]
\end{defi}
This definition is close to other sub-Gaussian definitions in the literature,
see e.g.\ \citet[Chapter 2.5]{vershynin2018high}. It basically means that the
tails of $X$ decreases at least as fast as the tails of the $\mathcal{N}(0, S^2)$ distribution, which is itself $(1, S)$-sub-Gaussian. The following result is classic.
\begin{thm}[Hoeffding's inequality]\label{thm:hoeffing_apx}
	Let $X_1, \ldots, X_N$ be $N$ i.i.d.\ random variables with mean $\mu$ and almost surely contained between $a$ and $b$. Then $N^{1/2}(\sum X_i/N - \mu)$ is $(1, (b-a)/2)$-sub-Gaussian.
\end{thm}

The following lemma is elementary from Definition~\ref{def:sub_gaussian}. The proof is omitted.

\begin{lem}\label{lem:hoeffding_apx}
Let $X$ and $Y$ be two (not necessarily independent) random variables. If $X$ is $(C_1, S_1)$-sub-Gaussian and $Y$ is $(C_2, S_2)$-sub-Gaussian, then $X+Y$ is $(C_1 + C_2, S_1 + S_2)$-sub-Gaussian.
\end{lem}

We are ready to state the main result of this section.

\begin{prop}[Hoeffding's inequality for ratios]\label{prop:apx:hoeffding}
	Let $a_N$, $b_N$, $a^*$, $b^*$ be random variables such that $\sqrt N(a_N - a^*)$ is $(C_a, S_a)$-sub-Gaussian and $\sqrt N(b_N - b^*)$ is $(C_b, S_b)$-sub-Gaussian.
	Then $\sqrt N \pr{\nfrac{a_N}{b_N} - \nfrac{a^*}{b^*}}$ is sub-Gaussian with parameters $(C^*, S^*)$ where
	\[\begin{cases}
	C^* &= C_a + C_b\\
	S^* &= \norm{\frac{1}{b^*}}_\infty (S_a + S_b \norm{\frac{a_N}{b_N}}_\infty).
	\end{cases} \]
The terms with inf-norm can be infinite if the corresponding random variables are unbounded.
\end{prop}
\begin{proof}
We have
\begin{align*}
\abs{\sqrt N \pr{\frac{a_N}{b_N} - \frac{a^*}{b^*}}} &\leq \abs{\sqrt N \pr{\frac{a_N}{b_N} - \frac{a_N}{b^*}}} + \abs{\sqrt N \pr{\frac{a_N}{b^*} - \frac{a^*}{b^*}}} \\
&= \abs{\frac{a_N}{b_N}} \abs{\frac 1{b^*}} \abs{\sqrt N (b_N - b^*)} + \abs{\frac 1{b^*}} \abs{\sqrt N (a_N - a^*)}
\end{align*}
by which the proposition follows from Lemma~\ref{lem:hoeffding_apx}.
\end{proof}

\end{document}